\documentclass[acmsmall,review=false, anonymous=false,nonacm]{acmart}

\input{datalabmacros}

\begin{document}

\newcommand{\ourTitle}{A Principled Solution to the Disjunction Problem\\ of Diagrammatic Query Representations}

\title[A Principled Solution to the Disjunction Problem of Diagrammatic Query Representations]{A Principled Solution to the Disjunction Problem\\ of Diagrammatic Query Representations}

\author{Wolfgang Gatterbauer}
\orcid{0000-0002-9614-0504}
\affiliation{%
    \orcidicon{0000-0002-9614-0504}
	Northeastern University\country{USA}}
\email{w.gatterbauer@northeastern.edu}

\begin{abstract}
Finding unambiguous diagrammatic representations for first-order logical formulas and relational queries with arbitrarily nested disjunctions has been a surprisingly long-standing unsolved problem.
We refer to this problem as \emph{the disjunction problem} (of diagrammatic query representations).

This work solves the disjunction problem.
Our solution
unifies, generalizes, and overcomes the shortcomings of prior approaches for disjunctions.
It extends the recently proposed \diagrams\xspace and is identical for disjunction-free queries. 
However, it can preserve the relational patterns and the safety
for all well-formed Tuple Relational Calculus (\TRC) queries,
even with arbitrary disjunctions.
Additionally, its size is proportional to the original $\TRC$ query 
and can thus be exponentially more succinct than \diagrams. 
\end{abstract}

\maketitle
\setcounter{page}{1}

\section{Introduction}
\label{sec:introduction}

The goal of query visualization is 
to transform a relational query
into a visual representation that helps a user quickly understand the intent of a query~\cite{Gatterbauer2022PrinciplesQueryVisualization}.
While there is a very long history of visual query languages~\cite{ICDE:2024:diagrammatic:tutorial,DBLP:journals/vlc/CatarciCLB97},
a problem that remains unsolved to this day
is the question of how to ``truthfully'' represent any logical disjunction in a graphical language.
Just like conjunction, disjunction is a fundamental logical operator to combine logical statements,
but it is far harder to represent graphically.
We call this \emph{the disjunction problem} (of visual query representations).

The problem has vexed researchers for centuries, even for basic First-Order 
Logic (FOL).\footnote{FOL is basically the same as Relational Calculus and thus equivalent in expressiveness to 
relationally complete languages.}
Pierce mentions the problem already in 1896 in his influential work on Venn diagrams: 
``It is only disjunctions of conjunctions that cause some inconvenience''~\cite[Paragraph 4.365]{peirce:1933}.
Gardner in his 1958 book `Logic Machines and Diagrams'~\cite{Gardner:1958:logicMachines}
discusses the challenging disjunction 
$(A \vee B) \rightarrow (C \vee D)$
and concludes that
``there seems to be no simple way in which the statement, as it stands, can be diagramed''~\cite[Section 4.3]{Gardner:1958:logicMachines}.
Shin in her work on the logic of diagrams writes 
``any form of representation for disjunctive information--whether a sign is introduced or not--is bound to be symbolic''
\cite[Ch.~3.2]{Shin:2002}.
Englebretsen~\cite{Englebretsen:1996:Shin}
in his review of Shin's 1995 book~\cite{shin_1995} writes:
``In her discussion of perception she shows that disjunctive information is not representable in any system.''
Thalheim
states~\cite{Thalheim:2007:Slides,Thalheim:2013:Slides}
that ``There is no simple way to represent Boolean formulas''
and gives a challenging example that is identical to \cref{Fig_Prior_Disjunctions_6} up to renaming of constants):
$R.A \equal 1 \vee R.B \equal 2 \vee (R.C \equal 3 \wedge R.D \equal 4)$.
Gatterbauer in his tutorial on visual query representations~\cite{ICDE:2024:diagrammatic:tutorial}
lists several open problems 
and includes 
\cref{Fig_Prior_Disjunctions_10} as challenge:
$(R.A \equal S.A \wedge R.B \equal 0) \vee (R.B \equal 1 \wedge R.C \equal 2)$.
A recent paper~\cite{DBLP:journals/pacmmod/GatterbauerD24} states in its conclusions 
that ``it is not clear how to achieve an intuitive and principled diagrammatic representation for arbitrary nestings of disjunctions, such as
$R.A \!<\! S.E \wedge (R.B \!<\! S.F \vee R.C \!<\! S.G)$
or 
$(R.A \!>\! 0 \wedge R.A \!<\! 10) \vee (R.A \!>\! 20 \wedge R.A \!<\! 30)$''.

\begin{figure}[t]
    \centering	
    \begin{subfigure}[b]{.2\linewidth}
        \centering
        \includegraphics[scale=0.4]{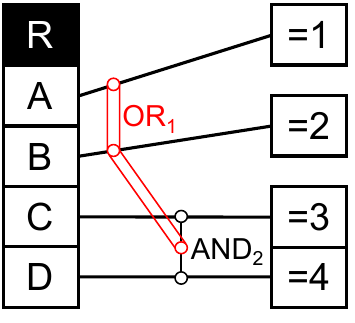}
        \vspace{2mm}
        \caption{\cite{Thalheim:2013:Slides}}
        \label{Fig_Prior_Disjunctions_6}
    \end{subfigure}	
    \hspace{3mm}
    \begin{subfigure}[b]{.2\linewidth}
    \centering        
    \includegraphics[scale=0.4]{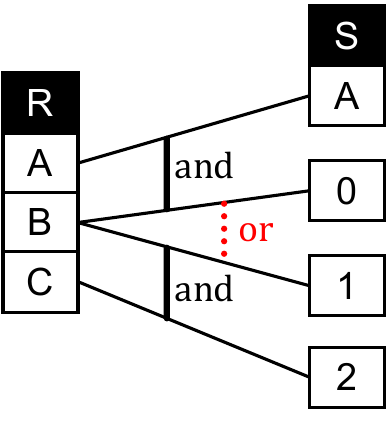}
    \caption{\cite{ICDE:2024:diagrammatic:tutorial}}
    \label{Fig_Prior_Disjunctions_10}
    \end{subfigure}	
    \hspace{3mm}
    \begin{subfigure}[b]{0.2\linewidth}
        \centering
        \includegraphics[scale=0.4]{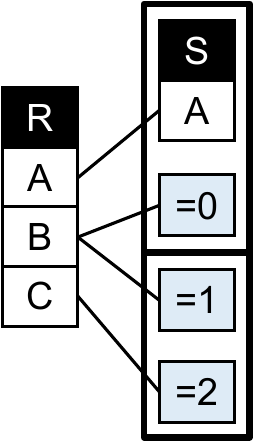}
        \vspace{0mm}    
        \caption{Interpretation \cref{eq:Fig_Disjunction_solutions_1}}
        \label{Fig_Disjunction_solutions_1}
    \end{subfigure}	
    \hspace{3mm}
    \begin{subfigure}[b]{0.2\linewidth}
        \centering
        \includegraphics[scale=0.4]{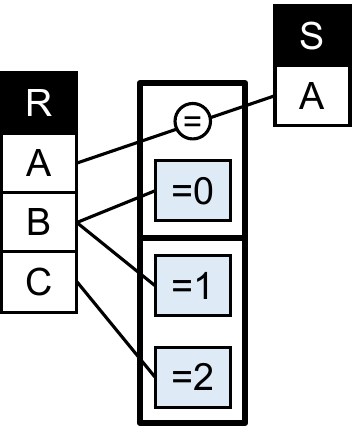}
        \vspace{0mm}    
        \caption{Interpretation \cref{eq:Fig_Disjunction_solutions_2}}
        \label{Fig_Disjunction_solutions_2}
    \end{subfigure}	
    \caption{(a, b): Two previously proposed approaches for representing disjunctions via edges 
    (see text).
    The approaches are incomplete as they leave details of quantification ambiguous and require symbolic annotations 
    to determine precedence of operators.
    (c,d): Our solution $\systemx$
    has precise semantics and can pattern-represent any well-formed $\TRC$ query
    (their interpretation is given in \cref{sec:RDconclusions}). 
    }
    \label{Fig_Introduction}
\end{figure}

\introparagraph{Our contribution}
We give a principled solution to the disjunction problem
that
\emph{unifies, generalizes, and overcomes the shortcomings} of the 3 main prior graphical approaches for disjunction.
Our solution, called $\systemx$,
is a diagrammatic representation 
of well-formed Tuple Relational Calculus (\TRC)
that preserves the relational pattern 
(and a 1-to-1 correspondence with the atoms) 
and the safey of a query.\footnote{Safety is a syntactic criterion that guarantees that the query is domain-independent and thus always returns finitely many answers~\cite{Topor2018}. See \cref{sec:safety} for details.
``Same pattern'' is a semantic notion and means (slightly simplified) that the representation uses the set of relation variables. 
See \cref{sec:relational pattern} for details.}
It is heavily inspired by \diagrams~\cite{DBLP:journals/pacmmod/GatterbauerD24}
which was shown to help users understand queries faster and more accurately than SQL
in a randomized user study.
However, it generalizes \diagrams: it is identical for disjunction-free queries,
yet it is more general and can be exponentially more concise.
It also preserves the safety conditions of $\TRC$
and it is the first to achieve 100\% pattern coverage
on a recently proposed textbook benchmark.

\introparagraph{Our approach in a nutshell}
We proceed in two steps:
\circled{1} We first give a rigorously defined representation that \emph{replaces join and selection predicates 
with built-in relations}
and that uses DeMorgan to replace disjunctions. This approach solves the disjunction problem 
(because it gives ``anchors'' to predicates defined \emph{in arbitrary nestings})
but results in a cluttered representation that loses the safety conditions of \TRC.
\circled{2} We then substitute the built-in relations with prior visual formalisms 
(while keeping the formal semantics of built-in relations)
and add a box-based visual shortcut for disjunction that brings back the safety conditions.
Our definition of boxes allows disjunctions \emph{at any nesting level}, 
while prior box-based approaches restrict disjunctions to be at the root.

\introparagraph{Outline}
\cref{sec:relatedwork} defines our problem and classifies prior approaches for representing disjunctions.
\cref{sec:formal setup} develops our notation for Tuple Relational Calculus (\TRC), its safety conditions, 
and the notion of pattern expressiveness based on an Abstract Syntax Tree (AST) representation of \TRC.
These ASTs are in a 1-to-1 correspondence to our later introduced diagrammatic representations.
\cref{sec:firstsolution} gives our preliminary 
solution to the disjunction problem.
It relies upon \emph{built-in relations} that represent constants and built-in predicates, and a DeMorgan-based transformation of \TRC.
While complete, it is practically unsatisfying due to its visual clutter and the fact that it cannot preserve the safety conditions of \TRC. 
\cref{sec:real solution} replaces the built-in relations 
with prior proposed visual formalisms,
yet keeps our rigorous and principled semantics.
It reintroduces a visual symbol for disjunction called ``\emph{DeMorgan-fuse box}'', which allows us to check safety of a query. This addition leads to \systemx.
\cref{sec:generality}
shows how our fuse boxes unify and generalize prior approaches for disjunction,
presents our solutions to the challenging queries from the introduction,
and justify our perceptual choices.
\cref{sec:pattern coverage} shows 100\% pattern coverage over a test set made available by the authors of \diagrams\xspace.
The optional appendix
includes many details, proofs, and illustrating examples.

\section{Background and Related Work on Diagrammatic Query Representations}
\label{sec:relatedwork}
We discuss diagrammatic (visual) query representations, 
define notions of \emph{a logical diagram} and \emph{relational patterns},
and classify prior approaches for representing disjunctions.

\subsection{Diagrammatic vs.\ Textual Representations}

We use the notion of a \emph{diagrammatic representation} synonymously with one that is
visual, graphical, or non-symbolic (in contrast to \emph{textual} or symbolic), 
and define logical diagrams as follows:

\begin{definition}[Logical Diagram]
    A logical diagram is a graphical representation of a logical formula
    in which the topological relationships between its elements 
    represent logical relationships between the elements of the formula.
\end{definition}

\noindent
Topological relationships are spatial relations that remain invariant under continuous deformations,
such as connectivity, containment, and adjacency. 
Intuitively, in order for a representation 
to be called diagrammatic,
it needs to show joins (i.e.\ the relationships between tables) 
as edges between the respective table attributes,
and it cannot contain non-atomic logical sentences that require symbolic interpretation of logical connectives,
such as ``$A \equal 1 \vee A \equal 2$''.
Our definition
captures the essence of many prior definitions of diagrams. We give examples:
``\emph{Diagram: a simplified drawing showing the appearance, structure, or workings of something; a schematic representation}''~\cite{OxfordLanguages:Diagram}.
``\emph{Diagram: a graphic design that explains rather than represents; especially: a drawing that shows arrangement and relations (as of parts)}''~\cite{MerriamWebster:Diagram}.
``\emph{The relationships established between two sets of elements constitute a diagram}''~\cite[p.~129]{Bertin:1981eu}.
``\emph{Logic diagram: a two-dimensional geometric figure with spatial relations that are isomorphic with the structure of a logical statement}''~\cite[p.~28]{Gardner:1958:logicMachines}.

Notice that while the relationships between elements are captured diagrammatically,
the elements themselves are still represented as text.
This is because relation names and attribute names do not constitute relational information themselves.
For example, the string ``Sailor'' is still used to represent
the name of a relation called ``Sailor'' (instead of an icon with domain-specific interpretation)
and similarly with an attribute named ``name'', 
however the fact that ``Sailor'' is the name of a relation, and the fact that ``name'' is one of its attributes
constitute relationships.
This separation of information carried by individual elements via text from relational information that can be read of diagrammatically
is a key motivation of diagrammatic representations and visualizations in general.
See for example Scott McCloud's influential work on understanding comics~\cite{mccloud:1993:comics}
that shows that using text to convey part of the information 
frees up the image 
to focus on other content, 
and vice versa.
In the case of diagrams, it is the topology that focuses on the relations between elements
rather than the individual elements themselves.
See also, Hearst's recent discussion~\cite{DBLP:journals/cacm/Hearst23} 
of the complex interactions when combining text with visualizations
and many references therein.

\subsubsection{Constitutive vs.\ enabling features}
\label{sec:enabling feature}
De Toffoli~\cite{Toffoli:2022} distinguishes between \emph{constitutive} and \emph{enabling} features of a notational system.
The former have precise mathematical meaning and are essential to interpret the notation correctly
(e.g.\ topological relationships).
The latter facilitate interpretation but are not essential
(e.g.\ colors or relative sizes).
We find this distinction helpful and use it when occasionally
discussing enabling (non-essential) features of our 
diagrammatic representation.
Similar distinctions were made many times throughout history, e.g.\ by Manders~\cite{10.1093/acprof:oso/9780199296453.003.0004} who contrasts
``co-exact'' attributes of a diagram 
(in essence, topological attributes)
from ``exact'' attributes (geometric attributes that are unstable under perturbations, like size and shape).

\subsection{Relational Patterns and the Disjunction Problem}
\label{sec:relational pattern}

Our goal is to find an unambiguous diagrammatic representation of a logical formula
that can preserve ``the structure'' of arbitrary disjunctions.
To properly formalize this problem, we build upon 
a recent definition of \emph{the relational pattern} of a query~\cite{DBLP:journals/pacmmod/GatterbauerD24}. 
Call the signature $\S$ of a query $q$ the set of all its \emph{relation variables} (i.e.\ all the quantified references to some input table).\footnote{The word \emph{relation variables} (and its abbreviation \emph{relvars}) was popularized by Date and Darwen~\cite{Date:2000ta,date2004introduction}
to distinguish between a reference to an input table (e.g., ``$r$'' in the $\TRC$ expression ``$\exists r \in R$'')
and an actual instance of a relation (e.g., a concrete instance of $R$) which they call a \emph{relation value}.
The difference is crucial for queries with repeated references to the same table, e.g.\ conjunctive queries with self-joins.
We also refer to relation variables as \emph{table references}.
}
Then treat each relation variable in the signature $\S$ of a query $q$ 
as a reference to a distinct relation value (i.e.\ a fresh input table)
and call the resulting query the \emph{dissociated query} $q'$ 
over the \emph{dissociated signature} $\S'$.
Finally define the logical function defined by $q'$ as the relational pattern of $q$:

\begin{definition}[Relational pattern~\cite{DBLP:journals/pacmmod/GatterbauerD24}]
	Given a query $q$ with 
 	signature $\S$.
	The \emph{relational pattern} of $q$ is the logical function defined by its
	dissociated query $q'(\S')$.
\end{definition}

\noindent
The intuition behind this formalism is that the dissociated query
defines a function that maps the relation values represented by a set of relation variables (not just a set of input tables) to an output table (an output relation value).
Thus, the dissociated query is a semantic definition
of a relational query pattern that can be applied to any relational query language and that can be used to compare the relative pattern expressiveness of different languages.
Two logically equivalent relational queries $q_1$ and $q_2$ then use the same relational pattern (they are \emph{pattern-isomorphic}) 
if their dissociated queries are also logically equivalent, up to renaming and reordering of their relation variables 
(i.e.\ there is a mapping between the relation variables of their dissociated queries that preserves logical equivalence).
In other words, $q_2$ is a \emph{pattern-isomorphic} representation of $q_1$.

This formalism allows us to give a precise definition of our problem:

\begin{definition}[Disjunction problem]
    The disjunction problem (of diagrammatic query representation) is the problem of finding a pattern-isomorphic diagrammatic representation for any First-Order Logic (FOL) formula, and unambiguous translations back and forth.
\end{definition}

\subsection{The safety problem of DeMorgan-based representations}

An additional challenge for diagrammatic representations is that 
the safety of relational queries is defined via syntactic criteria. 
A simple transformation of a safe formula via a DeMorgen (which alone cannot solve the disjunction problem) 
may lead to an unsafe formula.

\begin{example}[Union of queries]
    \label{ex:disjunction_union}
    Consider two unary tables $R(A)$ and $S(A)$ and the $\TRC$ query:
    \begin{align}
        \{q(A) \mid \exists r\in R[q.A \equal r.A] \;\h{\vee}\; \exists s\in S[q.A \equal s.A]\}      \label{eq:union query}
    \end{align}
    This query expresses the union of the two tables and is safe according to any safety definition we know of 
    (including Ullman's~\cite[Section 3.9]{Ullman1988PrinceplesOfDatabase}, see also \cref{sec:safety}).
    We can remove the disjunction by using DeMorgan,
    however, the resulting query is now unsafe due to the outer $\neg$ operator:
    \begin{align*}
        \{q(A) \mid \h{\neg( 
            \neg(}\exists r\in R[q.A \equal r.A]\h{)} 
            \;\h{\wedge}\; 
            \h{\neg(}\exists s\in S[q.A \equal s.A]\h{))}\}
    \end{align*}
\end{example}

The consequence is that if a diagrammatic representation should convey preserve the safety of a formula, 
then it requires a visual device for disjunction that goes beyond DeMorgan.

\begin{definition}[Safety problem]
    The safety problem (of diagrammatic query representation) 
    is the problem of finding a diagrammatic representation of FOL formulas that preserves their safety.
\end{definition}

\subsection{Existing Visual query representations and diagrammatic reasoning systems}
Visual query languages for writing queries have been investigated since the early days of databases and
a 1997 survey~\cite{DBLP:journals/vlc/CatarciCLB97} has already over 150 references, with examples such as
Query-By-Example (QBE)~\cite{DBLP:journals/ibmsj/Zloof77} and
Query By Diagram (QBD)~\cite{DBLP:journals/vlc/AngelaccioCS90, DBLP:conf/sigmod/CatarciS94}.
Around the same time, Ioannidis~\cite{DBLP:journals/csur/Ioannidis96a} 
lamented that most visual database interfaces were ``ad hoc solutions''
and that ``there are several hard research problems regarding complex querying and visualization that are currently open.''
Today, many commercial and open-source database systems have rudimentary
graphical SQL editors, such as SQL Server Management Studio (SSMS)~\cite{ssms}, Active Query Builder~\cite{activequerybuilder}, QueryScope from SQLdep~\cite{queryscope}, MS Access \cite{msAccess}, and
PostgreSQL's pgAdmin3~\cite{pgadmin}.
Also, new direct manipulation visual interfaces are being developed, such as 
DataPlay~\cite{DBLP:conf/uist/AbouziedHS12,DBLP:journals/pvldb/AbouziedHS12}
and
SIEUFERD~\cite{DBLP:conf/sigmod/BakkeK16}.
More recently, visual query representations have been proposed for the inverse functionality of 
understanding queries, with notable examples
VisualSQL~\cite{DBLP:conf/er/JaakkolaT03},
QueryVis~\cite{DanaparamitaG2011:QueryViz,Leventidis2020QueryVis}, and
SQLVis~\cite{DBLP:conf/vl/MiedemaF21}.

Seemingly disconnected from these developments, the
diagrammatic reasoning community~\cite{10.5555/546459,DBLP:conf/iccs/Howse08,DBLP:conf/icfca/Dau09}
studies diagrammatic representations that have sound and complete inference rules.
Most noteworthy is Shin's influential work~\cite{shin_1995} that proves that a slight modification of Venn-Peirce diagrams (see~\cref{sec:edge-based})
constitutes a sound and complete diagrammatic reasoning system for monadic FOL.
Many variants of diagrammatic reasoning systems have since been proposed
at the annual Diagrams conference~\cite{DBLP:conf/diagrams/2024}.
However neither of these proposals can represent general polyadic predicates (the maximum are dyadic relations~\cite[Table~1]{stapleton:2013:potential-of-diagrams}),
many of them are either not sound or not complete (\cite[Table~2]{stapleton:2013:potential-of-diagrams}),
and neither of them allow pattern-isomorphic representations, even for the fragments of logic they cover
(most often, they can't handle arbitrary disjunctions).

Recent VLDB and ICDE tutorials~\cite{DBLP:journals/pvldb/Gatterbauer23,ICDE:2024:diagrammatic:tutorial}
surveyed diagrammatic representations within and outside the database community
and listed diagrammatic representations of disjunction as open problem, 
which motivated our work.

\subsubsection{\diagrams}

$\diagrams$~\cite{DBLP:journals/pacmmod/GatterbauerD24} 
are a relationally complete and unambiguous diagrammatic representation of safe Tuple Relational Calculus (\TRC).
It uses UML notation for tables and their attributes and
represents negation scopes with hierarchically nested dashed rounded rectangles
that partition the canvas into zones (compartments).
Join predicates are shown with directed arrows and labels on the edges 
(an important detail for us later).
The authors validate their design in a randomized user study
showing that users understand queries faster and more accurately with \diagrams\xspace 
than SQL.
However, that representation (like all prior diagrammatic representations we know of) cannot faithfully represent relational patterns involving disjunctions.
Disjunctions require them to duplicate binding atoms (i.e.\ add fresh relation variables) and to thus change the relational pattern 
(see e.g.\ \cref{ex:Fig_disjunction_bigger_AST}).

We draw a lot of inspiration from that work,
yet develop a diagrammatic representation called $\systemx$that can represent 
\emph{all relational patterns of $\TRC$}
(i.e.\ it is \emph{pattern-complete} for \TRC).
In addition, our solution is backward compatible with $\diagrams$ 
(every \diagram has an identical representation as $\systemx$, but not vice versa).
Interestingly, we achieve this generalization 
by mostly redefining existing visual notations and giving them a stricter semantic interpretation. 
Our representation does not only preserve the table signature (the relation variables), 
but also the join and selection predicates, and thus all atoms from a given $\TRC$ query.

\subsection{The challenge of disjunctions and prior (incomplete) approaches}
\label{sec:intro_disjunctions}

We summarize here 5 conceptual approaches for representing disjunctions that we found throughout the literature.
We illustrate with query $\exists r \in R [r.A \equal 1 \;\h{\vee}\; r.A \equal 2]$ 
and use a standardized and often simplified notation that focuses on the key visual elements.
For the interested reader, 
\cref{app:Original_Drawings}
contains the original figures that led us to this classification.

\subsubsection{Text-based disjunctions}
Logical formulas can always be represented as text~(\cref{Fig_Prior_Disjunctions_1}).
This approach is not diagrammatic and we list it only for completeness.
Example uses are the condition boxes by QBE~\cite{DBLP:journals/ibmsj/Zloof77}
and handling of simple disjunctions by Dataplay~\cite{DBLP:journals/pvldb/AbouziedHS12}.

\subsubsection{(Vertical) form-based disjunctions}
QBE~\cite{DBLP:journals/ibmsj/Zloof77} allows 
filling out two separate rows with alternative information~(\cref{Fig_Prior_Disjunctions_2}).
Recent visual query representations such as SQLVis~\cite{DBLP:conf/vl/MiedemaF21} adopt this approach for simple disjunctions, such as our running example.
However, this formalism does not allow disjunctions between join predicates and selection predicates (e.g., 
$\exists r \in R [r.A \equal 1 \;\h{\vee}\; \exists s \in S [r.A \equal s.A]]$) or nested disjunctions.
Datalog~\cite{DBLP:journals/tkde/CeriGT89} expresses disjunctions with repeated rules, 
and each rule with a new relation variable, and thus not preserving the pattern:
$Q \datarule R(1). Q \datarule R(2)$.

\subsubsection{Edge-based disjunctions}
\label{sec:edge-based}
Around 1896, Charles Sanders Peirce~\cite{peirce:1933} extended Venn diagrams~\cite{venn:1880}
with edges between "O" and "X" markers to express disjunctions. An "O" means the partition is empty (false).
An "X" means there is at least one member (true). A line between two markers means that at least one of these statements is true (\cref{Fig_Prior_Disjunctions_9}).
Connecting disjunctive predicates via lines of various forms 
was suggested repeatedly,
e.g.\ by VQL~\cite{DBLP:journals/tkde/MohanK93},
VisualSQL~\cite{DBLP:conf/er/JaakkolaT03} (\cref{Fig_Prior_Disjunctions_5}),
and QueryViz~\cite{DanaparamitaG2011:QueryViz} (\cref{Fig_Prior_Disjunctions_4}).
Edges were mostly used for disjunctive filters within the same table and cannot represent
complicated more formulas, 
such as 
\cref{eq:Fig_Disjunction_solutions_1}
and
\cref{eq:Fig_Disjunction_solutions_2} discussed in \cref{sec:RDconclusions}.

\begin{figure}[t]
\centering	
\begin{subfigure}[b]{.13\linewidth}
\centering
\includegraphics[scale=0.4]{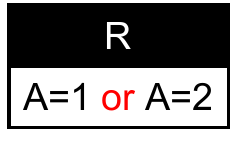}
\caption{}
\label{Fig_Prior_Disjunctions_1}
\end{subfigure}	
\hspace{1mm}
\begin{subfigure}[b]{.05\linewidth}
\centering
\includegraphics[scale=0.4]{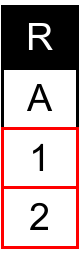}
\caption{}
\label{Fig_Prior_Disjunctions_2}
\end{subfigure}	
\hspace{1mm}
\begin{subfigure}[b]{.13\linewidth}
    \centering    
\includegraphics[scale=0.4]{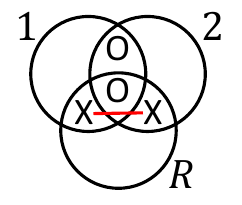}
\vspace{1.5mm}
\caption{}
\label{Fig_Prior_Disjunctions_9}
\end{subfigure}	
\hspace{0mm}
\begin{subfigure}[b]{.15\linewidth}
    \centering
\includegraphics[scale=0.4]{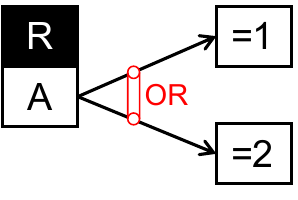}
\caption{}
\label{Fig_Prior_Disjunctions_5}
\end{subfigure}	
\hspace{1mm}
\begin{subfigure}[b]{.15\linewidth}
    \centering
\includegraphics[scale=0.4]{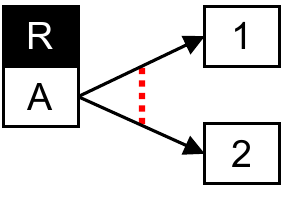}
\caption{}
\label{Fig_Prior_Disjunctions_4}
\end{subfigure}	
\hspace{2mm}
\begin{subfigure}[b]{.26\linewidth}
    \centering
    \includegraphics[scale=0.4]{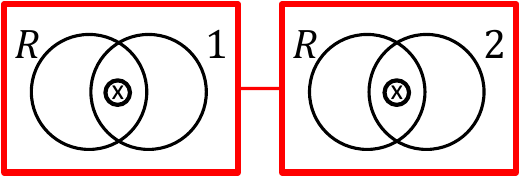}
    \vspace{1mm}
    \caption{}
    \label{Fig_Prior_Disjunctions_8}
\end{subfigure}	
\hspace{1mm}
\begin{subfigure}[b]{.27\linewidth}
    \centering
    \includegraphics[scale=0.4]{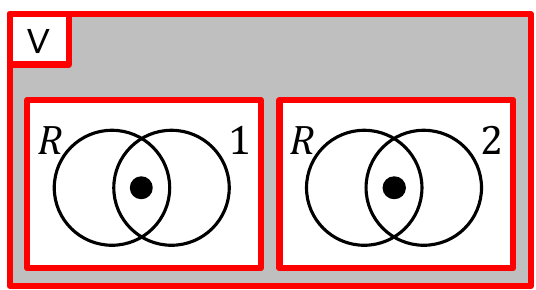}
    \caption{}
    \label{Fig_Prior_Disjunctions_13}
\end{subfigure}	
\hspace{2mm}
\begin{subfigure}[b]{.14\linewidth}
\centering    
\includegraphics[scale=0.4]{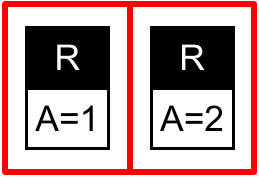}
\caption{}
\label{Fig_Prior_Disjunctions_7}
\end{subfigure}	
\hspace{2mm}
\begin{subfigure}[b]{.22\linewidth}
    \centering        
\includegraphics[scale=0.3]{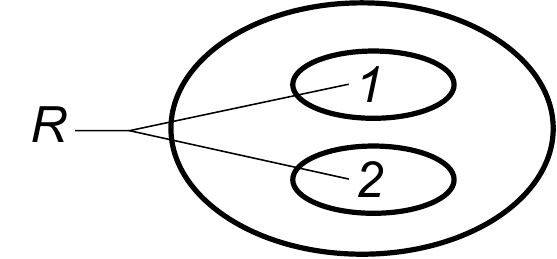}
\caption{}
\label{Fig_Prior_Disjunctions_11}
\end{subfigure}	
\hspace{1mm}
\begin{subfigure}[b]{.14\linewidth}
    \centering    
\includegraphics[scale=0.35]{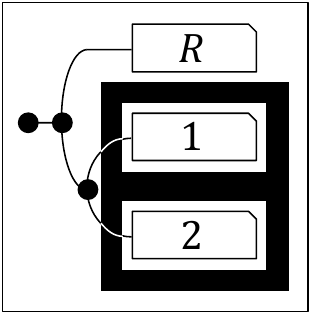}
\caption{}
\label{Fig_Prior_Disjunctions_12}
\end{subfigure}	
\caption{\Cref{sec:intro_disjunctions}: This summary shows 5 conceptual approaches 
for representing disjunctions applied to the deceptively simple problem of representing $R.A \equal 1 \vee R.A \equal 2$:
text-based (a),
form-based (b),
edge-based (c-e),
box-based (f-h),
and DeMorgan-based (i,j).}
\label{Fig_Prior_Disjunctions}
\end{figure}

\subsubsection{Box-based disjunctions}
Peirce
proposed another
solution to disjunctions~\cite{peirce:1933}:
He put unitary Venn diagrams into rectangular boxes and interpreted adjacent boxes as alternatives, i.e.\ disjuncts.
Shin~\cite{shin_1995} adds back lines between boxes 
(\cref{Fig_Prior_Disjunctions_8}).\footnote{We slightly simplified here Shin's proposal. The conclusion is the same, and the appendix gives the full details.}
Spider diagrams~\cite{HowseST2005:SpiderDiagrams} remove the lines between boxes
and place them in a larger ``box template'' with explicit $\vee$ labels
(\cref{Fig_Prior_Disjunctions_13}).
\diagrams~\cite{DBLP:journals/pacmmod/GatterbauerD24} represent a union of queries 
via adjacent ``union cells'' (\cref{Fig_Prior_Disjunctions_7}).
All of these prior box-based approaches represent disjunctions as unions of well-formed diagrams.
Ours is the first to 
allow 
and give a well-defined semantics to
disjunction of logical expressions \emph{within diagrams}.

\subsubsection{DeMorgan-based disjunctions}
We use the term for representations that use only symbols for negation and conjunction, 
and apply negation in a nested way in accordance with the logical identity
$A \vee B = \neg(\neg A \wedge \neg B)$.
Peirce's beta existential graphs~\cite{peirce:1933}
use closed curves to express negation and juxtaposition for conjunctions
(\cref{Fig_Prior_Disjunctions_11}).
String diagrams~\cite{DBLP:conf/diagrams/Haydon2020:StringDiagrams,StringDiagrams:2024:arxiv}
are a variant that represent bound variables by a dot at the end of lines 
(\cref{Fig_Prior_Disjunctions_12}).
We will show that DeMorgan-based disjunctions alone cannot recover the standard safety conditions of relational calculus.

\section{Formal setup on Tuple Relational Calculus (\TRC)}
\label{sec:formal setup}

We give a succinct and necessary background on $\TRC$.
\Cref{app:TRC details}
provides a far more detailed discussion of how our formalism for $\TRC$
and the formal safety conditions compare to prior work.

\subsection{Tuple Relational Calculus ($\TRC$)}
\label{sec:TRC}
Our key claim is that our representation has a natural and pattern-preserving mapping from \emph{any} safe $\TRC$ query.
To prove this claim, we are extra careful in defining safe $\TRC$.
Our formalism is heavily inspired by the sections on $\RC$ of several textbooks
\cite{DBLP:books/aw/AbiteboulHV95,
DBLP:books/mg/SKS20,
Ullman1988PrinceplesOfDatabase,
cowbook:2002,
Elmasri:dq,
10.5555/1882723.1882742},
yet our notation is somewhat streamlined.
\Cref{ex:safety} gives an end-to-end example using our notations.

\introparagraph{Well-formed $\TRC$ formulae}
Given a relational vocabulary $\vec R = \{R_1, R_2, \ldots \}$.
A $\TRC$ formula can contain 3 types of \emph{atoms}:
\begin{enumerate}
    \item \emph{Binding atoms} 
    ``$r \in R$'', 
    where $r$ is a tuple variable and $R \in \vec R$ is a relation.

    \item \emph{Join predicates}
    ``$r.A \,\theta\, s.B$'', 
    where $r$ and $s$ are tuple variables, 
    $A$ and $B$ are attributes from $r$ and $s$ respectively, 
    and $\theta$ is an arithmetic comparison operator from $\{<, \leq, =, \neq , >, \geq \}$.

    \item \emph{Selection predicates}
    ``$r.A \,\theta\, c$'', 
    where $r$, $A$, and $\theta$ are defined as before, 
    and $c$ is a constant.
\end{enumerate}

\noindent
A (tuple) variable is said to be free unless it is bound by a quantifier ($\exists, \forall$).
(Well-formed) formulas are built up inductively 
from atoms with the following 3 families of \emph{formation rules}:
\begin{enumerate}
    \item
    An atom is a formula.

    \item
    If $\varphi_1, \varphi_2, \ldots ,\varphi_k$ are formulas, 
    then so are 
    $\neg(\varphi_1)$,
    $(\varphi_1 \rightarrow \varphi_2)$, 
    $(\varphi_1 \vee \varphi_2 \vee \ldots \vee \varphi_k)$, and
    $(\varphi_1 \wedge \varphi_2 \wedge \ldots \wedge \varphi_k)$.

    \item
    If $\varphi$ is a formula 
    and $R_1, \ldots, R_k$ are relations from $\vec R$, 
    then
    $(\exists r_1 \in R_1, \ldots, \exists r_k \in R_k[\varphi])$ 
    and
    $(\forall r_1 \in R_1, \ldots, \forall r_k \in R_k[\varphi])$ 
    are also formulas.
\end{enumerate}

\noindent
We assume the usual operator precedence 
$(\neg) > (\wedge) > (\vee) > {(\rightarrow)} > (\exists, \forall)$
and can omit parentheses if this causes no ambiguity about the semantics of the formula.
WLOG, no variable can be bound more than once, and no variable occurs both free and bound.

\introparagraph{Well-formed $\TRC$ queries}
A Boolean query (or sentence) is a formula without free variables.
A non-Boolean query 
is an expression
$\{q(\vec H) \mid \varphi \}$
where $q$ is the only free variable of formula $\varphi$, 
and the \emph{header} $\vec H = (A_1, \ldots, A_k)$ is its schema,
i.e.\ the list of attributes (or components) of $q$ that appear in predicates of $\varphi$.
The set builder notation defines the set of tuples that makes 
$\varphi$ true.

\introparagraph{Abstract Syntax Tree (AST)}
The Abstract Syntax Tree (AST) of a $\TRC$ query is a tree-based representation 
that encodes a unique logical decomposition into subexpressions.
It
abstracts away certain syntactic details from the parse tree and
gives a unique reversal of the inductively applied formation rules.
\emph{Atoms} form the leaves (inputs). 
Inner nodes (gates)
belong to 3 families:
\begin{enumerate}
    \item
    The root for non-Boolean queries
    $\{q(\vec H) \mid \varphi \}$
    is formed by a \emph{query} node.
    Its two children are \emph{output} for the output relation $q(\vec H)$ and \emph{formula} for $\varphi$.
    The root of Boolean queries is formed by a \emph{formula} node.

    \item 
    $\exists r_1 \in R_1, \ldots \exists r_k \in R_k [\varphi]$ 
    and
    $\forall r_1 \in R_1, \ldots \forall r_k \in R_k [\varphi]$
    are represented by \emph{$\exists$} and \emph{$\forall$} quantifier nodes, respectively.
    Their two children are \emph{bindings} (with $k$ binding atoms as grandchildren), and \emph{formula} $\varphi$.

    \item 
    \emph{Logical connectives} are nodes that have either one child ($\neg$), 
    two children ($\rightarrow$), 
    or $k \!\geq\! 2$ children ($\wedge, \vee$).
\end{enumerate}

\noindent
We require that no $\neg$ is the child of another $\neg$ node
(we can always cancel double negations by $\neg \neg \varphi = \varphi$)
and that the polyadic connectives $(\wedge, \vee)$ to be ``flattened''~\cite[Sect.~5.4]{DBLP:books/aw/AbiteboulHV95}, 
i.e.\ they can have more than 2 children, yet no child of an $\wedge$ is another $\wedge$ (analogously for $\vee$).
Similarly, quantifier nodes 
$(\exists, \forall)$
can't have a quantifier node of the same type as a grandchild,
i.e.\ a $\exists$ quantifier node can't have another $\exists$ quantifier node as child of its formula child.

\introparagraph{Maximally scoped $\TRC$}
We call $\TRC$ formula \emph{maximally scoped}
if no
$\exists$ quantifier node is the child of an $\wedge$ node.
This is WLOG, as existential quantifiers can always be pushed before an $\wedge$ node, as in:
$\exists r \in R[\varphi_1]  \wedge  \exists s \in S[\varphi_2]
=
\exists r \in R, \exists s \in S[\varphi_1  \wedge  \varphi_2]$.

\subsection{Safety of $\TRC$}
\label{sec:safety}
The idea of \emph{safety} is to \emph{syntactically restrict} the well-formed $\TRC$ queries 
s.t.\ ($i$) safe queries are guaranteed to be domain-independent
(and thus have only finitely many answers), 
yet ($ii$) this subset can express all possible finite queries~\cite{Topor2018}.
In the following, we refer to the subtree of the AST in which each node is connected to the root via a path that is not blocked 
by any negation ($\neg$), implication ($\rightarrow$), nor ($\forall$ quantifier)
as the \emph{base partition} of the AST.
Boolean queries (closed formulas) are by definition domain-independent and thus safe.
We say that a non-Boolean $\TRC$ query $\{q(\vec H) \mid \varphi\}$ is safe if 
it is well-formed
and the following 4 conditions hold on $\varphi$:

\begin{enumerate}
    \item 
    Every attribute $A$ of the header $\vec H$ 
    is bound in $\varphi$ to either 
    ($i$)~an attribute $B$ of an existentially quantified table $\exists r \in R$ via an \emph{equijoin predicate} $q.A \equal r.B$, or 
    ($ii$) a constant $c$ via an \emph{equi-selection predicate} $q.A \equal c$.
    In both cases, we refer to the equality predicate as the \emph{binding predicate} of $q.A$.
    
    \item  
    Every binding predicate is in the base partition of the AST.

    \item  
    If there is an $\vee$ operator on the unique path
    from a binding predicate to the root node of the AST 
    then all child subformulas of that $\vee$ node have 
    only one free tuple variable, and it is the same variable with the same attributes defined.

    \item  
    If there is no $\vee$ operator on that path from a binding attribute involving a header attribute $q.A$
    then the 
    header
    attribute $q.A$ appears in no other predicate.
\end{enumerate}

Call \emph{base disjunction} any disjunction that appears in the base partition.
Intuitively, these conditions ensure that safe $\TRC$ queries have an AST 
in which each output attribute is bound to exactly one existentially quantified table column (or a constant)
in the base partition if, for all base disjunctions, all except for one subtree is removed.
\Cref{Fig_my_safety_example_AST} illustrates with a purposefully involved example.

\begin{figure}[t]
    \begin{subfigure}[t]{.55\linewidth}
        \begin{example}	
            \label{ex:safety}	
            Consider the following safe non-Boolean $\TRC$ query:
            \begin{align*}
            &\{ q(A,B) \mid 
                (\hblue{q.A \equal 0} \wedge (\exists r\in R [\hblue{q.B \equal r.B}] \,\horange{\vee}\, q.B \equal 1))  \\
            &\hspace{5mm}            
                \horange{\vee} \, (\exists r \in R [\hblue{q.A \equal r.A} \wedge \hblue{q.B \equal r.B} \\
            &\hspace{7mm}
                 \wedge \horange{\forall} s\in S[\exists r_2 \in R[r.A \equal r_2.A \wedge r_2.B \equal s.B ]]]) \}
            \end{align*}
            The figure to the right
            shows its unique AST. 
            Notice how the 4 safety conditions are fulfilled.
            In particular, the two child subformulas 
            ``\,$\exists r\in R [\hblue{q.B \equal r.B}]$\!''
            and
            ``\,$q.B \equal 1$\!''
            of the lower nested disjunction have both 
            $q(B)$ as free variable.
            However, the child subformulas of the higher disjunction have both $q(A,B)$ as free variable.
            \end{example}
\end{subfigure}
\hspace{3mm}
\begin{subfigure}[t]{.4\linewidth}
    \centering
    \vspace{2mm}
    \includegraphics[scale=0.35]{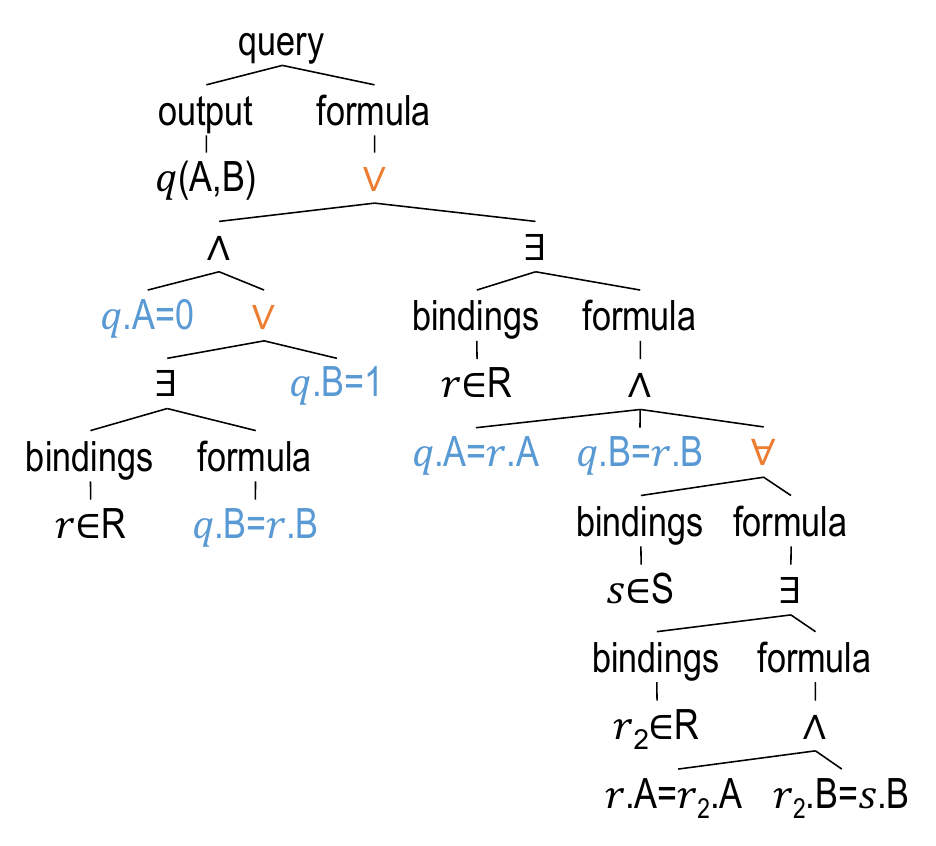}
\end{subfigure}
\caption{\Cref{ex:safety}: AST for $\TRC$ query with nested disjunctions.}
\label{Fig_my_safety_example_AST}   
\end{figure}

\section{A preliminary solution with built-in relations}
\label{sec:firstsolution}

This section develops an extension to $\diagrams$ that solves the disjunction problem and 
gives the resulting representation 
the same pattern expressiveness as $\TRC$.
The approach relies on ``built-in relations'' 
(i.e.\ unary and binary relations that are not part of the input schema and represent constant and built-in predicates). 
The approach is 
simple in that it does not require any novel visual syntactic devices 
(it uses even a smaller visual vocabulary than \diagrams).
However, it is practically unsatisfying due to its additional visual clutter, 
and the fact that it \emph{cannot preserve the safety conditions} of $\TRC$.
We address these problems in the subsequent section.

Here, we first discuss two important fragments of $\TRC$ (\cref{sec:ENC-TRC}),
then discuss our idea of built-in relations (\cref{sec:anchors}),
before we 
prove it to be pattern-isomorphic to $\TRC$
(\cref{sec:Theorem-built-in-relations}).

\subsection{$\TRC^{\neg \exists \wedge}$ is an atom-preserving fragment of $\TRC$ (but not safe $\TRC$)}
\label{sec:ENC-TRC}
We first show that 
universal quantifiers $\forall$,
material implications $\rightarrow$, and
disjunctions can be replaced in $\TRC$ by using only the symbols for 
negation $\neg(...)$, 
existential quantification $\exists$,
and conjunction $\wedge$
\emph{without changing the relational pattern}.\footnote{This statement is not obvious. 
Take as an example the single connective NOR ($\downarrow$) which
is also truth-functionally complete~\cite[Sections 7.4]{Barker-Plummer:2011}.
It is easy to see that NOR is \emph{not pattern preserving}: 
$\neg(\varphi) = \varphi \downarrow \varphi$.}
While it is standard textbook knowledge that the connectives $\neg, \wedge$ are truth-functionally complete~\cite[Sections 7.4]{Barker-Plummer:2011}, 
we show the slightly more general statement that we can preserve \emph{all atoms} in the AST.

\begin{lemma}
    \label{lemma:ENC-TRC}
    Given a $\TRC$ formula $\varphi$ with 
    universal quantification, implication, or disjunction. 
    Then there exists a logically equivalent $\TRC$ formula $\varphi'$ that
    ($i$) is pattern equivalent to $\varphi$, 
    ($ii$) does not use universal quantifiers, implications, or disjunction,
    ($iii$) uses the identical set of atoms, and
    ($iv$) can be found in a polynomial number of steps in the size of $\varphi$.
\end{lemma}

We call 
``Existential-Negation-Conjunctive $\TRC$''
($\TRC^{\neg \exists \wedge}$)
the resulting syntactic fragment 
of $\TRC$
that does not use the connectives $\vee$ and $\rightarrow$, nor the universal quantifier $\forall$.\footnote{We have consulted a long list of standard textbooks on logic
\cite{Barker-Plummer:2011,GeneserethK:2018,enderton2001mathematical},
online resources, and LLMs, yet have not found a standardized, shorter, non-ambiguous terminology for this rather natural fragment, despite being often implicitly used.}
We call 
``Existential-Negation $\TRC$'' 
($\TRC^{\neg \exists \wedge \vee}$)
the variant of that syntactic fragment that allows disjunctions.

\subsubsection{Comparison with the non-disjunctive fragment}
\label{sec:non-dijunctive-fragment}
The \emph{non-disjunctive fragment} defined by the authors of \diagrams~\cite{DBLP:journals/pacmmod/GatterbauerD24},
includes an extra condition requiring all predicates to be ``guarded'' 
(each predicate needs to contain a ``local attribute'' whose relation is quantified within the scope of the last negation).
This condition leads to a reduction in logical expressiveness, which the authors fixed by adding a union operator and corresponding new visual element. 
It also leads to cases where expressing a query requires a \emph{different relational signature}.
This is in contrast to $\TRC^{\neg \exists \wedge \vee}$ and $\TRC^{\neg \exists \wedge}$ which are only syntactic restrictions 
of non-leaf nodes of the AST.
\Cref{ex:Fig_disjunction_bigger_AST} in
\Cref{sec:non-dijunctive-fragment-illustration} 
contrasts example ASTs for 
$\TRC^{\neg \exists \wedge \vee}$ and $\TRC^{\neg \exists \wedge}$.

\subsubsection{Safety is preserved by $\TRC^{\neg \exists \wedge \vee}$, but not by $\TRC^{\neg \exists \wedge}$}
Recall that safety is a syntactic criterion, 
and applying DeMorgan can render a safe query unsafe and v.v.\ 
(recall \cref{ex:disjunction_union}).
Thus removing disjunction from the vocabulary makes it impossible to represent all logical queries \emph{while preserving safety}.
It is easy to see that the safe query from \cref{ex:disjunction_union}
cannot be expressed in $\TRC^{\neg \exists \wedge}$ 
while maintaining safety:
The output needs to be restricted to the union of $R(A)$ and $S(A)$.
In the absence of disunction $\vee$ that can only be achieved with DeMorgan, which renders the resulting query unsafe
since the output can't be bound to any table column.
In contrast, 
$\TRC^{\neg \exists \wedge \vee}$ \emph{preserves safety} 
since all transformations for removing $\forall$ and $\rightarrow$
from a safe $\TRC$ query must happen \emph{outside the base partition} of its AST, 
and thus no binding predicate changes during the transformation.

\subsection{Built-in relations reduce the visual vocabulary but extend the pattern expressiveness of $\diagrams$}
\label{sec:anchors}

We next add unary and binary \emph{built-in relations} to the vocabulary of $\diagrams$ 
and show that this addition enables the resulting visual representation to represent all patterns of $\TRC$.
The intuition is that these additional tables can serve as ``anchors'' for negation scopes and thus permit a direct translation from $\TRC^{\neg \exists \wedge}$ to such extended \diagrams.\footnote{We call those relations ``built-in'' in reminiscences of ``built-in predicates'' like $<$ in SQL~\cite{Ullman1988PrinceplesOfDatabase}.
An alternative name we had considered is ``anchor relations'' as they give us anchors for negation and disjunction scopes.}
Furthermore, by expressing predicates as relations, we \emph{do not have to introduce any new visual elements}.
However, our representation loses the ability to encode safety (and becomes visually more complex).
We fix both issues in the next section.

\subsubsection{Constants represented as unary built-in relations}
\label{sec:built-in:constants}
For each constant $c$ and each arithmetic comparison operator $\theta \in \{=, <, \leq, >, \geq, \neq\}$,
we allow a new unary relation
\fcolorbox{black}{unaryBuiltInPredicate}{\textcolor{white}{$\theta c$}}
descriptively named ``$\theta c$''
that contains the subset of the (possibly infinite) domain that fulfills that condition. 
Each selection predicate ``$R.A \,\theta\, c$'' is then represented as equijoins with a different occurrence of that relation.\footnote{When clear from the context, we write the table name instead of a table variable.}
WLOG, we adopt Ullman's notation~\cite{Ullman1988PrinceplesOfDatabase} for the ordered, unnamed perspective
and name the column $\$1$.
\Cref{Fig_built-in_anquors_1} shows the representation for the selection predicate ``$R.A \!<\! 4$''.
Notice that our choice of blue background color for unary built-in relations is only ``enabling''
(\cref{sec:enabling feature}).

\subsubsection{Join predicates represented as binary built-in relations}
\label{sec:built-in:joins}
For each arithmetic comparison operator $\theta \in \{=, <, \leq, >, \geq, \neq\}$,
we allow a new binary relation 
\fcolorbox{black}{binaryBuiltInPredicate}{\textcolor{white}{$\theta$}}
descriptively named ``$\theta$''
that contains the subset of the (possibly infinite) cross-product of the domain that fulfills that arithmetic comparison. 
Each join predicate ``$R.A \,\theta\, S.B$'' is then represented as two equijoins of R and S with an instance of such a relation.
We again adopt Ullman's notation and name the columns $\$1$ and $\$2$, respectively.
\Cref{Fig_built-in_anquors_2} shows the representation for the join predicate ``$R.A \!<\! S.B$''.
Notice that our choice of orange background color for binary built-in relations is again only ``enabling''.

\subsubsection{Correct placement of built-in relations}
\label{sec:correct-placement}
The placement of the labels for built-in predicates by $\diagrams$ is not important since they are interpreted as ``labels'' on the edges, and the correct interpretation is guaranteed by the ``guard'' of each predicate (i.e. the inner-most nested relations).
This freedom gets restricted with built-in relations as we illustrate next.

\begin{example}
    \label{ex:builtin relations}
Consider the two upper $\diagrams$ in \cref{Fig_built-in_anquors_1,Fig_built-in_anquors_2}.
According to \cite{DBLP:journals/pacmmod/GatterbauerD24},
their interpretation is identical since the placement of the label \builtinlabel{$<$} 
does not affect its interpretation.
Both are interpreted as
``There exists a value in $R.A$ s.t.\ there is no value in $S.B$ that is bigger'', i.e.\
\begin{align}
\phantom{=} &
    \exists r \in R [\neg (
        \exists s \in S [\h{r.A \!<\! s.B}])] 
        \label{eq:biultin relations base}
\end{align}

\noindent
Once we replace the built-in predicate 
$r.A \!<\! s.B$
with a built-in relation 
\fcolorbox{black}{binaryBuiltInPredicate}{\textcolor{white}{$<$}},
the placement of that relation matters,
as the two possible placements (bottom in \cref{Fig_built-in_anquors_1,Fig_built-in_anquors_2}) are interpreted differently:
\begin{align}
\phantom{=} &
\exists r \in R [\neg (
    \exists s \in S, \h{\exists \anchor{j} \in \anchor{\textrm{``$<$''}}} [ 
    \h{r.A \equal \anchor{j}.\$1 \wedge \anchor{j}.\$2 \equal s.B}]  )]     \label{eq:builtintranslation1}\\
\phantom{=} &
\exists r \in R, \h{\exists \anchor{j} \in \anchor{\textrm{``$<$''}}} [\neg (
    \exists s \in S [ 
    \h{r.A \equal \anchor{j}.\$1 \wedge \anchor{j}.\$2 \equal s.B}]  )]     \label{eq:builtintranslation2}
\end{align}  

\noindent
Query \cref{eq:builtintranslation1} is identical to \cref{eq:biultin relations base},
but query \cref{eq:builtintranslation2} states something far more permissive:
"There exists a value in $R.A$ s.t.\ there exists a bigger value that is not in $S.B$".
For example, assume the database is $R(1)$ and $S(2)$.
Then variant \cref{eq:builtintranslation1} is false (as expected), 
whereas  variant \cref{eq:builtintranslation2} is true for the assignment:
$R.A \equal \textrm{``$<$''}.\$1 \equal 1$ and 
$\textrm{``$<$''}.\$2 \equal 3$ 
(since the value $3$ does not exist in $S.B$).
\end{example}

\begin{figure}[t]
    \centering
    \begin{subfigure}[b]{0.12\linewidth}
        \centering
        \includegraphics[scale=0.4]{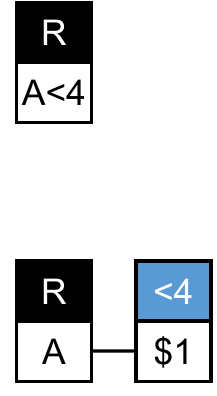}
        \vspace{1mm}
        \caption{}
        \label{Fig_built-in_anquors_1}
    \end{subfigure}	
    \hspace{1mm}
    \begin{subfigure}[b]{0.21\linewidth}
        \centering
        \includegraphics[scale=0.4]{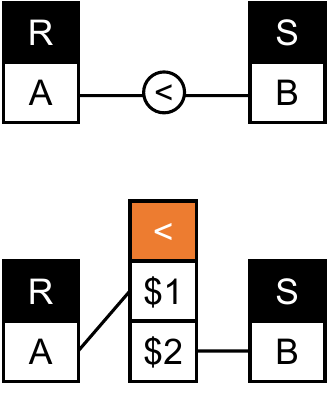}
        \vspace{1mm}
        \caption{}
        \label{Fig_built-in_anquors_2}
    \end{subfigure}	    
    \hspace{10mm}
    \begin{subfigure}[b]{0.2\linewidth}
        \centering
        \includegraphics[scale=0.4]{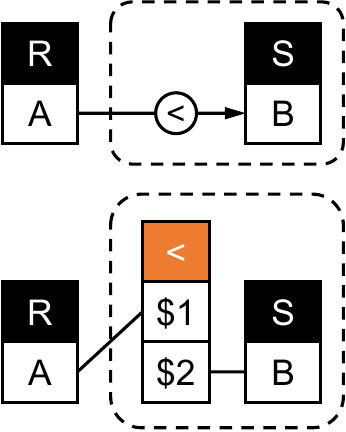}
        \caption{}
        \label{Fig_new_anquors_3}       %
    \end{subfigure}	
    \hspace{1mm}
    \begin{subfigure}[b]{0.2\linewidth}
        \centering
        \includegraphics[scale=0.4]{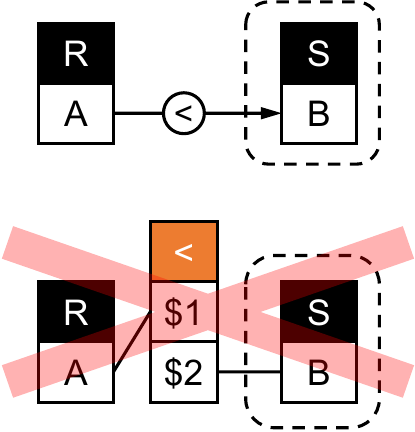}
        \caption{}
        \label{Fig_new_anquors_4}       %
    \end{subfigure}	
    \caption{
        \cref{sec:built-in:constants,sec:built-in:joins}: (a) Unary (blue) and (b) binary (orange) ``anchor relations'' are added to the visual vocabulary of $\diagrams$ in order to give the resulting diagrammatic representation system the same pattern expressiveness as \TRC.    
    \Cref{sec:correct-placement} (c), (d):
    the placement of operator labels 
    does not matter for $\diagrams$. However, it does make a difference when replacing the labels with new relations.}
    \label{Fig_Anquors}
\end{figure}

Achieving a correct translation (the expected interpretation) is straightforward: 
we place each built-in relation in exactly the negation scope that it appears in the $\TRC$ expression.

\begin{example}[\Cref{ex:builtin relations} continued]
\label{ex:builtin relations2}
Our representation replaced the exact position of the join predicate with the built-in relation,
which nests it in the correct negation scope:
\begin{align*}
    \phantom{=} &
    \exists r \in R [\neg (
        \exists s \in S[ \h{\exists \anchor{j} \in \anchor{\textrm{``$<$''}} [ 
        r.A \equal \anchor{j}.\$1 \wedge \anchor{j}.\$2 \equal s.B]}  ])]  
\end{align*}    
\end{example}

\subsection{$\diagrams$ with built-in relations are pattern-complete for $\TRC$}
\label{sec:Theorem-built-in-relations}

We are ready to state our first of two main results of this paper.

\begin{theorem}[Full pattern expressiveness]
    \label{th:pattern expressiveness}
    There is an algorithm that translates any well-formed $\TRC$ query into \diagrams\ extended with built-in relations. 
    That representation has an unambiguous logical interpretation (there is another algorithm that translates that diagram back into $\TRC$)
    and has the same atoms as the original $\TRC$ query and thus has the same relational pattern.
\end{theorem}

\noindent
The proof uses constructive translations from $\TRC$ to $\diagrams$ with built-in relations and back
that are pattern-preserving. We thus give a solution to the disjunction problem.

\section{\systemx: a backward compatible solution without built-in relations}
\label{sec:real solution}
Our preliminary solution from the previous section solves the disjunction problem.
However, it has arguably two important problems:
(1) We lost the ability to use standard syntactic safety conditions to determine whether a query is domain-independent.
(2) The built-in relations and multiple nestings of negations increase the visual complexity and make the diagrams hard to 
read.\footnote{There is a reason why we have the disjunction operator in logic and natural language although it is not strictly necessary. 
Citing from~\cite{Barker-Plummer:2011} on disjunctions:
``...the fewer connectives we have, the harder it is to understand our sentences.''
}
The two solutions we introduce in this section are conceptually simple:
\Cref{sec:anchor simplifications} substitute the built-in relations with visual formalisms 
proposed in prior literature,
yet keep our rigorous and principled semantics defined earlier (recall \cref{Fig_Anquors}).
\Cref{sec:disjunctions} reintroduces disjunctions 
as (visual) shortcuts 
for our earlier rigorous semantics.
The result is a precisely defined pattern-preserving diagrammatic representation for \emph{any $\TRC$ query}
that allows visual verification of the safety conditions
and that specializes into \diagrams\xspace for the fragment of disjunction-free TRC.

\subsection{Substituting built-in relations}
\label{sec:anchor simplifications}

\subsubsection{Simpler anchors for unary built-in relations}
\label{sec:simpler:anchors}
Unary built-in relations consist of two boxes: the predicate (condition) name (e.g. \fcolorbox{black}{unaryBuiltInPredicate}{\textcolor{white}{$<$4}}) and an attribute box \selPredicateBox{\$1}.
We eliminate this unnecessary
indirection\footnote{This is similar in spirit to Tufte's recommendation to avoid legends if possible:
``labels are placed (directly) on the graphic itself; no legend is required.''~\cite{tufte2001visual}
}
and substitute both boxes with one box containing the condition
(e.g.\ \fcolorbox{black}{unaryBuiltInPredicatesimple}{{$<$4}}).
We thereby also recover visual formalisms from prior proposals, such as
VisualSQL~\cite{DBLP:conf/er/JaakkolaT03}
and VQL~\cite{DBLP:journals/tkde/MohanK93}
(see \cref{Fig_Prior_Disjunctions_4,Fig_Prior_Disjunctions_5}):
a selection is an equijoin between an attribute and a condition 
(e.g.\ 
\selPredicateBox{A}\textemdash{}\fcolorbox{black}{unaryBuiltInPredicatesimple}{{$<$4}}).
That condition still provides an anchor and could be in a deeper nesting than the table.
We call this the ``canonical'' representation.

If the condition is in the same negation scope as the relation, 
then we apply a shortcut that 
fuses the two attributes and thereby
recovers the selection formalism used by $\diagrams$
(e.g.\ \fcolorbox{black}{unaryBuiltInPredicatesimple}{{A$<$4}}).\footnote{We use a slight blue background for selection conditions as enabling feature (\cref{sec:enabling feature}), similiar to the yellow background used by QueryVis~\cite{DanaparamitaG2011:QueryViz}.}
Our formalism is thus backward compatible to $\diagrams$,
yet also allows us to give the condition a separate ``anchor'', which we need to express certain relational patterns.

\subsubsection{Binary built-in relations}
Binary built-in relations consist of three boxes: 
The predicate name (e.g.\ \fcolorbox{black}{binaryBuiltInPredicate}{\textcolor{white}{$<$}}) 
and two attributes connected to the respective relational attributes via equijoins.
We substitute these built-in relations with the symbols that were originally used by $\diagrams$ \emph{as labels} on
directed edges (arrows) (e.g.\ \builtinlabel{$<$}).
The important difference is that we treat the former \emph{labels} now as \emph{anchors} with the full semantic interpretation we developed in the last section (see e.g.\ \cref{Fig_anquors_simplification_4}).
This semantics allows us to explicitly place the anchor in a deeper nesting than either of the relations joined by that comparison predicate
and thereby improve upon the limited pattern expressiveness of \diagrams.

\begin{figure}[t]
    \centering
    \begin{subfigure}[b]{0.13\linewidth}
        \centering
        \includegraphics[scale=0.4]{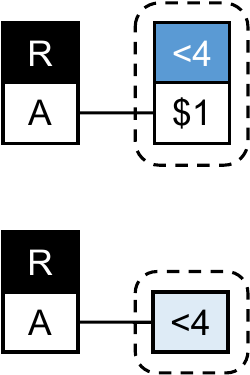}
        \vspace{0.5mm}
        \caption{}
        \label{Fig_anquors_simplification_3}
    \end{subfigure}	
    \hspace{0mm}
    \begin{subfigure}[b]{0.2\linewidth}
        \centering
        \includegraphics[scale=0.4]{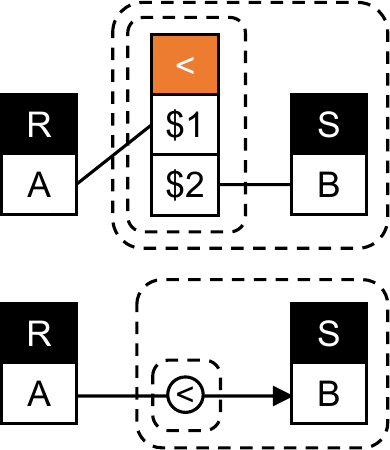}
        \caption{}
        \label{Fig_anquors_simplification_4}
    \end{subfigure}	
    \hspace{2mm}
    \begin{subfigure}[b]{.1\linewidth}
        \centering	
        \includegraphics[scale=0.4]{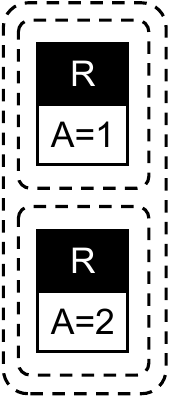}
        \caption{}
        \label{Fig_Disjunctions_b_1}
    \end{subfigure}	
    \hspace{0mm}
    \begin{subfigure}[b]{.11\linewidth}
        \centering	
        \includegraphics[scale=0.4]{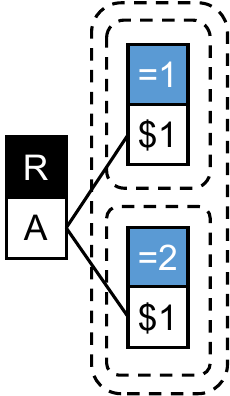}
        \caption{}
        \label{Fig_Disjunctions_b_2}
    \end{subfigure}	
    \hspace{1mm}
    \begin{subfigure}[b]{.11\linewidth}
        \centering	
        \includegraphics[scale=0.4]{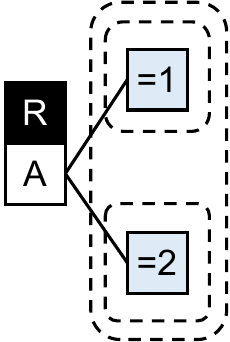}
        \caption{}
        \label{Fig_Disjunctions_b_3}
    \end{subfigure}	
    \hspace{1mm}
    \begin{subfigure}[b]{.12\linewidth}
        \centering	
        \includegraphics[scale=0.4]{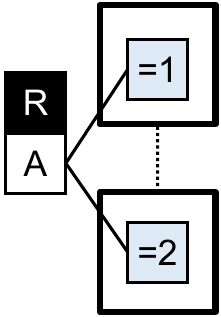}
        \vspace{0.8mm}
        \caption{}
        \label{Fig_Disjunctions_b_4}
    \end{subfigure}	
    \hspace{0mm}
    \begin{subfigure}[b]{.12\linewidth}
        \centering	
        \includegraphics[scale=0.4]{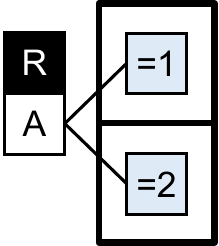}
        \vspace{2.9mm}
        \caption{}
        \label{Fig_Disjunctions_b_5}
    \end{subfigure}	
    \caption{
    \Cref{ex:simplifying builtin}(a-b):
    $\systemx$ are at both the middle and bottom rows,
    $\diagrams$ at the bottom row,
    $\diagrams$ with built-in relations at the top row.
    \Cref{ex:disjunction} (c-g): $\exists r \in R [r.A \equal 1 \vee r.A \equal 2]$: 
    (c): \diagrams.
    (d): $\diagrams$ with built-in relations.
    (e-g): $\systemx$.
    (f-g) show our \emph{visual shortcut} for disjunctions, formally justified with what we refer to as ``DeMorgan-fuse boxes''.}
    \label{Fig_anquors_simplification}
\end{figure}

\begin{example}[Substituting built-in relations]
    \label{ex:simplifying builtin}
Consider the Boolean $\TRC$ query
$\exists r \in R [\h{\neg (r.A \!<\! 4)}]$
shown in the lower row of \cref{Fig_anquors_simplification_3}
as $\systemx$.
The top row shows \diagrams\xspace with built-in relations which correspond to 
$\exists r \in R [\h{\neg (
    \exists \anchor{c} \in \anchor{\textrm{``$<$4''}} [ 
        r.A \equal \anchor{c}.\$1] )}]$.
Next, consider the Boolean query
\begin{align*}
    & \exists r \in R [\neg (
            \exists s \in S [\h{\neg(r.A < S.B)}])] \\
\intertext{
shown as $\systemx$ on the bottom of \cref{Fig_anquors_simplification_4}.
The top shows $\diagrams$ with built-in relations which correspond to 
}
& \exists r \in R [\neg (
        \exists s \in S[ 
            \h{\neg(\exists \anchor{j} \in \anchor{\textrm{``$<$''}} [ 
            r.A = \anchor{j}.\$1 \wedge \anchor{j}.\$2 = s.B])} ])]     
\end{align*}    
\end{example}

\subsection{Visual shortcut for disjunctions}
\label{sec:disjunctions}

As already mentioned earlier, 
frugality in primitive elements has downsides:
The fewer connectives we have, the harder it is to understand our sentences.
Despite negation and conjunction being truth-functionally complete, 
we regularly use disjunctions in logic and natural language.
We next introduce a visual shortcut for disjunction that 
allows us to recover a checkable safety condition and that
generalizes the various (incomplete) approaches for disjunctions we have seen in 
\cref{sec:intro_disjunctions}.

Our key idea is to \emph{keep the formal semantics we have developed so far} (and that solves the disjunction problem), 
but to allow a visual shortcut that we refer to as ``DeMorgan-fuse boxes.''
These boxes allow us to express 
$\neg(\neg(\varphi_1) \wedge \ldots \wedge \neg(\varphi_k))$, $k>2$
with 
$\varphi_1 \vee \ldots \vee \varphi_k$
by substituting nested double negations 
with bold rectangles, optionally connected via dotted lines.

\begin{definition}[DeMorgan-fuse boxes]
    Bold rectangular boxes that are either adjacent to each other or connected via dotted lines are 
    interpreted as if their anchored elements were connected via disjunctions.
    They are interpreted as if each of the boxes were replaced with individual negation scopes and the union of those boxes with another negation scope.
\end{definition}

\begin{example}[Simple disjunction]
	\label{ex:disjunction}
Consider the following disjunction from \cref{sec:intro_disjunctions}:
\begin{align}
	& \exists r \in R [r.A \equal 1 \;\h{\vee}\; r.A \equal 2]  \label{eq:simple-disjunction}
\end{align}
$\diagrams$ need to show two $R$ tables, either with union cells 
(\cref{Fig_Prior_Disjunctions_7})
or with a double negation (\cref{Fig_Disjunctions_b_1}):
\begin{align*}
	& \hspace{7.6mm}\exists r \in R [r.A \equal 1] \;\h{\vee}\; \exists r \in R[r.A \equal 2] \\
	& \hspace{1mm}\h{\neg}(\h{\neg}(\exists r \in R [r.A \equal 1]  \;\h{\vee}\; \exists r \in R[r.A \equal 2])) \\	
	& \h{\neg}(\h{\neg}(\exists r \in R [r.A \equal 1]) \;\h{\wedge}\; \h{\neg}(\exists r \in R[r.A \equal 2])) 
\end{align*}
$\diagrams$ with built-in relations replace the selection predicates with equijoins to two built-in relations (\cref{Fig_Disjunctions_b_2}):
\begin{align*}
    & \exists r \in R [\h{\neg(\neg(} \exists e_1 \!\in\! \textup{``=1''}[r.A \equal e_1.\$1]\h{)} 
        \;\h{\wedge}\; 
        \h{\neg(}\exists e_2 \!\in\! \textup{``=2''}[r.A \equal e_2.\$1] \h{))} ]                        
\end{align*}
$\systemx$ can either represent the statement via De Morgan (\cref{Fig_Disjunctions_b_3}):
\begin{align*}
	& \exists r \in R [\h{\neg(\neg(} r.A \equal 1 \h{)\; \wedge\; \neg(} r.A \equal 2 \h{))}]
\end{align*}
or via DeMorgan-fuse boxes (\cref{Fig_Disjunctions_b_5}), optionally connected via dotted edges (\cref{Fig_Disjunctions_b_4}).
\end{example}

\subsection{$\systemx$ is pattern-complete for $\TRC$ and solves the safety problem}
\label{sec:representationB}

We are ready to state our second of two main results of this paper.

\begin{theorem}[Pattern-isomorphism and safety of $\systemx$]
    \label{th:safety preservation}
    Every \diagram\ with built-in predicates has a pattern-isomorphic representation as $\systemx$ and vice versa.
    At the same time, $\systemx$ preserves the safety conditions of $\TRC$, 
    i.e.\ the syntactic safety conditions can be directly verified from the diagrammatic representation.
\end{theorem}

\noindent
The proof uses constructive translations from $\TRC$ to $\systemx$ and back
that are pattern-preserving and safety-preserving.
$\systemx$ thus solves both the disjunction and the safety problem.
Recall that $\TRC^{\neg \exists \wedge \vee}$ preserves the relational pattern and the safety conditions.
Because our translation from $\TRC^{\neg \exists \wedge \vee}$ preserves 
the negation scopes, 
the disjunctions,
and all atoms from the AST, the 4 safety conditions from
\cref{sec:safety} can be immediately read and verified from a $\systemx$ diagram.
\Cref{Fig_my_safety_example_Diagram_example} discusses our running example.

\begin{figure}[t]
    \begin{subfigure}[t]{.65\linewidth}
        \begin{example}[\cref{ex:safety} continued]
            \label{ex:safety_continued}        
            The $\TRC$ query from \cref{ex:safety} is equivalent to
            the following safe $\TRC^{\neg \exists \wedge \vee}$ fragment:
            \begin{align}
                &\{ q(A,B) \mid 
                    (\hblue{q.A=0} \wedge (\exists r\in R [\hblue{q.B=r.B}] \,\horange{\vee}\, q.B=1)) 
                    \hspace{10mm}\notag \\
                &\hspace{2mm}
                    \horange{\vee} \,
                    (\exists r \in R [\hblue{q.A=r.A} \wedge \hblue{q.B=r.B} \,\wedge \label{query:safety_continued}\\
                &\hspace{4mm}
                    \horange{\neg(\exists s}\in S[\horange{\neg (\exists r_2} \in R[r.A=r_2.A \wedge r_2.B=s.B ]\horange{)}]\horange{)}]) \}    \notag
            \end{align}
            
            \noindent
            The figure to the right shows the
            $\systemx$ representation for this query, which also corresponds to the AST from
            \cref{ex:safety} and
            \cref{Fig_my_safety_example_AST}.
            Notice how the 4 safety conditions can be applied directly on this diagram to verify that this query is safe.
        \end{example}
\end{subfigure}
\hspace{3mm}
\begin{subfigure}[t]{.3\linewidth}
    \centering
    \vspace{2mm}
    \includegraphics[scale=0.4]{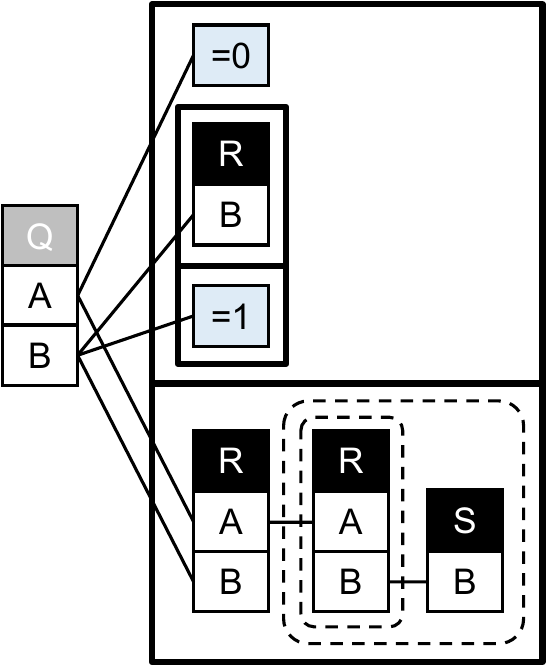}
\end{subfigure}
\caption{\Cref{ex:safety_continued} (\cref{ex:safety} continued).}
\label{Fig_my_safety_example_Diagram_example}
\end{figure}

\subsubsection{Size of representation}
\label{sec:size}
$\systemx$ has the same asymptotic size as $\TRC$. 
The reason is that the leaves of the AST (and the negation and disjunction scopes) get directly mapped to objects in the diagram.
At the same time, $\systemx$ is an exponentially smaller representation of $\TRC$ than \diagrams.
This is because $\diagrams$ require CNF formulas to be first transformed into DNF (i.e.\ to have disjunctions or unions as the root, which requires an exponential blow-up),
while our approach leaves disjunctions as \emph{inner operators} in the AST.

\begin{proposition}[Size preservation of $\systemx$]
    \label{prop:exponentialsize}
    $\systemx$ has the same asymptotic size as $\TRC$ and can be exponentially smaller than \diagrams.
\end{proposition}

\section{Generality of \systemx\ and enabling features}
\label{sec:generality}

We show that $\systemx$
unifies and generalizes prior approaches for disjunction (\cref{sec:generalizing prior solutions}),
justify its perceptual choices (\cref{sec:perceptual choices}),
and give solutions to prior challenges (\cref{sec:prior challenges}).

\subsection{Unification of prior approaches}
\label{sec:generalizing prior solutions}
We formalized our disjunction boxes as a visual shortcut for DeMorgan which allow us to recover safety conditions for $\TRC$.\footnote{Similarly, the expansion by built-in relations gives our diagrams formal and precise semantics, and we deploy the visual symbols as a shortcut for that formal semantics.}
They are also reminiscent of the ``union cells'' used by \diagrams\xspace
and prior box-based approaches going back all the way to Pierce (\cref{sec:intro_disjunctions}).
We next show that our formalism
actually \emph{unifies and generalizes the major diagrammatic approaches for disjunction} that we have surveyed earlier 
(thus all except text-based and form-based in \cref{Fig_Prior_Disjunctions})
and that it recovers each of them as a special case.

(1) \emph{DeMorgan-based disjunctions}: The semantics of our approach fundamentally builds upon DeMorgan and it is thus an instance by replacing the shortcut with its semantics.

(2) \emph{Box-based disjunctions}: our visual shortcut is similar to Peirce~\cite{peirce:1933}, Shin~\cite{shin_1995}, compound spider diagrams~\cite{HowseST2005:SpiderDiagrams}, and the union cells for $\diagrams$~\cite{DBLP:journals/pacmmod/GatterbauerD24}.
However, ($i$) our boxes are not just a union box, but it can be used in any nesting depth of the AST,
which gives us an exponentially more succinct representation.
Furthermore, ($ii$) our boxes are semantically justified via the DeMorgan-fuse as merely a visual shortcut for the actual semantics.

(3) \emph{Edge-based disjunctions}:
Similar to Shin~\cite{shin_1995}, we allow disjunction boxes to not merely be adjacent but also connected via dotted lines. 
We have two reasons:
($i$) Lines connecting the boxes give more flexibility in the placement of the boxes and
allow disjuncts to be placed in non-adjacent areas.\footnote{We do not discuss layout algorithms, and merely give a topological definition. It is easy to construct examples where a requirement of having anchors placed in adjacent boxes leads to overly strict constraints on the placement that would require extremely distorted join edges.}
($ii$) Optional lines also allow us to recover all prior edge-based solutions:
In the grammar of node-link diagrams~\cite{Ware:2020:InfoVis}
``\emph{line marks}'' connect ``\emph{point marks}'' (and not other line marks).
Edge-based disjunctions, however, draw specially highlighted or annotated lines between predicates, which are lines themselves
(see \cref{ex:disjunctions as edges} in \cref{app:recover edge-based disjunctions}).
As a basic visual construct, lines make a connection between entities~\cite{CardMackinleyShneiderman:1999:Readings}.
Thus, even if not shown, the lines connected via a disjunction line require some anchor points on the predicate edges.
Our disjunction boxes, even if infinitesimally small, give them these anchors and formal semantics via DeMorgan-fuse boxes.

(4) \emph{Form-based disjunctions}: Form-based approaches of disjunction like QBE~\cite{DBLP:journals/ibmsj/Zloof77} are not diagrammatic.
However, vertically arranged boxes 
(alternative choices for the same predicate are shown in different rows) 
have become a familiar visual pattern.
In order to mirror this familiar syntax 
we recommend showing disjunction boxes vertically aligned (as ``enabling feature'', \cref{sec:enabling feature}).

\subsection{Perceptual choices and Peirce shading }
\label{sec:perceptual choices}

\Cref{sec:perceptual justifications} justifies some of enabling features of the for \systemx, i.e.\ its perceptual choices.
We also show how to leverage an idea of alternative shading by Peirce~\cite{peirce:1933} that allows multiple readings of a given diagram  and thereby even recover universal quantifiers.

\subsection{Solutions to prior challenging queries}
\label{sec:prior challenges}

\Cref{sec:RDconclusions} shows our solution to prior challenges mentioned in \cref{sec:introduction}.

\section{100\% coverage of textbook benchmark}
\label{sec:pattern coverage}

\begin{wrapfigure}[9]{r}{0.5\textwidth}
        \centering
        \vspace{-7mm}
        \hspace{-5mm}
        \includegraphics[scale=0.23]{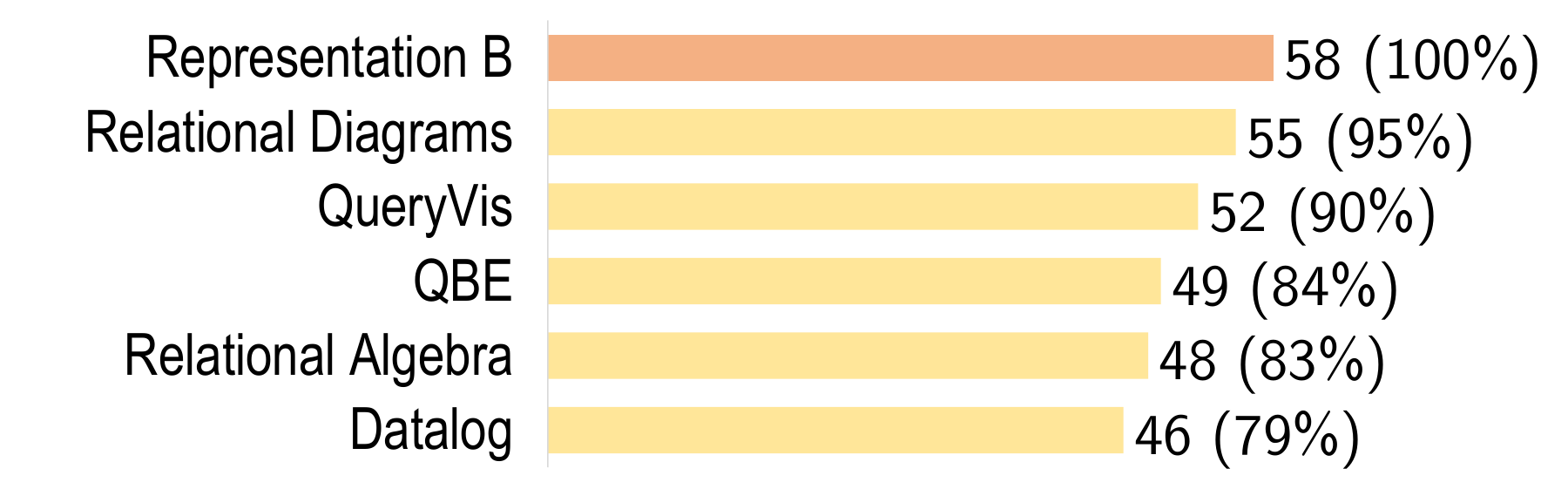}
        \caption{\cref{sec:pattern coverage}: Fraction among 58 queries from 5 textbooks with pattern-isomorphic representations in listed representations. $\systemx$ is the first diagrammatic representation
        to achieve 100\% coverage.}
        \label{Fig_Textbook_Analysis}
\end{wrapfigure}

The authors of \diagrams~\cite{DBLP:journals/pacmmod/GatterbauerD24}
gathered 58 queries from 5 popular database textbooks
\cite{cowbook:2002,DBLP:books/mg/SKS20,Elmasri:dq,date2004introduction,ConnollyBegg:2015}
and made them available on OSF.\footnote{Textbook benchmark: \url{https://osf.io/u7c4z/}. 
We noticed a slight discrepancy in the numbers and fixed the counts.}
We refer to that set simply as ``the textbook benchmark.''
They evaluated the pattern expressiveness of various text-based and diagram-based languages (we replicate their numbers) 
and showed that $\diagrams$ covered 95\% (55/58) of the queries in that benchmark.
Our approach is pattern-complete for $\TRC$ and
thus achieves 100\% pattern coverage (\cref{Fig_Textbook_Analysis}).
\Cref{sec:pattern coverage appendix} lists the 3 queries that $\diagrams$ cannot pattern-represent and shows their pattern-isomorphic representation in $\systemx$.

\section{Conclusion}
We derived $\systemx$,
a diagrammatic representation system that can represent any $\TRC$ query 
without changing the table signature and thus solves the disjunction problem.
We arrived at this solution by
first replacing join and selection predicates with relations 
and then defining the semantics based on placing relations in nested negation scopes and interpreting juxtaposition as conjunction.
We then hid the details behind existing visual formalisms
while redefining their semantics in the more rigorous relation-based interpretation
(recall \cref{Fig_Anquors}).
We then defined a box-based visual shortcut that 
can encode safety condition
and \emph{unifies, generalizes, and overcomes the shortcomings} of the 3 major prior graphical approaches for disjunctions.
It did this by showing that disjunction boxes (originally proposed by Peirce) 
can be pushed from the root into the branches of an AST representation of queries,
while keeping their rigorous semantic interpretation.

Solving the disjunction problem 
is important 
since SQL and all relational query languages 
have a firm basis in mathematical logic.
Any future visual query representation that can handle additional functionalities from SQL
(including aggregates and grouping)
needs to also incorporate a solution to the more general and longer-studied disjunction problem.

\begin{acks}
    This work was supported in part by
    the National Science Foundation (NSF) under award IIS-1762268.
\end{acks}

\clearpage
\clearpage

\bibliographystyle{ACM-Reference-Format}

\bibliography{queryvis-disjunction.bib}

\clearpage
\appendix
\section{Related work: Discussion of original diagrams from prior work (\cref{sec:intro_disjunctions})}
\label{app:Original_Drawings}

This section extends \cref{sec:intro_disjunctions} 
and gives a more detailed account of prior approaches for representing disjunctions.
It notably includes screenshots from the original literature.

\subsection{Text-based disjunctions}
\label{sec:text-based}
QBE~\cite{DBLP:journals/ibmsj/Zloof77} introduced the use of \emph{condition boxes}. 
Those are non-visual representations of Boolean conditions. See \cref{Fig_original_disjunctions_QBE2} for an example:
``Print the names of employees whose salary is between \$10000 and \$15000, provided it is not \$13000... Figure 24 illustrates the formulation of the AND operation using a condition box.''
We consider text-based representations as non-diagrammatic representations.

Dataplay~\cite{DBLP:conf/uist/AbouziedHS12} adapted a text-based style for disjunctions. 
\Cref{Fig_original_disjunctions_Dataplay} shows part of a query ``which ﬁnds students who are in the CS department, and took any of CS11, CS16 or CS18, and got at least one A in any course.''

\begin{figure}[h]
\centering	 
\begin{subfigure}[b]{.35\linewidth}
    \centering
\includegraphics[scale=0.29]{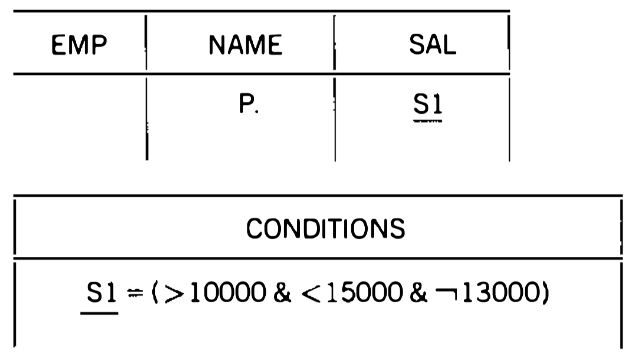}
\caption{\cite[Fig.~24]{DBLP:journals/ibmsj/Zloof77}}
\label{Fig_original_disjunctions_QBE2}
\end{subfigure}	
\hspace{1mm}
\begin{subfigure}[b]{.6\linewidth}
    \centering
\includegraphics[scale=0.36]{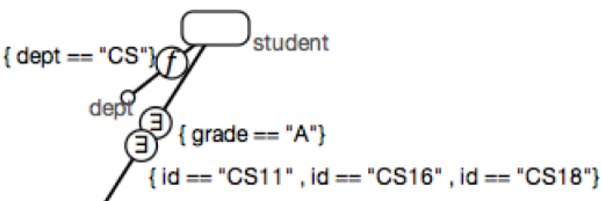}
\caption{\cite[Fig.~2]{DBLP:conf/uist/AbouziedHS12}}
\label{Fig_original_disjunctions_Dataplay}
\end{subfigure}	
\caption{\Cref{sec:text-based} (Option 1): Text-based disjunctions.}
\label{Fig_original_disjunctions_text}
\end{figure}

\subsection{Vertical form-based disjunctions}
\label{sec:form-based}
QBE~\cite{DBLP:journals/ibmsj/Zloof77} also introduces the ability to specify disjunctions by filling out two separate rows with alternative information. \Cref{Fig_original_disjunctions_QBE1} shows a query with a disjunctive filter:
``Print the names of employees whose salary is either \$10000 or \$13000 or \$16000.
This is illustrated in Figure 23. Different example elements are used in each row, so that the three lines express independent queries. The output is the union of the three sets of answers.''

SQLVis~\cite{DBLP:conf/vl/MiedemaF21} similarly provides a form-based interface for specifying conditions with disjunctions specified in separate rows.
\cref{Fig_original_disjunctions_SQLVis} shows a query selecting all attributes from customers whose city is either Amsterdam or Utrecht.

Datalog~\cite{DBLP:journals/tkde/CeriGT89} also expresses disjunctions (or unions) with repeated rules. 
Each rule written in one row, the union can also be interpreted as vertical disjunction:

\begin{figure}[h]
\centering	
\begin{subfigure}[b]{.31\linewidth}
    \centering
\includegraphics[scale=0.29]{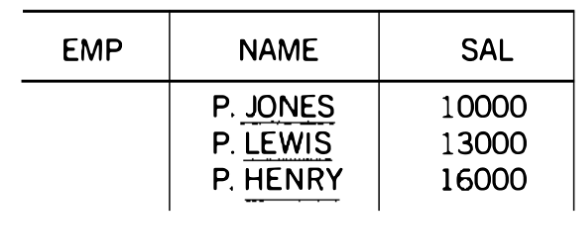}
\caption{\cite[Fig.~23]{DBLP:journals/ibmsj/Zloof77}}
\label{Fig_original_disjunctions_QBE1}
\end{subfigure}	
\hspace{2mm}
\begin{subfigure}[b]{.31\linewidth}
\centering
\includegraphics[scale=0.25]{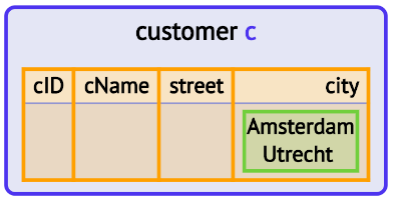}
\caption{\cite{SQLvis}}
\label{Fig_original_disjunctions_SQLVis}
\end{subfigure}	
\begin{subfigure}[b]{.31\linewidth}
\centering
\begin{align*}
    & Q \datarule R(x), x \equal 1 \\
    & Q \datarule R(x), x \equal 2 
\end{align*}   
\caption{Datalog}
\label{Fig_Datalog}
\end{subfigure}	
\caption{\Cref{sec:form-based} (Option 2): (Vertical) form-based disjunctions.}
\label{Fig_original_disjunctions_verticalform}
\end{figure}

\subsection{Edge-based disjunctions}
\label{sec:edge-based}

\begin{figure}[t]
    \centering	
    \begin{subfigure}[b]{.16\linewidth}
    \centering	
    \includegraphics[scale=0.22]{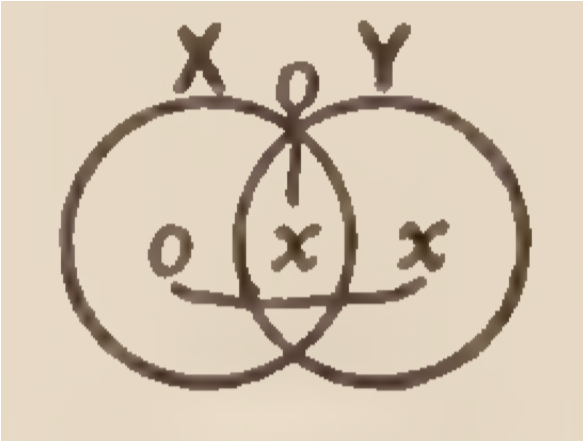}
    \caption{\cite[Fig.~30]{peirce:1933}}
    \label{Fig_original_disjunctions_VennPeirce}
\end{subfigure}	
\hspace{2mm}
\begin{subfigure}[b]{.18\linewidth}
    \centering	
    \includegraphics[scale=0.25]{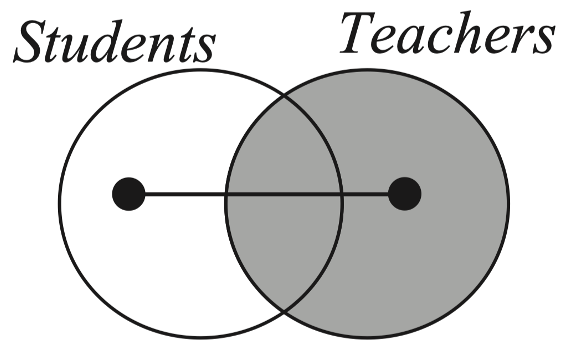}
    \caption{\cite[Fig.~9$d_3$]{DBLP:conf/iccs/Howse08}}
    \label{Fig_original_disjunctions_Spider}
\end{subfigure}	
\hspace{1mm}    
\begin{subfigure}[b]{.32\linewidth}
    \centering	
    \includegraphics[scale=0.44]{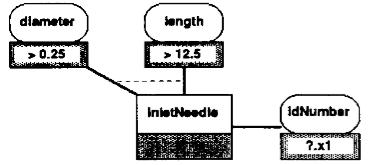}
    \vspace{-1mm}
    \caption{\cite[Fig 6(b)]{DBLP:journals/tkde/MohanK93}}
    \label{Fig_original_disjunctions_Mohan}
\end{subfigure}	
\hspace{1mm}    
\begin{subfigure}[b]{.23\linewidth}
    \centering	
    \includegraphics[scale=0.12]{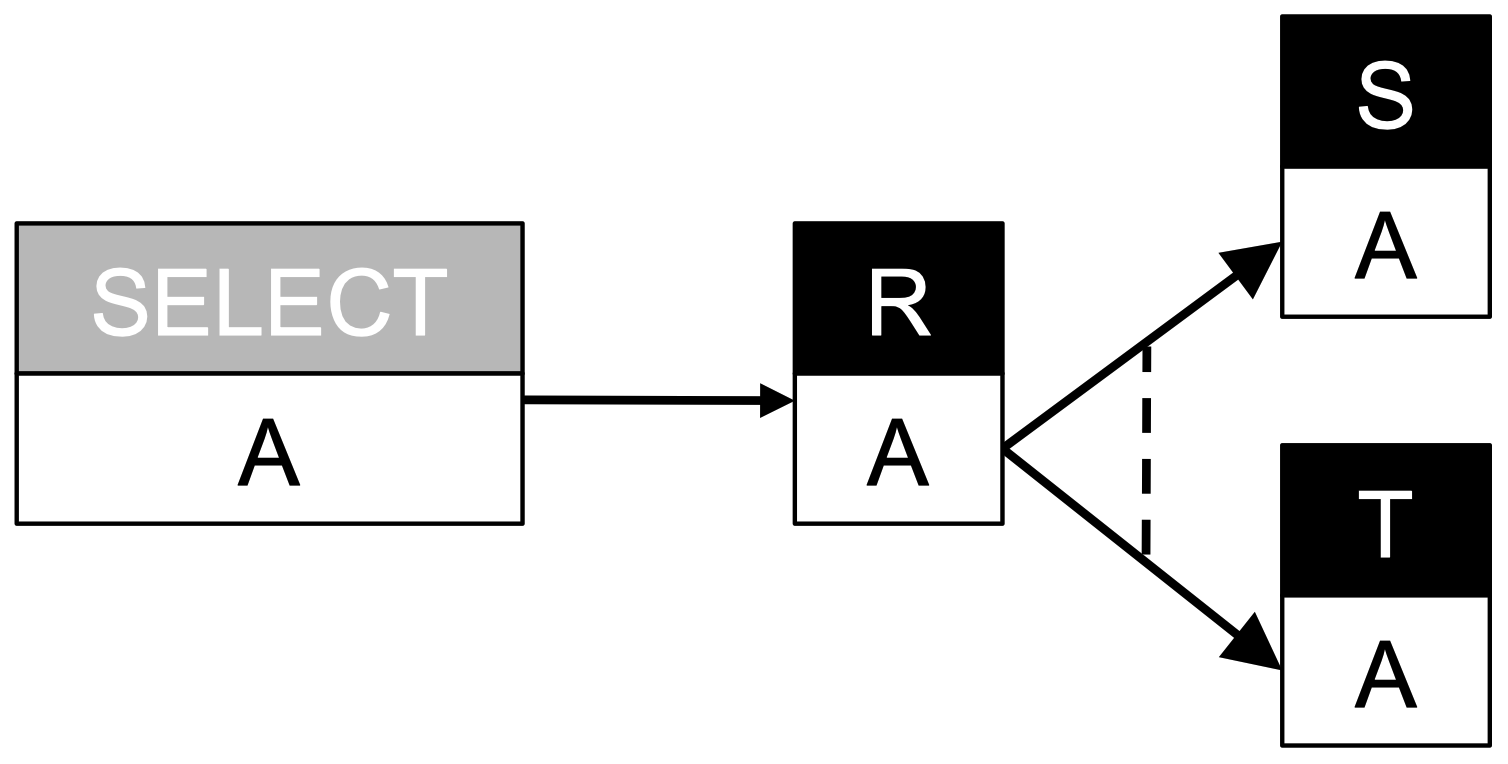}
    \caption{\cite{QV:2011:Slides}}
    \label{Fig_original_disjunctions_QueryVizSlide}
\end{subfigure}	
\hspace{50mm}    
\begin{subfigure}[b]{.47\linewidth}
    \vspace{1mm}
    \centering
    \includegraphics[scale=0.26]{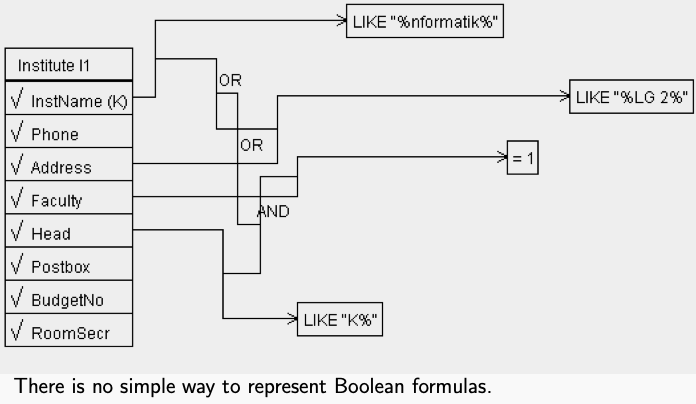}
    \caption{\cite{Thalheim:2013:Slides}}
    \label{Fig_original_disjunctions_Thalheim2}
\end{subfigure}	
\hspace{8mm}
\begin{subfigure}[b]{.41\linewidth}
    \centering
    \includegraphics[scale=0.36]{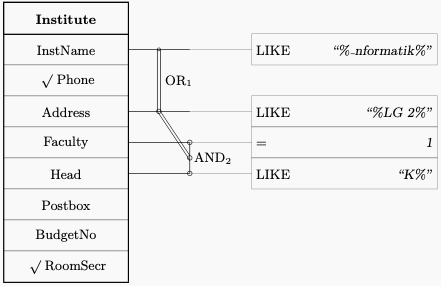}
    \vspace{1mm}
    \caption{\cite{Thalheim:2007:Slides}}
    \label{Fig_original_disjunctions_Thalheim1}
\end{subfigure}	
\begin{subfigure}[b]{.99\linewidth}
    \vspace{1mm}
    \centering
    \includegraphics[scale=0.38]{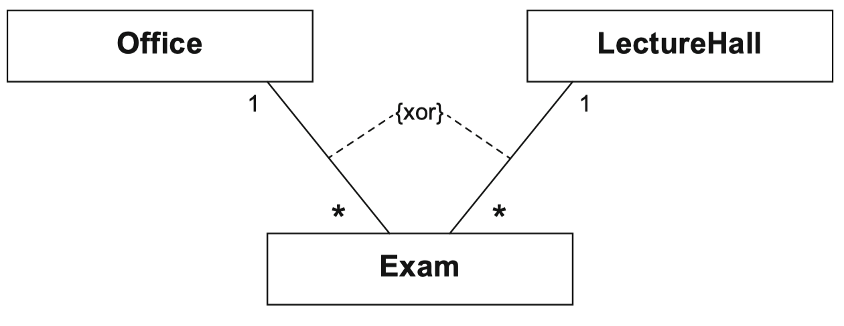}
    \vspace{-1mm}
    \caption{\cite{DBLP:series/utcs/SeidlSHK15}}
    \label{Fig_original_disjunctions_UML_Kappel}
\end{subfigure}	
\caption{\Cref{sec:edge-based} (Option 3): Edge-based disjunctions.}
\label{Fig_original_disjunctions_lines}
\end{figure}

The first use of edges connecting anchors of statements as syntactic devices to express disjunctive information that we found is by Charles Sanders Peirce around 1896~\cite{peirce:1933} 
in his work on extending the at the time recently introduced Venn diagrams~\cite{venn:1880}.
Peirce used "O" markers to show that a partition is empty (is false) and
"X" markers to show that a partition contains at least one member (is true).
To express disjunctions
he placed a line between any sequence of such markers meaning that at least one of these statements must be true (\cref{Fig_original_disjunctions_VennPeirce}).
Peirce writes~\cite[Paragraphs 4.359 and 4.360]{peirce:1933}: 
``Suppose, then, that signs in different compartments, if disconnected, are to be taken conjunctively, and, if connected, disjunctively, or vice versa...
Let this rule then be adopted: Connected assertions are made alternatively, but disconnected ones independently, i.e., copulatively.''
\Cref{Fig_original_disjunctions_VennPeirce} is the first figure in the literature that we found 
(though there may exist earlier ones that Peirce got inspired from and did not acknowledge or did not know of) 
that uses an edge to express disjunctions. The accompanying footnote reads:
``Either all $X$ is $Y$ or some $\neg X$ is $Y$, and some $X$ is $Y$ or all $\neg X$ is $Y$.”

Spider diagrams~\cite{HowseST2005:SpiderDiagrams} extend Euler diagrams~\cite{Euler:1802} with a variant of ``$X$-sequences'' (Peirce's syntactic devices for disjunction limited to existential information)
that are expressively equivalent to first-order monadic logic with equality.
The spider diagram in \cref{Fig_original_disjunctions_Spider} ``asserts that there is an element that is a student or a teacher but not both, and there are no other teachers.''

This idea to connect disjunctive information via lines of various forms was taken up repeatedly in various formalisms for visual query languages in the database community.
Mohan and Kashyap~\cite[Fig 6(b)]{DBLP:journals/tkde/MohanK93} propose VQL and write
``When two or more attributes are OR'ed in this fashion, it shows up on the visual query representation as a dashed line between those attributes as shown in Fig. 6(b).''
\cref{Fig_original_disjunctions_Mohan} shows the representation for ``Q3: Find the idNumbers of all the InletNeedles that have either diameter greater than 0.25 or have length greater than 12.5.''

In a presentation by the authors of QueryVis~\cite{DanaparamitaG2011:QueryViz}, 
we found \cref{Fig_original_disjunctions_QueryVizSlide} that uses similar visual notation. 
The diagram is accompanied by two \SQL statements, which translated into $\TRC$ would read as:
\begin{align*}
\phantom{=} 
&\{ q(A) \mid 
    \exists r \in R, s \in S [q.A \equal r.A \wedge (r.A = s.A \vee \exists t \in T
    [t.A \equal r.A])]\}    \\
&\{ q(A) \mid 
    \exists r \in R, s \in S, t \in T [q.A \equal r.A \wedge (r.A = s.A \vee t.A \equal r.A)]\}
\end{align*}
Based on our formalism, only the second interpretation is valid.

UML class diagrams use a dashed line to express an XOR constraint (exclusive or) that an object of class A is to be associated with an object of class B or an object of class C but not with both. 
\Cref{Fig_original_disjunctions_UML_Kappel} shows that an exam can take place either in an office or in a lecture hall but not in both~\cite[Fig.~4.14]{DBLP:series/utcs/SeidlSHK15}.

All of the prior approaches can only illustrate disjunctions for simple filters within the same table. 
This approach does not generalize to more complicated statements involving conjunctions and disjunctions.
In a presentation of VisualSQL~\cite{DBLP:conf/er/JaakkolaT03}, the main author Thalheim writes ``There is no simple way to represent Boolean formulas''~\cite{Thalheim:2013:Slides} and lists a challenging example shown in \cref{Fig_original_disjunctions_Thalheim2}.
The diagram returns Institutes where
InstName LIKE "\%nformatik\%"
or 
Address LIKE "\%LG 2\%"
or
(Faculty = 1
AND
Head LIKE "K\%").
In a different presentation, the same query is shown as~\cref{Fig_original_disjunctions_Thalheim1}.
Notice here how the precedence between the nested Boolean operators is indicated with numbered subscripts.

\subsection{Box-based disjunctions}
\label{sec:box-based}
In his chapter ``Of Euler's Diagrams,'' Peirce describes the difficulty of displaying disjunctions (especially for DNF) with line-based disjunctions~\cite[Paragraph 4.365]{peirce:1933}:
``It is only disjunctions of conjunctions that cause some inconvenience; such as “Either some A is B while everything is either A or B, or else All A is B while some B is not A.”’ Even here there is no serious difficulty. Fig.\ 59 
(\cref{Fig_original_disjunctions_VennPeirce2})
expresses this proposition. It is merely that there is a greater complexity in the expression than is essential to the meaning. There is, however, a very easy and very useful way of Fig.\ 59 avoiding this. It is to draw an Euler's Diagram of Euler's Diagrams each surrounded by a circle to represent its Universe of Hypothesis. There will be no need of connecting lines in the enclosing diagram, it being understood that its compartments contain the several possible cases. Thus, Fig.\ 60 
(\cref{Fig_original_disjunctions_VennPeirce3})
expresses the same proposition as Fig.\ 59.''
He thus proposed an alternative solution that we refer to as \emph{box-based disjunctions}:
He put Venn diagrams into rectangular boxes and interpreted each Venn diagram in a box as a disjunct.
Intuitively, he puts certain disjunctions into DNF and then describes them a as a union of expressions without disjunctions.
\Cref{Fig_original_disjunctions_VennPeirce3} is the first such figure we found.

\begin{figure}[t]
\centering	
\begin{subfigure}[b]{.20\linewidth}
    \includegraphics[scale=0.76]{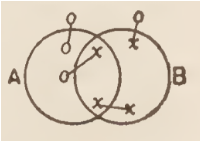}
    \caption{\cite[Fig.~59]{peirce:1933}}
    \label{Fig_original_disjunctions_VennPeirce2}
\end{subfigure}	
\hspace{10mm}
\begin{subfigure}[b]{.32\linewidth}
    \includegraphics[scale=0.65]{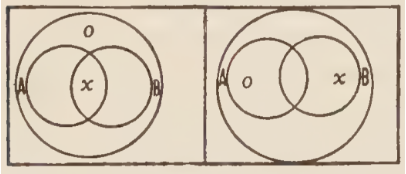}
    \caption{\cite[Fig.~60]{peirce:1933}}
    \label{Fig_original_disjunctions_VennPeirce3}
\end{subfigure}	
\hspace{20mm}
\begin{subfigure}[b]{.38\linewidth}
    \centering
    \includegraphics[scale=0.22]{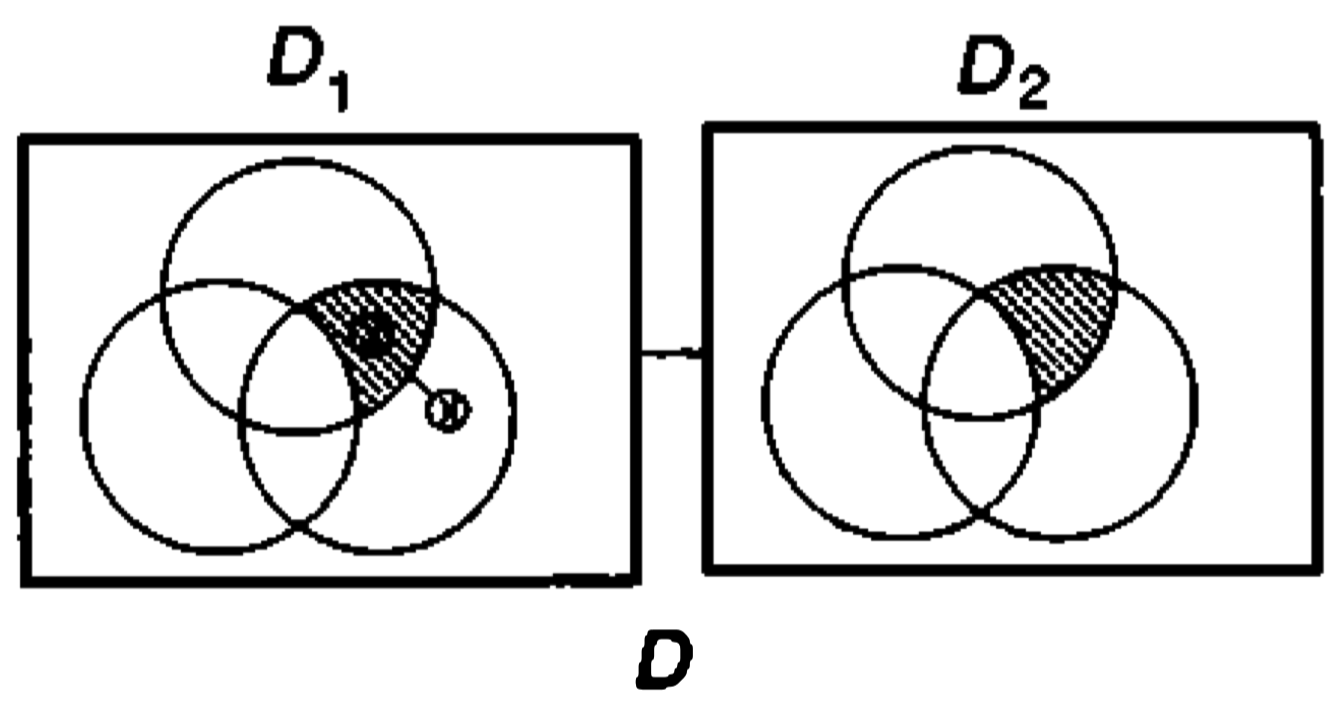}
    \vspace{-1mm}
    \caption{\cite[Example 45]{shin_1995}}
    \label{Fig_original_disjunctions_ShinVenn}
\end{subfigure}	
\hspace{7mm}
\begin{subfigure}[b]{.15\linewidth}
    \centering        
    \includegraphics[scale=0.47]{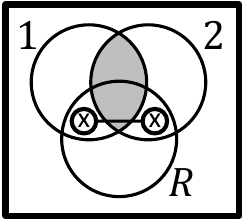}
    \vspace{4.1mm}
    \caption{}
    \label{Fig_Prior_Disjunctions_14}
\end{subfigure}	    
\hspace{7mm}
\begin{subfigure}[b]{.14\linewidth}
    \includegraphics[scale=0.5]{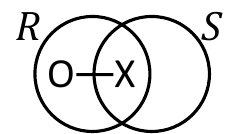}
    \vspace{1mm}
    \caption{}
    \label{Fig_Prior_Disjunctions_Shin_2}
\end{subfigure}	
\hspace{20mm}
\begin{subfigure}[b]{.33\linewidth}
    \centering
    \includegraphics[scale=0.5]{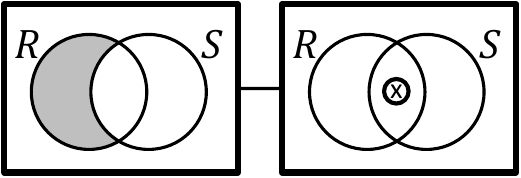}
    \vspace{1mm}
    \caption{}
    \label{Fig_Prior_Disjunctions_Shin_1}
\end{subfigure}	
\hspace{7mm}
\begin{subfigure}[b]{.33\linewidth}
    \centering
    \includegraphics[scale=0.29]{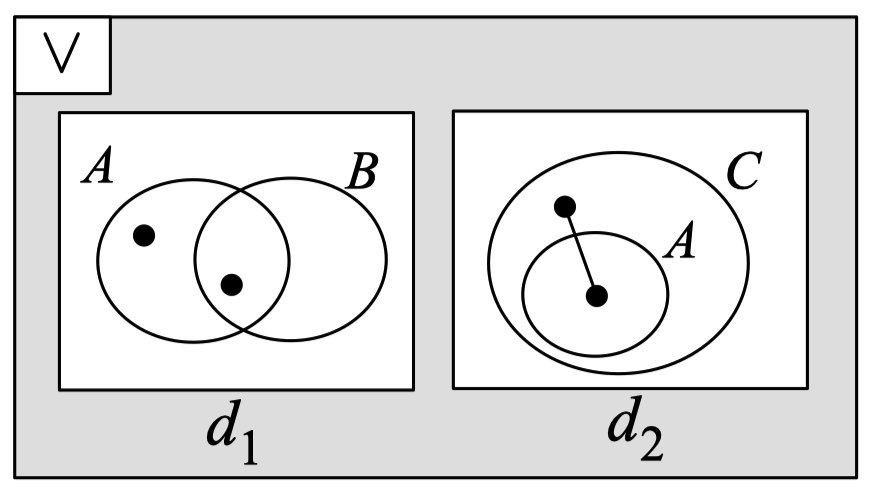}
    \caption{\cite[Fig.~12(i)]{HowseST2005:SpiderDiagrams}}
\label{Fig_spider_diagrams_original}
\end{subfigure}	    
\hspace{7mm}
\begin{subfigure}[b]{.14\linewidth}
    \centering
    \includegraphics[scale=0.255]{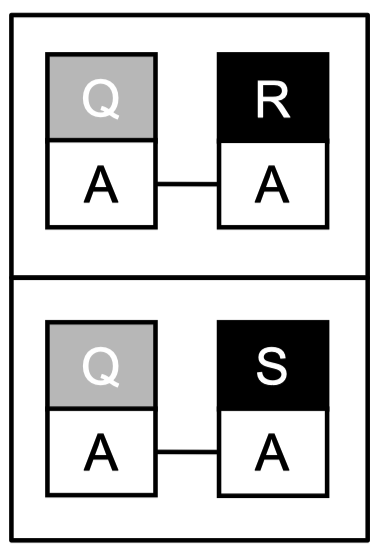}
    \vspace{0.5mm}
    \caption{\cite[Fig.~9e]{DBLP:journals/pacmmod/GatterbauerD24}}
    \label{Fig_original_disjunctions_RelationalDiagrams}
\end{subfigure}	
\caption{\Cref{sec:box-based} (Option 4): Box-based disjunctions.}
\label{Fig_original_disjunctions_boxes}
\end{figure}

Shin~\cite{shin_1995} keeps the "X" markers for existential statements and lines for dijunctions between existential statements
from Peirce-Venn diagrams~\cite{peirce:1933}.
But she replaces the "O" markers (universal statements) again with the original shading from Venn,
which removes their possible anchors for disjunctions.
She then uses straight lines connecting different boxed diagrams to represent disjunctions 
\emph{between statements at least one of which is universal}
(\cref{Fig_original_disjunctions_ShinVenn}).
The example from the introduction (\cref{Fig_Prior_Disjunctions_8})
does actually not need box-based disjunction and would be displayed by Shin as \cref{Fig_Prior_Disjunctions_14}.\footnote{In the introduction, we focused on the key visual metaphors and thus wanted to keep the example consistent which required us to simplify.}
For a more complete example, consider the statement 
``All R are S, or some R is S'':
$\forall x [R(x) \rightarrow S(x)] \vee \exists x [R(x) \wedge S(x)]$.
Peirce's original modification to Venn would show the statement as
\cref{Fig_Prior_Disjunctions_Shin_2},
yet Shin requires the disjunction boxes in this case: \cref{Fig_Prior_Disjunctions_Shin_1}.
Overall, Shin combines an edge-based disjunction with a box-based disjunction.
In Shin's words~\cite[Section 4]{shin_1995} (emphasis added):
``However, I find it easier to introduce a way to handle the information conveyed by a disjunction, rather than that conveyed by a negation. ... Therefore, we will introduce a new syntactic device in order to represent disjunctive information. Let us recall Peirce's suggestions for disjunctive information... First, Peirce introduced a line that connects syntactic objects. However, to adopt this we would have to sacrifice some visual power of the diagram. Peirce's other suggestion is to put Venn diagrams into rectangles and to interpret the Venn diagrams in each rectangle as disjunct. In this way, each of the Venn diagram(s) does not have to have a confusing line between o's or x's. However, in Venn-I \emph{we introduced a rectangle to represent a background set}. Therefore, we will not adopt Peirce's method directly, but \emph{will connect diagrams by a line}. In this way, we do not have to introduce any new syntactic object.''

Spider diagrams by Howse et al.~\cite{HowseST2005:SpiderDiagrams}
combine idea from Euler diagrams~\cite{Euler:1802}, Venn diagrams~\cite{venn:1880}, 
Peirce-Venn diagrams~\cite{peirce:1933}, and Shin-Venn diagrams~\cite{shin_1995}
whose logical expressiveness is equivalent to first-order monadic logic with equality.
They use what they call ``the box template'' to define $k$-ary compound diagrams consisting of smaller unitary (or recursively defined compound) diagrams. These box-based disjunction use explicit $\vee$ labels (\cref{Fig_spider_diagrams_original}).

\diagrams~\cite{DBLP:journals/pacmmod/GatterbauerD24} also cite inspiration from Peirce and represent a union of queries via adjacent ``union cells''. In their words~\cite[Sect.~5]{DBLP:journals/pacmmod/GatterbauerD24}:
``we allow placing several Relational Diagrams on the same canvas, each in a separate union cell. Each cell of the canvas then represents only conjunctive information, yet the relation among the different cells is disjunctive (a union of the outputs).''
\Cref{Fig_original_disjunctions_RelationalDiagrams} represents the query:
$\{q(A) \mid \exists r \in R[q.A = r.A] \vee \exists s\in S[q[A = s.A]    \}$.

All of these prior box-based approaches represent disjunction as a union of expressions and use variants of what we call ``union boxes''
to represent a union of fully formed sentences. While spider diagrams allow these box templates to be nested, the resulting structure still consists of \emph{a set of unitary diagrams}.
We are not aware of any prior approach that pushes the nesting with boxes from the root into the leaves of logical expressions within individual diagrams (e.g.\ notice how \emph{the content of the disjunction boxes are still connected to diagrammatic elements outside it} in \cref{Fig_Disjunction_FutureWork_1}).

\subsection{DeMorgan-based disjunctions}
\label{sec:deMorgen}

We call DeMorgan-based disjunctions, representations that use only symbols for negation and conjunction, 
and apply negation in a nested way in accordance with the logical identity
$A \vee B = \neg(\neg A \wedge \neg B)$,
which is famously named after Augustos De Morgan.

Existential graphs by Charles Sanders Peirce~\cite{peirce:1933}
are a diagrammatic notation for logical statements
that uses closed curves to express negation and juxtaposition for conjunctions. 
Nesting of those closed curves (called `cuts') permits expressing implication.
\Cref{Fig_original_disjunctions_PeirceEG1} is interpreted as:
``The two seps of Fig. 72, taken together, form a curve which I shall call a scroll... 
The only essential feature is that there should be two seps, of which the inner,
however drawn, may be called the inloop.''~\cite[Paragraph 4.436, Fig.~72]{peirce:1933}.
\Cref{Fig_original_disjunctions_PeirceEG2} is interpreted as
``‘Every salamander lives in fire,' or `If it be true that something is a salamander then it will always be true that something lives in fire’?''~\cite[Paragraph 4.449, Fig.~83]{peirce:1933}.
Appropriate nesting of three such cuts then creates disjunctions as in
\cref{Fig_original_disjunctions_PeirceBetaDisjunction1}
(``Everbody always either laughs or cries'')
and
\cref{Fig_original_disjunctions_PeirceBetaDisjunction2}
(``There is somebody who in every case either laughs or cries'').

\newsavebox{\bigimage}
\begin{figure}[t]
\centering	
\begin{subfigure}[b]{.14\linewidth}
    \centering	
\includegraphics[scale=0.35]{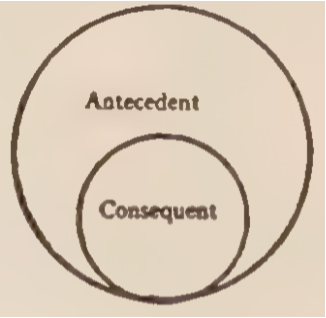}
\caption{\cite[Fig.~72]{peirce:1933}}
\label{Fig_original_disjunctions_PeirceEG1}
\end{subfigure}	
\hspace{2mm}
\begin{subfigure}[b]{.33\linewidth}
    \centering	
\includegraphics[scale=0.35]{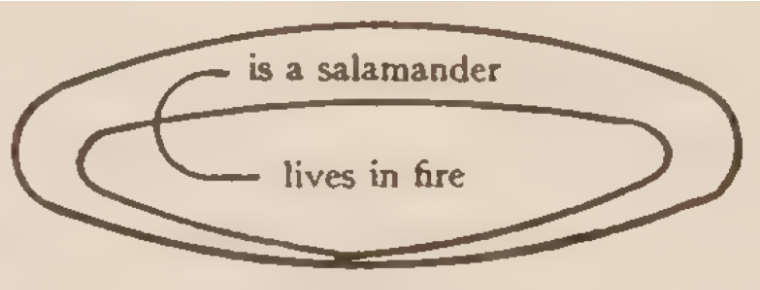}
\caption{\cite[Fig.~83]{peirce:1933}}
\label{Fig_original_disjunctions_PeirceEG2}
\end{subfigure}	
\hspace{2mm}
\begin{subfigure}[b]{.2\linewidth}
    \centering	
\includegraphics[scale=0.35]{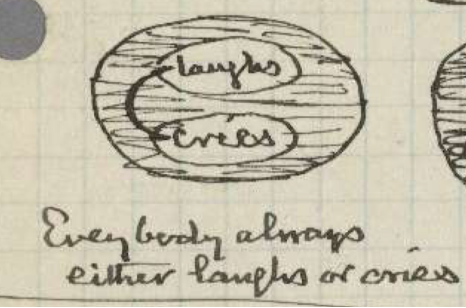}
\caption{\cite{Peirce:1906:Disjunction}}
\label{Fig_original_disjunctions_PeirceBetaDisjunction1}
\end{subfigure}	
\hspace{2mm}
\begin{subfigure}[b]{.19\linewidth}
    \centering	
\includegraphics[scale=0.35]{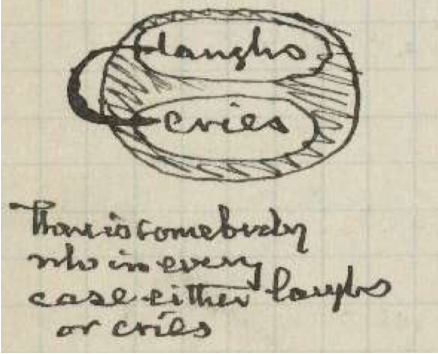}
\caption{\cite{Peirce:1906:Disjunction}}
\label{Fig_original_disjunctions_PeirceBetaDisjunction2}
\end{subfigure}	
\hspace{10mm}
\begin{subfigure}[b]{.20\linewidth}
    \centering	
\includegraphics[scale=0.23]{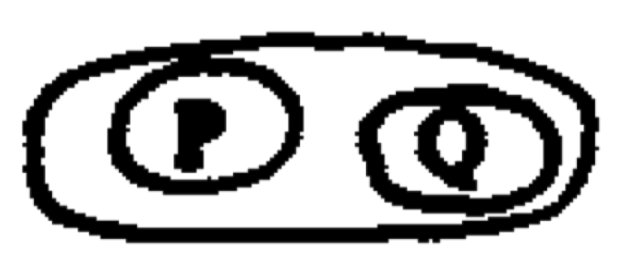}
\caption{\cite[Sect.\ 3, Fig.\ 7]{Roberts:1973}}
\label{Fig_original_disjunctions_Roberts_2}
\end{subfigure}	
\hspace{8mm}
\begin{subfigure}[b]{.35\linewidth}
    \centering	
\includegraphics[scale=0.4]{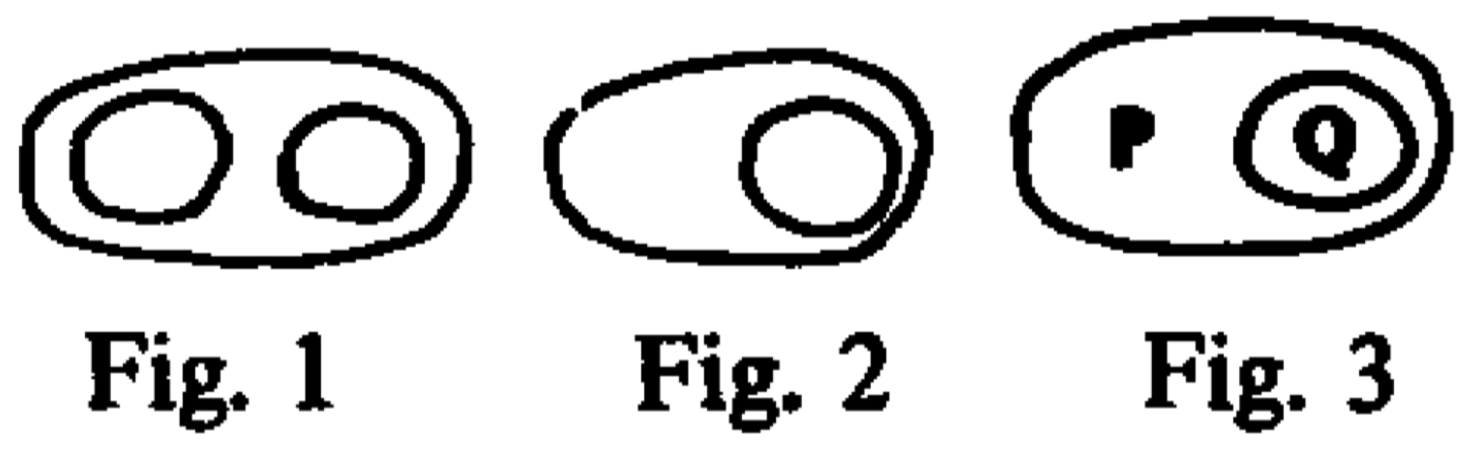}
\vspace{-2mm}
\caption{\cite[Sect.~7.2, Fig.~1-3]{Roberts:1973}}
\label{Fig_original_disjunctions_Roberts_3}
\end{subfigure}	
\hspace{10mm}
\sbox{\bigimage}{%
\begin{subfigure}[b]{.5\linewidth}
\vspace{1mm}	
\includegraphics[scale=0.55]{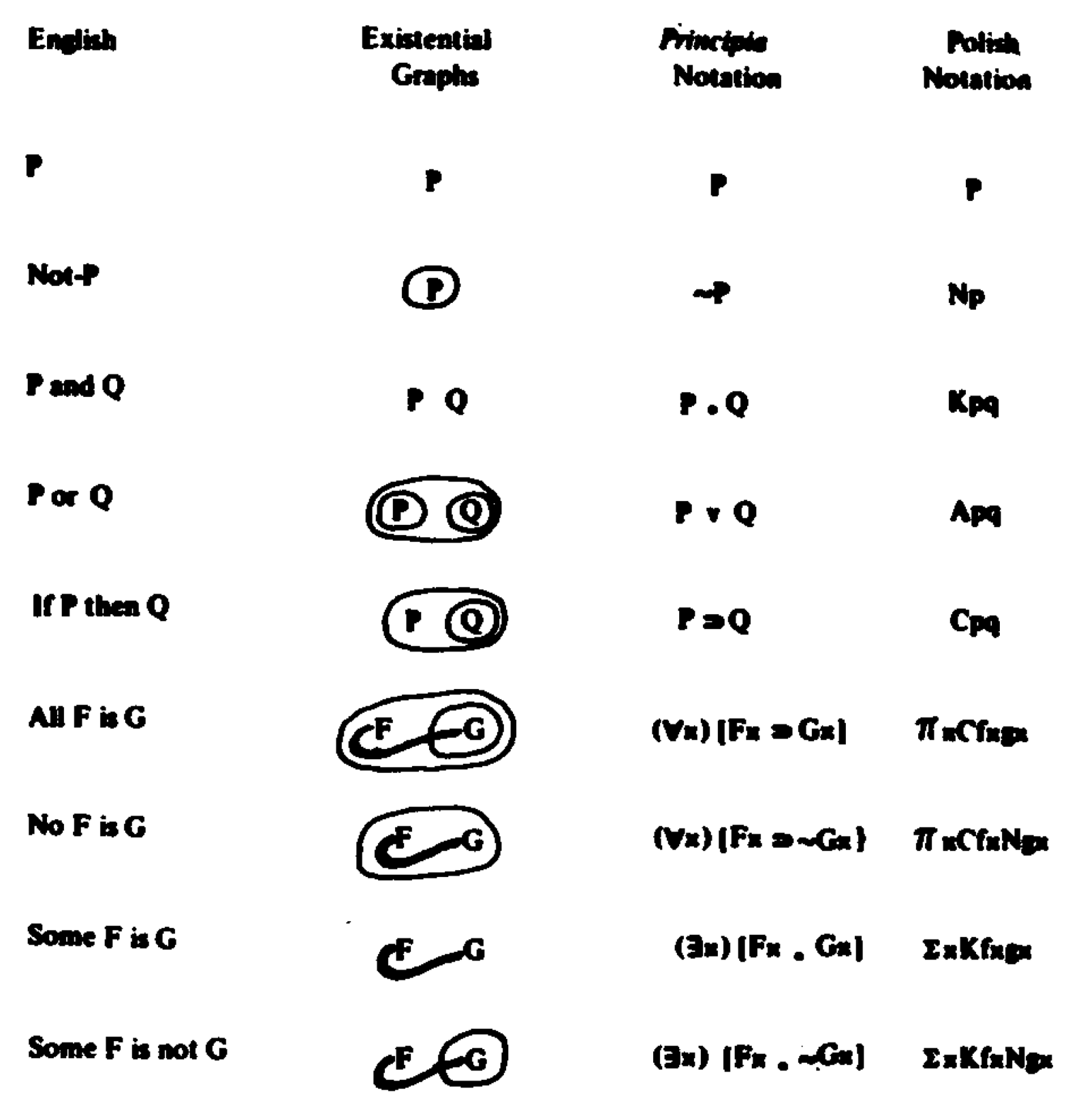}
\caption{\cite[Appendix 2, condensed]{Roberts:1973}}
\label{Fig_original_disjunctions_Roberts_4}
\end{subfigure}	
}
\usebox{\bigimage}
\hspace{3mm}
\begin{minipage}[b][\ht\bigimage][s]{.4\textwidth}
    \begin{subfigure}[b]{.9\linewidth}
        \centering
    \vspace{3mm}
    \includegraphics[scale=0.28]{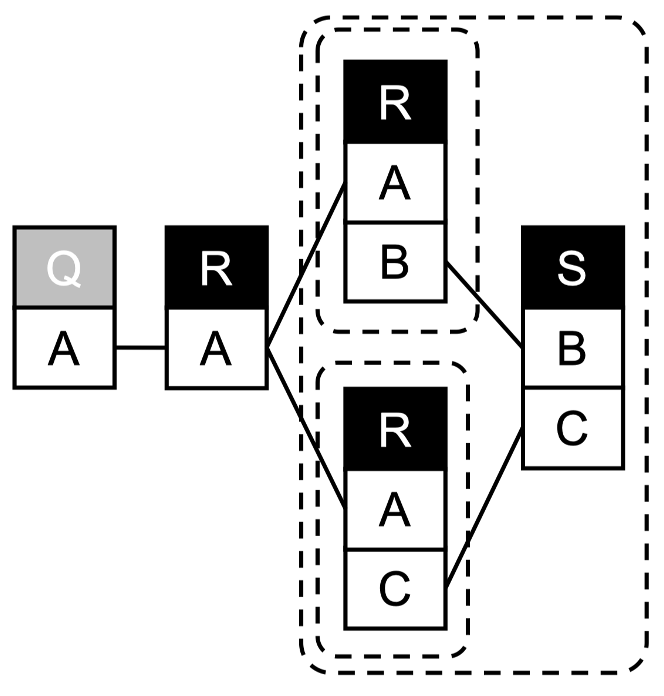}
    \caption{\cite[Fig.~9c]{DBLP:journals/pacmmod/GatterbauerD24}}
    \label{Fig_original_disjunctions_RelationalDiagrams2}
    \end{subfigure}	
    \vfill
    \begin{subfigure}[b]{.9\linewidth}
        \centering
    \includegraphics[scale=0.17]{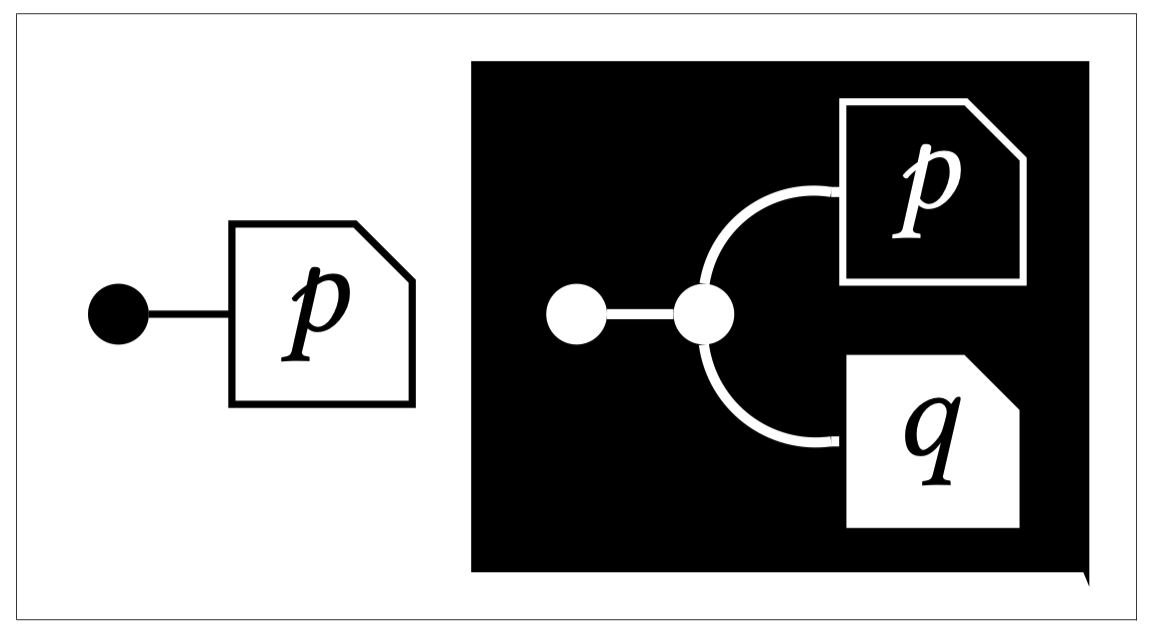}
    \caption{\cite[Fig.~8]{StringDiagrams:2024:arxiv}}
    \label{Fig_original_disjunctions_String}
    \end{subfigure}	
\vspace{0pt}%
\end{minipage}
\caption{\Cref{sec:deMorgen} (Option 5): DeMorgan-based disjunctions.}
\label{Fig_original_disjunctions_DeMorgan}
\end{figure}

Roberts (under the guidance of Max Harold Fisch) in the 1970s worked through the notes by Peirce and made the key ideas available to a larger audience in his influential interpretational work~\cite{Roberts:1973}.
We include figures from his work as 
\cref{Fig_original_disjunctions_Roberts_2}
(``Fig.~7 can be read 'It is false, that P is false and Q is false'.'')
and
\cref{Fig_original_disjunctions_Roberts_3}:
``a graph of even meager complexity can be read in several different English sentences, if only the reader will keep in mind two or three basic patterns. Fig. 1 is the pattern for an alternative proposition, and Fig. 2 is the pattern or a conditional. Just imagine what Fig. 3 would look like with a double cut surrounding P, and it will be obvious that 'Either not-P or Q' means the same as 'If P then Q'. A student of EG will learn automatically many of the linguistic equivalences that require an excess of time and symbols in algebraic notations.''
\Cref{Fig_original_disjunctions_Roberts_4} gives an overview of the possible cut nesting and resulting logical patterns.

\diagrams~\cite{DBLP:journals/pacmmod/GatterbauerD24} were inspired by the nestings of Peirce (and inspired our work).
\cref{Fig_original_disjunctions_RelationalDiagrams2} 
shows the representation of the following query by applying De Morgan:
\begin{align*}
\phantom{=} &\{ q(A) \mid 
    \exists r \in R [q.A \equal r.A \wedge \neg (\exists s \in S\\
    &\hspace{10.5mm}[{\h{\neg (}\exists r_2 \in R [ }(r_2.B \equal s.B \,\h{\vee}\, r_2.C \equal s.C) \wedge r_2.A \equal r.A ] \h{)} ])] \} \\
    = &\{ q(A) \mid 
      \exists r \in R [q.A \equal r.A \wedge \neg (\exists s \in S\\
        &\hspace{10.5mm}[\h{\neg (}\exists r_2 \in R [r_2.B \equal s.B \wedge r_2.A \equal r.A ]\h{)} \;\h{\wedge}\;	\\
        &\hspace{11.7mm}\h{\neg (}\exists r_3 \in R [r_3.C \equal s.C \wedge r_3.A \equal r.A ]\h{)} ]] \}	
\end{align*}

String diagrams~\cite{DBLP:conf/diagrams/Haydon2020:StringDiagrams,StringDiagrams:2024:arxiv}
are a variant of Peirce’s beta graphs that allow free variables in addition to bound variables. 
Both types of variables are represented by lines, yet bound variables are represented by a dot.
\Cref{Fig_original_disjunctions_String} shows the representation for 
$\exists x[p(x)] \wedge \forall y[p(y) \rightarrow q(y)]$

\section{Extended discussion on defining $\TRC$ (\cref{sec:formal setup})}
\label{app:TRC details}
This section extends \cref{sec:TRC,sec:safety}
and provides a more detailed discussion of how our formalism for $\TRC$
and the formal safety conditions
compare to prior work.

\subsection{Different well-formed $\TRC$ formalisms}
The database literature vastly differs in their treatment and formal definition of $\TRC$. 
We give here an abbreviated comparison to our chosen formalism and explain the motivation for our choices.
In essence, we believe that our streamlined formalization of $\TRC$ 
(and the resulting safety conditions based on ASTs) 
are a key reason why the translation from $\TRC$ to $\systemx$ appear so natural.

\subsubsection{Dedicated output variable}
Several textbooks~\cite{Elmasri:dq,cowbook:2002,Ullman1988PrinceplesOfDatabase,DBLP:books/mg/SKS20}
and Codd's original formulation of $\TRC$~\cite{DBLP:persons/Codd72}
allow a tuple variable bound to a particular table to become the output variable.
For example Ramakrishnan and Gehrke~\cite{cowbook:2002} would allow the following notation:
$$
    \{ r \mid r \in R \wedge r.B \!>\! 4 \}
$$
This notation does not make explicit the schema of the output (in this case $r(A, B)$) nor its arity. 
This could be partially fixed by always requiring to show the header of the output tuple.
We, in contrast, always require a new free output variable 
(which is also required by the above notation if the output attributes are chosen from more than one table, 
but just not enforced when not needed) 
which 
simplifies the formation rules of well-formed formulas
and also later simplifies the translation into a diagrammatic representation

We thus write the above query instead as:
$$
    \{ q(A, B) \mid r \in R[q.A\!=\!r.A \wedge q.B\!=\!r.B \wedge r.B \!>\! 4] \}
$$

\subsubsection{Combined variable bindings and quantification}
In our definition of well-formed formulas, 
we restrict any bound variable (either existentially or universally quantified) 
to range over a specific relation before it is used.
This choice is also adopted by most textbooks such as 
\cite{cowbook:2002,DBLP:books/mg/SKS20}.
This notation permits our concise definition of the Abstract Syntax Tree (AST) of a $\TRC$ query
because bound attributes can only be referenced 
within a subtree of the quantifier node in the syntax tree that binds them.

Codd's original formulation~\cite{DBLP:persons/Codd72},
and also the books by Ullman~\cite{Ullman1988PrinceplesOfDatabase} 
and Elmasri and Navathe~\cite{Elmasri:dq},
separate our combined quantification and binding (``$\exists r \in R$'')
into a quantification (``$\exists r$'') and a \emph{range term} (``$r \in R$'').
For example, Ullman's definition~\cite{Ullman1988PrinceplesOfDatabase} would allow writing above query as
$$
    \{ q(A, B) \mid \h{\exists r}[q.A\!=\!r.A \wedge q.B\!=\!r.B \wedge r.B \!>\! 4 \wedge \h{r \in R}] \}
$$
Notice how the variable $r$ is used in an atom (``$q.A=r.A$'') whose binding is defined 
in a sibling node (``$r \in R$'') of a shared parent conjunction node $\wedge$
instead of a parent. 
This more permissive notation does not add logical expressiveness 
but leads to Ullman's formal safety condition \cite[Sect.~3.9]{Ullman1988PrinceplesOfDatabase} requiring 
an inductive argument about ``limited variables'' that are bound via possibly multiple equijoins.
To see this, notice that the above query could also be written 
by defining an additional (and arguably unnecessary) tuple variable $t$ 
that acts only as intermediary between output variable $q$ and the bound variable $r$:
$$
    \{ q(A, B) \mid \exists r, \h{\exists t}
        [q.A\!=\!\h{t.A} \wedge \h{t.A} \!=\! r.A \wedge q.B\!=\!r.B \wedge r.B \!>\! 4 \wedge r \in R] \}
$$

\subsubsection{Multiple bindings of the same variable}
Separation of bindings and quantifications permits 
using the same tuple variable in multiple range terms as in
$$
    \{ q \mid q \in R \wedge \neg (q \in S) \}
$$
Our notation captures the same fragment but requires each ``range term'' to use a separate tuple variable, 
as in 
$$
    \{ q(A) \mid \exists r \in R[q.A \!=\! r.A \wedge \neg(\exists s \in S[r.A \!=\! s.A]] \}
$$

\subsubsection{Named vs.\ unnamed perspective}
We follow most literature in using the named perspective of $\TRC$. 
We thus assume all attribute names of relations are known and their relative order does not matter.
Notable exceptions are 
Codd's original formulation of $\TRC$~\cite{DBLP:persons/Codd72} and
Ullman's textbook~\cite{Ullman1988PrinceplesOfDatabase} 
where they refer to the 2nd attribute (component) of a tuple variable $r$ by ``$r.[2]$'' as in:
$$
    \{ r^{(2)} \mid r \in R \wedge r.[2] \!>\! 4 \}
$$

\subsubsection{Minor notational differences}
Our set notation for the binding predicate ``$r\in R_i$'' is also used by Ramakrishnan and Gehrke~\cite{cowbook:2002} and Silberschatz et al.~\cite{DBLP:books/mg/SKS20}.
Codd~\cite{DBLP:persons/Codd72} originally used the notation ``$P_i r$'',
with $P_i$ being a monadic predicate corresponding to table $R_i$, and called it a \emph{range term}.
Ullman~\cite{Ullman1988PrinceplesOfDatabase} and Elmasri and Navathe~\cite{Elmasri:dq} write the binding predicate
instead as ``$R_i(r)$''.

Our notation ``$r.A$'' for a tuple attribute (or component) is also used by 
Ramakrishnan and Gehrke~\cite{cowbook:2002} and Elmasri and Navathe~\cite{Elmasri:dq}.
Codd~\cite{DBLP:persons/Codd72}, Ullman~\cite{Ullman1988PrinceplesOfDatabase}, and Silberschatz et al.~\cite{DBLP:books/mg/SKS20} instead 
use the notation ``$r[A]$'' or ``$r(A)$''.

Our notation ``$\exists r\in R, \exists s\in S [r>4]$'' for the scope of tuple variables
is similar to the notation of Ramakrishnan and Gehrke~\cite{cowbook:2002}
and Silberschatz et al.~\cite{DBLP:books/mg/SKS20}.
However, we use brackets ``$[\ldots]$'' instead of parentheses ``$(\ldots)$'' for variable scopes, 
and use parentheses instead for the scope of negation $\neg(\ldots)$. 
Other books~\cite{Elmasri:dq,Ullman1988PrinceplesOfDatabase} use parentheses around the quantifiers as in 
``$(\exists r)(\exists s)(r\in R \wedge s \in S \wedge r>4)$''.

Date~\cite{date2004introduction} uses a notation for ``tuple calculus'' that is entirely different (see e.g.~\cref{Fig_Date_Query}).
The books by Abiteboul et al.~\cite{DBLP:books/aw/AbiteboulHV95}
and Arenas et al.~\cite{ArenasBLMP:DatabaseTheory}
do not cover $\TRC$ at all.

\subsection{Different safey criteria for $\TRC$}

The safety conditions for $\TRC$ queries 
are rarely discussed in database textbooks.
Abiteboul et al.'s concept of range-restricted $\RC$ queries~\cite{DBLP:books/aw/AbiteboulHV95} is only defined for $\DRC$ but not $\TRC$.
Ullman's definition of safety~\cite{Ullman1988PrinceplesOfDatabase} extends to $\TRC$ (we discuss it in more detail below) but it excludes correlated nested queries (e.g.\ find sailors who reserved all red boats), 
which we explicitly like to include.
The books
by Ramakrishan and Gehrke~\cite{cowbook:2002},
Elmasri and Navathe~\cite[Sect.~8.6.8]{Elmasri:dq}, and
Silberschatz
\cite[Sect.~27.2]{DBLP:books/mg/SKS20}
mention the concept of safety and its connection to domain-independence
but do not include any syntactic definition for safety.
The very recent book in progress by Areans et al.~\cite{ArenasBLMP:DatabaseTheory}
do mention safety in its current draft form.
The safety definition for $\TRC$ on Wikipedia~\cite{wiki:TRC} is clearly wrong
(it does not make a difference between existential and universal quantifiers),
and has been wrong for over 10 years, according to the revision history.

Since our formal safety conditions differ slightly from prior definitions 
(they are at times more permissive and at other times more restrictive),
we give here a more extensive justification for our choices, 
and at the end also a proof of them being a valid safety restriction.

\subsubsection{One single binding predicates per output attribute}
Our definition of safety does not allow an output attribute to be bound more than once, 
except if the multiple bindings correspond to disjunction in which case they are required.
Our requirement simplifies the reading since every output attribute is connected via one equality predicate to exactly one table column (for each disjunct).
Our requirement does not remove any expressiveness
since it is always possible to remove multiple bindings except for one and replace them with join predicates, and
one single binding via an equijoin is always required by safety conditions
(e.g.\ Ullman~\cite{Ullman1988PrinceplesOfDatabase} defines this condition with the notion of a ``limited variable'' which is bound via a sequence of equijoins to a domain).

We give two illustrating examples.

\begin{figure}[t]
    \centering
    \begin{subfigure}[b]{0.16\linewidth}
        \centering
        \includegraphics[scale=0.42]{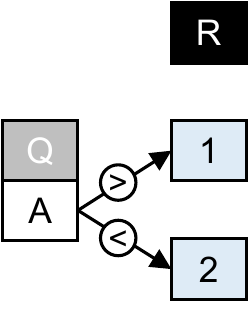}
        \vspace{0mm}
        \caption{}
        \label{TRC_unsafe_example_1}
    \end{subfigure}	
    \hspace{6mm}
    \begin{subfigure}[b]{0.2\linewidth}
        \centering
        \includegraphics[scale=0.42]{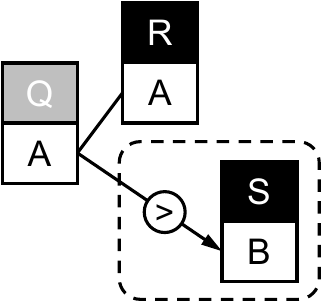}
        \vspace{0mm}
        \caption{}
        \label{TRC_unsafe_example_2}
    \end{subfigure}	
    \hspace{3mm}    
    \begin{subfigure}[b]{0.3\linewidth}
        \centering
        \includegraphics[scale=0.42]{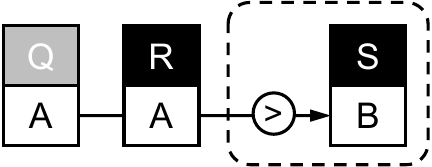}
        \vspace{3mm} 
        \caption{}
        \label{TRC_unsafe_example_3}
    \end{subfigure}	
    \caption{\Cref{ex:TRC_unsafe_example}: Three well-formed $\TRC$ queries. Only (c) is safe according to our safety conditions.}
    \label{TRC_unsafe_example}
\end{figure}

\begin{example}[]
    \label{ex:TRC_unsafe_example}
    First, consider the following well-formed but unsafe $\TRC$ query:
    $$
        \{ q(A) \mid \exists r \in R[q.A\!>\!1 \wedge q.A\!<\!2]\}
    $$
    It is \emph{not domain-independent} and returns all (possibly infinite) domain values between 1 and 2, as long as the table $R$ is not empty.
    
    Next, consider the following unsafe but domain-independent query:
    $$
        \{ q(A) \mid \exists r \in R[q.A\!=\!r.A \wedge \neg (\exists s \in S[q.A\!>\!s.A])] \}
    $$
    It is unsafe according to our safety definition as the output attribute $q.A$ is referenced twice.
    To make it safe, we can remove one occurrence and replace it with $r.A$ to which $q.A$ is equated in the binding predicate:
    $$
        \{ q(A) \mid \exists r \in R[q.A\!=\!r.A \wedge \neg (\exists s \in S[r.A\!>\!s.A])] \}
    $$
    See \Cref{TRC_unsafe_example} for their respective $\systemx$ representations.
\end{example}

\begin{example}[]
    \label{ex:TRC_AST_output_binding}
    Consider the following well-formed $\TRC$ query:
    \begin{align*}
        &
        \{q(A) \mid \exists r \in R [\h{q.A} \!=\! r.A] \wedge \h{q.A}\!=\!4 \}
    \end{align*}    
    It is well-formed, but not safe according to our definition because $q.A$ is bound more than twice, and those bindings are connected in the AST. 
    The query, however, is equivalent to the following safe query:
    \begin{align*}
        &
        \{q(A) \mid \exists r \in R [q.A \!=\! r.A] \wedge r.A\!=\!4 \}
    \end{align*}    
    Notice how the single output edge from $q.A$ makes the resulting $\systemx$ representation easier to read.
    The following slight variation also binds $q.A$ twice, but those bindings are connected via a disjunct and are thus needed:
    \begin{align*}
        &
        \{q(A) \mid \exists r \in R [q.A \!=\! r.A] \,\h{\vee}\, q.A\!=\!4 \}
    \end{align*}    
    \Cref{TRC_AST_output_binding} shows the corresponding $\systemx$ representations together with their AST.
\end{example}

\begin{figure}[t]
\centering
\begin{subfigure}[b]{0.25\linewidth}
    \centering
    \includegraphics[scale=0.42]{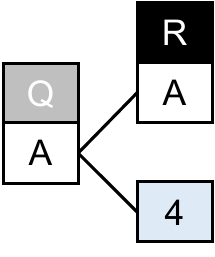}
    \vspace{1.5mm}
    \caption{}
    \label{TRC_AST_output_binding_1}
\end{subfigure}	
\hspace{1mm}
\begin{subfigure}[b]{0.25\linewidth}
    \centering
    \includegraphics[scale=0.42]{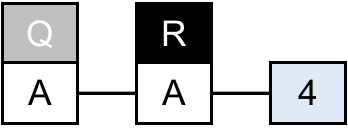}
    \vspace{1.5mm}
    \caption{}
    \label{TRC_AST_output_binding_2}
\end{subfigure}	
\hspace{1mm}    
\begin{subfigure}[b]{0.25\linewidth}
    \centering
    \includegraphics[scale=0.42]{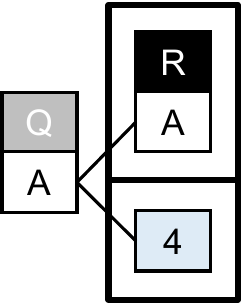}
    \caption{}
    \label{TRC_AST_output_binding_3}
\end{subfigure}	
\begin{subfigure}[b]{0.25\linewidth}
    \centering
    \includegraphics[scale=0.42]{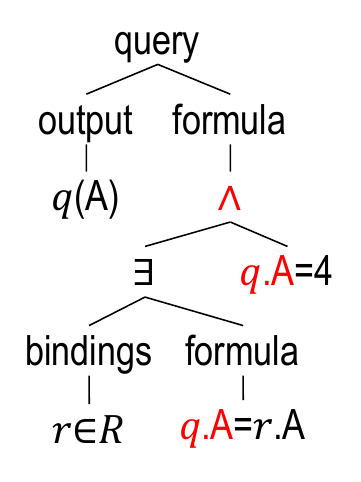}
    \caption{}
    \label{TRC_AST_output_binding_4}
\end{subfigure}	
\hspace{1mm}
\begin{subfigure}[b]{0.25\linewidth}
    \centering
    \includegraphics[scale=0.42]{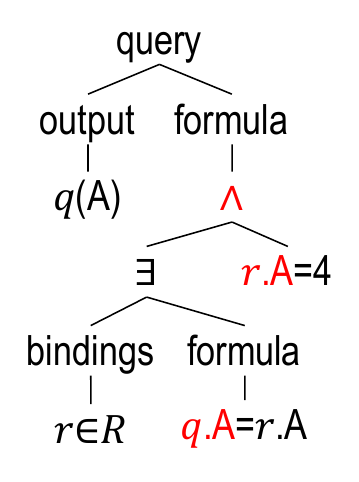}
    \caption{}
    \label{TRC_AST_output_binding_5}
\end{subfigure}	
\hspace{1mm}    
\begin{subfigure}[b]{0.25\linewidth}
    \centering
    \includegraphics[scale=0.42]{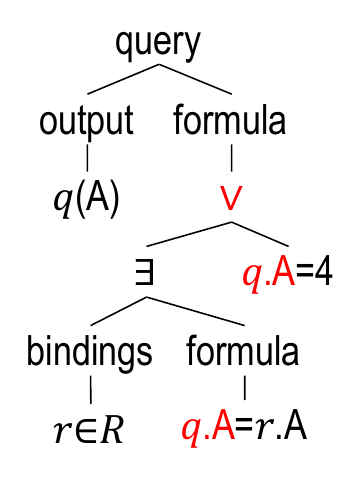}
    \caption{}
    \label{TRC_AST_output_binding_6}
\end{subfigure}	
\caption{\Cref{ex:TRC_AST_output_binding}: Three  well-formed $\TRC$ queries
together with their corresponding ASTs.
Only (b) and (c) are safe according to our safety condition.}
\label{TRC_AST_output_binding}
\end{figure}

\subsubsection{Comparison with Ullman's safety criteria}
Ullman's safety criteria~\cite[Section 3.9]{Ullman1988PrinceplesOfDatabase} 
are notably more restrictive than ours in an important way.
They not only disallow universal quantifiers,
they also include a requirement that 
all free variables need to be bound to a domain
\emph{for all existentially quantified subformulas},
such as ``$\exists r_1 \in R_1, \ldots, r_k \in R_k[\varphi]$''.
This requirement is motivated by the fact that any such safe formula has a direct and natural translation into $\Datalog$~\cite[Theorem 3.11]{Ullman1988PrinceplesOfDatabase} and thus also $\RA$.
However, this restriction forbids nested correlated queries
that are allowed by our safety condition and also by $\SQL$.

For example, consider the following subformula of query 
\cref{query:safety_continued}
in
\cref{ex:safety_continued}:
$$
    \neg(
    \exists s\in S[\neg (\exists r_2 \in R[r.A\!=\!r_2.A \wedge r_2.B\!=\!s.B ])]
    )
$$
In Ullman's formalism and our notation, it would be written as:
$$
    \neg(
    \exists s[s \in S \wedge \neg (\exists r_2[r_2 \in R \wedge r.A\!=\!r_2.A \wedge r_2.B\!=\!s.B ])]
    )
$$
The free variable $r$ violates Ullman's safety condition
in the maximal conjunction  
``$s \in S \wedge \neg(\ldots)$''
because $r$ is free but not ``limited'', 
as it only appears in the 2nd conjunct
``$\neg (\exists r_2[\ldots])$''
which is negated.

As a consequence, Ullman's safe queries \emph{cannot express certain relational query patterns that our safety condition allows}.
On the other hand, Ullman's conditions allow using the free variables multiple times in arbitrary ways in a formula, 
which makes the inductive translation from $\RA$ shorter, yet does not add any pattern expressiveness.

\subsubsection{Safety and DeMorgen}
We give two examples of the fact that safety for disjunctions is treated differently than 
their DeMorgan equivalent of negation combined with conjunction, across various definitions of safety. 
Interestingly both directions are possible: 
queries with disjunction can be safe (whose DeMorgan equivalent with negation and conjunction is not),
but queries with disjunctions can also be unsafe (yet they become safe after applying DeMorgan).

(1) For the syntactic criterion of ``range-restricted queries'' by Abiteboul et al.~\cite[Sect.~5.4]{DBLP:books/aw/AbiteboulHV95}, 
part of the inductive definition for the set of range-restricted variables ($\mathit{rr}$) is given by
\begin{align*}
    \varphi_1 \wedge \varphi_2  
        \; : \;\;
        & \mathit{rr}(\varphi) = \mathit{rr}(\varphi_1) \cup \mathit{rr}(\varphi_2)  \\
    \varphi_1 \vee \varphi_2 
        \; : \;\;
        & \mathit{rr}(\varphi) = \mathit{rr}(\varphi_1) \cap \mathit{rr}(\varphi_2)  \\        
    \neg \varphi_1  
        \; : \;\;
        & \mathit{rr}(\varphi) = \emptyset 
\end{align*}    
It follows that any safe query defined over a formula 
$\varphi_1 \vee \varphi_2$ 
immediately becomes unsafe after applying DeMorgan and writing it as the pattern-equivalent formula
$\neg(\neg \varphi_1 \wedge \neg \varphi_2)$.

(2) Example~3.29 by Ullman~\cite{Ullman1988PrinceplesOfDatabase} shows the $\DRC$ formula 
$$
    r(x, y, z) \wedge \neg(p(x, y) \vee s(y, z)) \,
$$
that is not ``safe'' according to Ullman's definitions of safety, 
because the subformula
$p(x, y) \vee s(y, z)$
violates the syntactic safety condition that
whenever an $\vee$ operator is used, the two formulas connected must have the same set of free variables.
However by applying DeMorgan on the negated subformula and then ``flattening'' 
(i.e.\ no $\wedge$ node appears as a child of another $\wedge$ in the AST) 
he derives the pattern-equivalent formula
$$
    r(x, y, z) \wedge \neg p(x, y) \wedge \neg s(y, z)
$$
which is safe because all three variables in the conjunction are limited by the positive conjunct $r(x, y, z)$.

\subsubsection{Proof of correctness of our safety definition}
For completeness, we give a self-contained proof that our safety condition is a valid syntactic restriction that preserves relational completeness.

\begin{lemma}
\label{lemma:safety}
The four conditions given in \cref{sec:safety} are a valid safety restriction, 
i.e.\ the resulting safe $\TRC$ fragment is relationally complete and every safe $\TRC$ query is domain-independent.
\end{lemma}

\begin{proof}[Proof \cref{lemma:safety}]
It suffices to prove two steps:
(1) The safe fragment is relationally complete. We show that by giving a constructive translation of any relational algebra query into that fragment.\footnote{We decided to prove (1) from first principles by translating from $\RA$ 
in order to make our results as self-contained as possible. We also
we did not find a self-contained and easy-to-use safety condition for $\TRC$ in standard references:
Existing definitions of safety for $\TRC$ are either too restrictive 
(e.g.\ Ullman's definition of safety of $\TRC$~\cite{Ullman1988PrinceplesOfDatabase} is the most carefully developed we 
found, but prohibits the use of universal quantification, and it does not even allow nested correlated queries),
or only explain but do not formally define safety 
(e.g.\ \cite{DBLP:books/mg/SKS20,
cowbook:2002,
Elmasri:dq}),
or do not treat $\TRC$ and its formaly safety conditions at all (e.g.\ \cite{DBLP:books/aw/AbiteboulHV95,ArenasBLMP:DatabaseTheory,DBLP:journals/tods/GelderT91}).}
(2)~Every query expressed in that fragment is domain-independent.

\ul{(1) $\RA \rightarrow$ safe $\TRC$}:
The proof for this direction 
is an adaptation of the proof of Lemma 3.6 in~\cite{Ullman1988PrinceplesOfDatabase}.
We show by induction on the number of operators in the relational algebra expression $E$ 
that there is a $\TRC$ formula, 
with a single free tuple variable $q$, that defines the same relation as $E$. 

The basis, zero operators, requires that we consider two cases: 
$E$ is either a database relation $R$ or (a relation of) constants. 
If $E$ is a relation $R$ of arity $k$ and schema $(A_1, \ldots, A_k)$ 
then formula 
$$
\exists r \in R[q.A_1=r.A_1 \wedge \ldots \wedge q.A_k=r.A_k]
$$ 
suffices and is safe by our definition.
If $E$ is a
constant relation, say $\{r_1, \ldots, r_n\}$
with header $\vec H = (A_1, \ldots, A_k)$,
consisting of $n$ tuples of arity $k$, 
then we write the $\TRC$ formula
\begin{align*}
    &(q.A_1 = r_1.A_1 \wedge q.A_2 = r_1.A_2 \wedge \ldots \wedge q.A_k = r_1.A_k) \,\vee \\
    &(q.A_1 = r_2.A_1 \wedge q.A_2 = r_2.A_2 \wedge \ldots \wedge q.A_k = r_2.A_k) \,\vee \\
    &\ldots  \\
    &(q.A_1 = r_n.A_1 \wedge q.A_2 = r_n.A_2 \wedge \ldots \wedge q.A_k = r_n.A_k)     
\end{align*}    
Notice that $r_i.A_j$ is a constant
for all $i \in [n]$ and $j\in [k]$. 
Thus, the above formula is safe by our safety definition.

For the induction, consider an expression whose outermost operator is one of 5 operators: 
Union $\cup$,
difference $-$,
projection $\pi$,
Cartesian product $\times$, 
or selection $\sigma$.\footnote{Notice that we do not need to treat the renaming operator separately:
WLOG we assume that every tuple variable is defined only once in any safe $\TRC$ query, 
and we thus never run into renaming issues in the named perspective of  $\TRC$ as we do in the named perspective of $\RA$.}

Case 1: $E = E_1 \cup E_2$: 
From the induction, we assume that there are safe $\TRC$ formulae $\varphi_1$ and $\varphi_2$ for $E_1$ and $E_2$, respectively. 
We know from the definition of the union operator that $E_1$ and $E_2$ must have the same arities and schemas.
By renaming if necessary, we may assume that the lone free tuple variable in both 
$\varphi_1$ and $\varphi_2$ is $q$. 
Because $E_1$ and $E_2$ have the same arity, 
the free tuple variables of $\varphi_1$ and $\varphi_2$ must also have the same arity, thus renaming is permitted. 
Then $\varphi_1 \vee \varphi_2$ is a safe $\TRC$ formula for E. 

Case 2: $E = E_1 - E_2$: 
As in case 1, assume there are formulas $\varphi_1$ and $\varphi_2$ 
for $E_1$ and $E_2$ with free variable $q$, respectively. 
We know by definition of the set difference that $E_1$ and $E_2$ must have the same arities and schemas.
From the induction, we know that for each attribute $q.A_i$, 
both $\varphi_1$ and $\varphi_2$
contain binding predicates of the form 
$q.A_i = x_{1i}$ and 
$q.A_i = x_{2i}$, respectively
(here each $x_{ji}$ is either a constant or 
an attribute in an existentially quantified relation).
Let $\varphi_2'$ be the formula we get from $\varphi_2$
by replacing every binding predicate 
$q.A_i = x_{2i}$
by 
$x_{1i} = x_{2i}$.
Then $\varphi_1 \wedge \neg \varphi_2'$ 
is a safe $\TRC$ formula for $E$ and its free variables only appear in the non-negated subformula $\varphi_1$.

Case 3: $E = \pi_{A_{i1}, \ldots, A_{id}}(E_1)$:
Let $\varphi_1$ be a safe $\TRC$ formula for $E_1$ with a single free variable $q_1$ of arity $k, k>d$.
From the induction, we know that $\varphi_1$ 
contains binding predicates of the form $q_1.A_i = x_{i}$ for $x \in [k]$.
Then a formula for $E$, with free variable $q$ of arity $d$,
is identical to $\varphi_1$
where every binding predicate 
$q_1.A_i = x_{i}$ for $i\in[d]$ is replaced by 
$q.A_i = x_{i}$, 
and every binding predicate for $i \in \{d+1, \ldots, k\}$ is removed.

Case 4: $E = E_1 \times E_2$: 
Let $\varphi_1$ and $\varphi_2$ be two $\TRC$ formulas for $E_1$ and $E_2$ 
with 
free variables $q_1$ and $q_2$ of arities $k_1$ and $k_2$, respectively.
As before, we assume $q_1$ and $q_2$ are bound 
with binding predicates $q_1.A_i = x_{1i}, i \in [k_1]$ and 
$q_2.A_i = x_{2i}, i \in [k_2]$, respectively.
If necessary, we first rename the attributes of $q_1$ and $q_2$ in their binding predicates s.t.\ they are disjoint.
WLOG
$q_1$ has then attributes $A_1, \ldots, A_{k_1}$
and $q_2$ has attributes $A_{k_1+1}, \ldots, A_{k_1+k_2}$.
We then rename the tuple variables themselves in the binding predicates of $\varphi_1$ and $\varphi_2$ s.t.\ they use the same variable $q$ 
(i.e.\ $q.A_i = x_{1i}, i \in [k_1]$ in $\varphi_1$ and 
$q.A_{i+k_1} = x_{2i}, i \in [k_2]$ in $\varphi_2$).
The $\TRC$ formula for $E$, with lone free variable $q$ of arity $k_1 + k_2$, is
$\varphi_1 \wedge \varphi_2$.
Finally, we assume that whenever possible, we reformulate the resulting formula $\varphi$ to be maximally scoped 
(i.e.\ we push the $\wedge$ behind existential quantifiers).

Case 5: $E = \sigma_c (E_1)$:
Here, $c$ is a logical condition that combines join and selection predicates defined over attributes from $E_1$ 
with an arbitrary nesting of logical operators $\wedge, \vee, \neg$.
Let $\varphi_1$ be a safe $\TRC$ formula for $E_1$ with single free variable $q$,
and let $c'$ be $c$ in which references to tables in $E_1$ 
are replaced with appropriate tuple variables from $\varphi_1$.
Then a formula for $E$, with free variable $q$
is $\varphi_1 \wedge c'$.

\ul{(2) Every safe $\TRC$ query is domain-independent}:
This statement follows immediately from our definition of safety and observing that every occurrence of the single free variable $q$ is possible only via a binding predicate,
i.e.\ a predicate that binds an attributes of $q$ to either column of a table or a constant, 
and thus to a finite set of domain values.
\end{proof}

\section{Details for \cref{sec:firstsolution}}

\subsection{Proof \cref{lemma:ENC-TRC} (\cref{sec:ENC-TRC})}
\label{sec:proof-lemma-atom-preserving}

\begin{proof}[Proof \cref{lemma:ENC-TRC}]
    The proof follows from a recursive application of 
    standard logical transformations~\cite[Sections 3.6 \& 10.3]{Barker-Plummer:2011} 
    and observing that they do not change the atoms of the AST, and thus neither the predicates nor the table signature.
    
    Let $\varphi_1, \ldots, \varphi_k$ represent any well-formed formulas. Then the following transformations preserve the atoms:
\begin{enumerate}

    \item 
    DeMorgan's law for universal quantifiers:\footnote{DeMorgan's laws are named after Augustus de Morgan who formalized them (though they have been known earlier).
    We are using the notation of Barker-Plummer et al~\cite{Barker-Plummer:2011} and write ``DeMorgen'' without a space between `De' and `Morgan'.}  
    $(\forall r_1 \in R_1, \ldots,$ $\forall r_k \in R_k [\varphi]) 
    = \neg(\exists r_1 \in R_1, \ldots, \exists r_k \in R_k [\neg(\varphi)])$.
    Notice here that 
    the binding atoms ($r_i \in R_i$) are staying the same
    despite the quantifiers are changing.

    \item 
    Removing implications: 
    $((\varphi_1) \rightarrow (\varphi_2))
    = (\neg(\varphi_1) \vee (\varphi_2))$

    \item 
    DeMorgan's law for disjunctions: 
    $(\varphi_1 \vee \varphi_2 \vee \dots \vee \varphi_k)
    = \neg( \neg(\varphi_1) \wedge \neg(\varphi_2) \wedge \dots \wedge \neg(\varphi_k))$

    \item 
    DeMorgan's law for conjunctions:
    $(\varphi_1 \wedge \varphi_2 \wedge \dots \wedge \varphi_k)
    = \neg( \neg(\varphi_1) \vee \neg(\varphi_2) \vee \dots \vee \neg(\varphi_k))$

    \item 
    Removing double negations: 
    $\neg(\neg(\varphi_1)) 
    = (\varphi_1)$

\end{enumerate}
The transformations can be applied in order $(1) \rightarrow (2) \rightarrow (3) \rightarrow (5)$ 
from the outside inwards
and thus finish in polynomial time.
\end{proof}

\subsection{Comparing the ASTs for $\TRC^{\neg \exists \wedge \vee}$ and $\TRC^{\neg \exists \wedge}$ (\cref{sec:non-dijunctive-fragment})}
\label{sec:non-dijunctive-fragment-illustration}

\begin{example}
\label{ex:Fig_disjunction_bigger_AST}
Consider the following $\TRC$ query
that is a variation on relational division. It returns values from $R.A$
that co-occur in $R$ with either $S.B$ or $S.C$, for all tuples from $S$ with $S.A\!>\!0$:
\begin{align}
&\{ q(A) \mid \exists r \in R [
    q.A \equal r.A \wedge \h{(\forall} s \in S[s.A\!>\!0 \;\h{\rightarrow}\;
    \label{eq:Fig_disjunction_bigger_AST:1}\\
&\hspace{4mm}\h{(}\exists r_2 \in R [ \h{(}r_2.B \equal s.B \;\h{\vee}\; r_2.C \equal s.C \h{)} 
    \wedge r_2.A \equal r.A ] \h{ )} ] \h{)}] \}
    \notag
\intertext{Replacing $\forall$ and $\rightarrow$ leads to an expression in
$\TRC^{\neg \exists \wedge \vee}$:}
&\{ q(A) \mid \exists r \in R [
    q.A \equal r.A \wedge \h{\neg(\exists} s \in S[s.A\!>\!0 \;\h{\wedge}\;     \label{eq:Fig_disjunction_bigger_AST:2}\\
&\hspace{4mm}\h{\neg (}\exists r_2 \in R [ \h{(}r_2.B \equal s.B \;\h{\vee}\; r_2.C \equal s.C \h{)} 
    \wedge r_2.A \equal r.A ] \h{ )} ] \h{)}] \}
    \notag
\intertext{Further replacing $\vee$ leads to a query in $\TRC^{\neg \exists \wedge}$:}
&\{ q(A) \mid \exists r \in R [
    q.A \equal r.A \wedge \h{\neg(\exists} s \in S[s.A\!>\!0 \;\h{\wedge}\;     
    \label{eq:Fig_disjunction_bigger_AST:3}\\
&\hspace{4mm}\h{\neg(}\exists r_2 \in R [ 
    \h{\neg(\neg(}r_2.B \equal s.B\h{)} \;\h{\wedge}\; \h{\neg(} r_2.C \equal s.C \h{))} 
    \wedge r_2.A \equal r.A ] \h{ )} ] \h{)}] \}
    \notag
\end{align}
Notice that the above three expressions are not only pattern-equivalent
(we do not change the signature),
but also all the selection predicates (e.g.\ ``\,$S.A\!>\!0$'')
and the join predicates (e.g.\ ``\,$r_2.C \equal s.C$'') are identical.
Only the logical connectives ($\neg, \vee, \wedge, \rightarrow$), quantifiers and parentheses may have changed.
In other words, \emph{all atoms} are identical, and hence all leaves in their ASTs
(\cref{Fig_disjunction_bigger_AST}).

This is in contrast to the non-disjunctive fragment required by $\diagrams$. 
Concretely,
\cite{DBLP:journals/pacmmod/GatterbauerD24} suggests to use DeMorgan with quantifiers 
as to distribute the quantifier over the disjuncts, before transforming to a conjunction:
$\neg (\exists r \in R[A \vee B]) = 
\neg (\exists r \in R[A] \vee \exists r \in R[B]) =
\neg \exists r \in R[A] \wedge \neg \exists r \in R[B]$.
This leads to a disjunction-free query, 
yet leads to 
a different signature 
(3 occurences of table ``\,$R$\!'')
and thus a different relational pattern:
\begin{align}
&\{ q(A) \mid \exists r \in R [
    q.A \equal r.A \wedge \h{\neg (\exists} s \in S[s.A \!>\! 0 \;\h{\wedge} \\
&\hspace{4mm}\h{\neg (\exists} r_2 \in R [r_2.B \equal s.B \wedge r_2.A \equal r.A ]\h{)} \;\h{\wedge}\; 
    \notag\\
&\hspace{4mm}\h{\neg (\exists} r_3 \in R [r_3.C \equal s.C \wedge r_3.A \equal r.A ]\h{)} ])] \} 
    \hspace{29mm}\notag	
\end{align}
\end{example}

\begin{figure}[t]
    \centering
    \begin{subfigure}[b]{0.3\linewidth}
        \centering
        \includegraphics[scale=0.37]{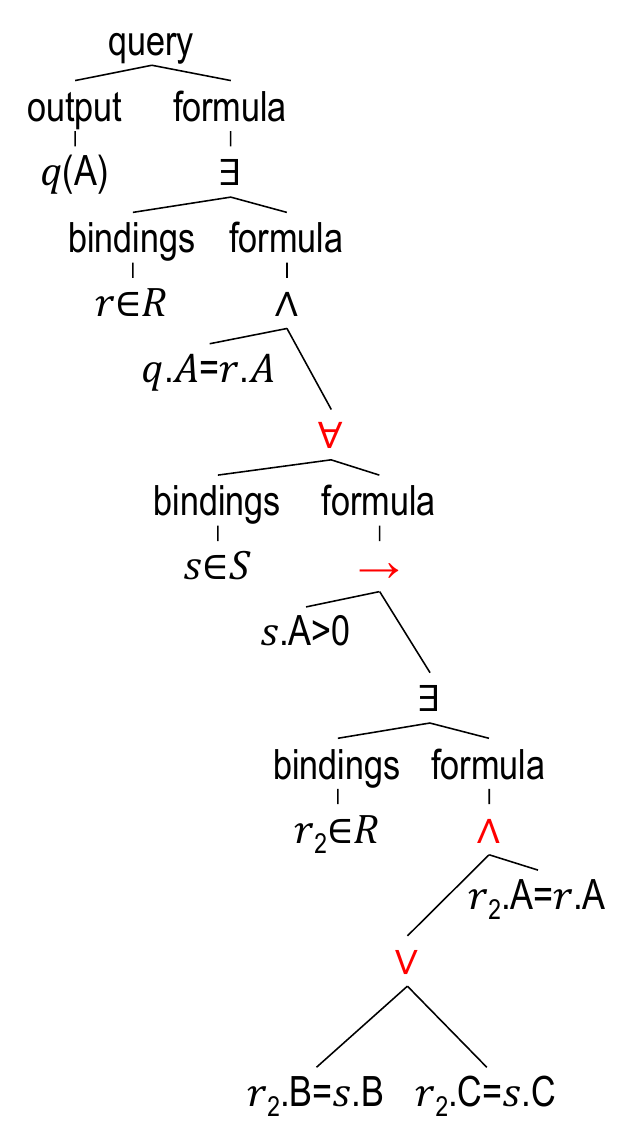}
        \caption{$\TRC$ query \cref{eq:Fig_disjunction_bigger_AST:1}}
        \label{Fig_disjunction_bigger_AST_1}
    \end{subfigure}	
    \hspace{4mm}
    \begin{subfigure}[b]{0.3\linewidth}
        \centering
        \includegraphics[scale=0.37]{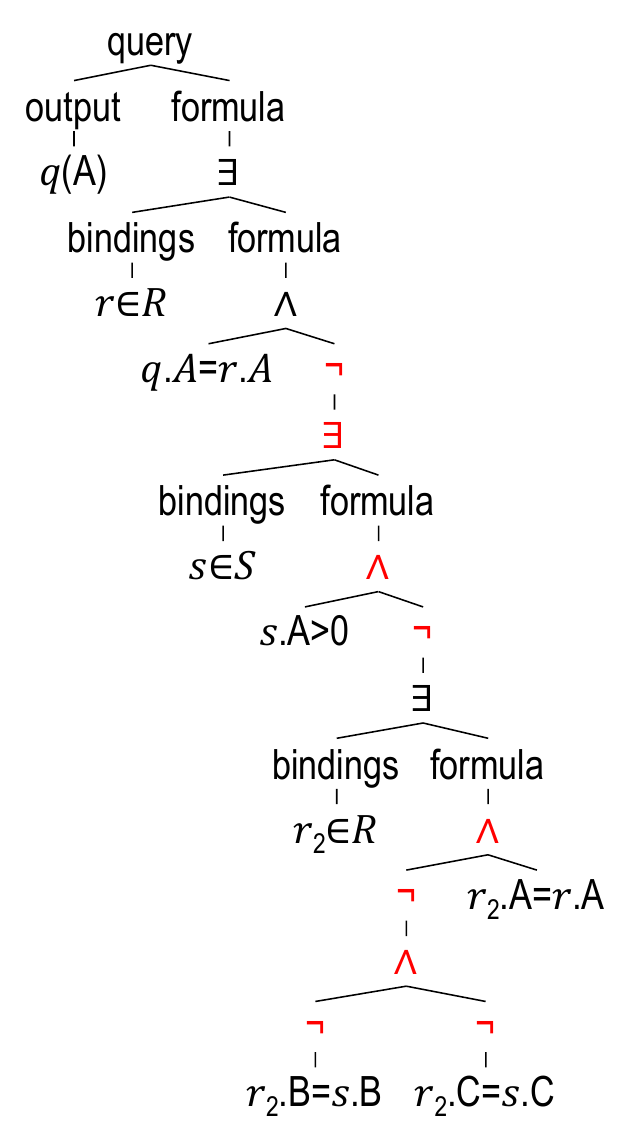}
        \caption{$\TRC$ query \cref{eq:Fig_disjunction_bigger_AST:3}}
        \label{Fig_disjunction_bigger_AST_2}
    \end{subfigure}	
    \caption{\cref{ex:Fig_disjunction_bigger_AST}:
    ASTs of logically identical queries.
    Notice how all leaves of the trees are identical despite (b) using a smaller vocabulary for internal nodes than (a).}
    \label{Fig_disjunction_bigger_AST}
\end{figure}

\subsection{$\diagrams$ with built-in relations are pattern-complete for $\TRC$ (\cref{sec:Theorem-built-in-relations},\cref{th:pattern expressiveness})}

We next prove \cref{th:pattern expressiveness} by giving a constructive translation from 
$\TRC^{\neg \exists \wedge}$ and back.

\subsubsection{Translation from $\TRC^{\neg \exists \wedge}$ to $\diagrams$ with built-in relations}
\label{sec:translation from TRC}
We next give the straight-forward 4-step translation from well-formed $\TRC^{\neg \exists \wedge}$ queries into a variant of $\diagrams$ enhanced with built-in predicates. 
This translation is heavily inspired by the 5-step translation given in~\cite{DBLP:journals/pacmmod/GatterbauerD24}
for their version of $\diagrams$, however it differs in crucial steps.
Notably, by showing that our translation preserves the atoms (which includes the binding predicates)
for \emph{all} well-formed $\TRC$ queries,
we also show that our variant can express all relational patterns of $\TRC$.
We illustrate the translation with query \cref{eq:Fig_disjunction_bigger_AST:3} 
from \cref{ex:Fig_disjunction_bigger_AST} 
displayed in \cref{Fig_disjunction_bigger_builtin}.

\begin{figure}[t]
    \centering
    \includegraphics[scale=0.4]{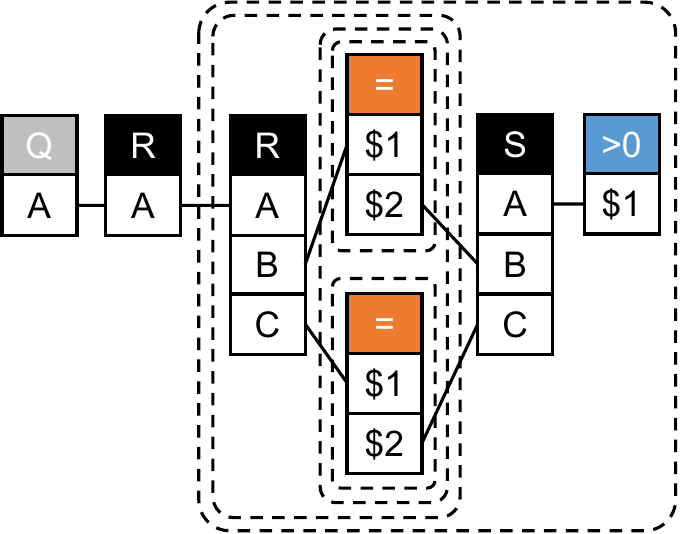}
    \caption{\Cref{sec:translation from TRC}: Pattern-preserving diagrammatic representation with built-in relations for query \cref{eq:Fig_disjunction_bigger_AST:3} 
    from \cref{ex:Fig_disjunction_bigger_AST}.}
    \label{Fig_disjunction_bigger_builtin}
\end{figure}

\emph{(1) Preprocessing}:
First, we translate any well-formed $\TRC$ into $\TRC^{\neg \exists \wedge}$ by preserving its atoms
as described in \cref{sec:ENC-TRC}.
Then, we replace every join and selection predicate with the corresponding built-in relation as described in \cref{sec:anchors}.
Equijoin predicates ($R.A = S.B$) that occur in the same negation scope as one of their relations $R$ or $S$
do not need to be replaced 
(e.g.\ this is the case in \cref{Fig_new_anquors_3} after replacing the ``$<$'' operator with ``$=$'',
but not in 
\cref{Fig_Disjunction_solutions_2} and
\cref{ex:Fig_disjunction_bigger_AST_continued}).
As example, query \cref{eq:Fig_disjunction_bigger_AST:3} is written as:
\begin{align*}
    &\{ q(A) \mid \exists r \in R [
        q.A \equal r.A \wedge \neg(\exists s \in S \h{,c \in \anchor{\textrm{``$>\!0$''}} }
        [s.A \equal c.\$1  \,\wedge     \\
    &\hspace{4mm}\neg(\exists r_2 \in R [ 
        \neg(\neg( \h{\exists j_1 \in \anchor{\textrm{``$=$''}}} 
        [r_2.B \equal j_1.\$1 \wedge j_1.\$2 \equal s.B] ) \,\wedge  \\
    &\hspace{8mm} \neg( \h{\exists j_2 \in \anchor{\textrm{``$=$''}}} 
        [r_2.C \equal j_2.\$1 \wedge j_2.\$2 \equal s.C )) 
        \wedge r_2.A \equal r.A ]  ) ] )] \}
\end{align*}
Notice that the equijoins ``$q.A \equal r.A$'' and ``$r_2.A \equal r.A$'' are not replaced with built-in relations

\emph{(2) Creating canvas partitions}:
Similar to $\diagrams$,
The scopes of the negations in a $\TRC$ are nested by definition.
We translate the nested hierarchy of negation scopes (the \emph{negation hierarchy}) 
into a nested partition of the canvas by dashed boxes with rounded corners.

\emph{(3) Placing input, built-in, and output relations}:
Relation names from each binding predicate are placed into the canvas partition 
that corresponds to the respective negation scope.
Notably, built-in relations are treated identically as other relations. 
In particular, multiple occurrences of the same built-in relation (e.g.\ $<$) lead to separate relations being placed in appropriate partitions.
For our example, the binding 
$\exists j_1 \in \anchor{\textrm{``$=$''}}$
is represented by a built-in relation
\fcolorbox{black}{binaryBuiltInPredicate}{\textcolor{white}{$=$}}
in nesting level 4 of the negation hierarchy.
Clearly inspired by $\diagrams$, the output relation $Q$ for non-Boolean queries and all attributes from its headers are represented with a relation
\selectBox{} 
outside all nesting scopes, which we refer to as the \emph{base partition} in reference to the AST.
Recall that for well-formed $\TRC$ queries $\{q(\vec H) \mid \varphi\}$, $\varphi$ can contain only one single free variable $q$.

\emph{(4) Placing equijoin predicates}:
In contrast to \diagrams, both join and selection predicates 
are now treated identically as equijoins with appropriate built-in relations.
For each equijoin $R.A \equal S.B$ in the query, we simply add two attributes, one for each relation $R$ and $S$ and connect them via an unlabeled edge.
For table attributes that occur in multiple equijoins, we only draw one attribute and connect it to multiple edges.
Notice that by this construction, the attributes of each occurrence of a built-in relation are connected to exactly one edge.
For example, 
``$r_2.B \equal j_1.\$1$''
is represented with an edge between attribute \attributeBox{B} of 
the 2nd occurence of relation \tableBox{R} 
with attribute \attributeBox{\$1} of the first occurence of built-in relation 
\fcolorbox{black}{binaryBuiltInPredicate}{\textcolor{white}{$=$}}.

\introparagraph{Completeness}
This 4-step translation guarantees the uniqueness of the following aspects: (1) nesting hierarchy (corresponding to the negation hierarchy), (2) where input and built-in relations are placed (canvas partitions corresponding to the negation scope), (3) which relation attributes participate in equijoins with what other relational attributes.
Notice that due to our unified treatment of selection and join predicates as equi-joins with appropriate built-in relations, our translation (after preprocessing) is slightly simpler than the one originally proposed by the authors of $\diagrams$~\cite{DBLP:journals/pacmmod/GatterbauerD24} and, 
more importantly, also more general: 
\emph{Any well-formed $\TRC$ query} can be represented in a way that preserves all atoms and thus also relational patterns.

\subsubsection{Translation from $\diagrams$ with built-in relations to $\TRC^{\neg \exists \wedge}$}
\label{sec:translation back to TRC}
We next describe the reverse 5-step translation from any valid $\diagrams$ with built-in predicates 
to a valid and unique $\TRC^{\neg \exists \wedge}$ expression.

\emph{(1) Determine the negation hierarchy}:
From the nested canvas partitions, 
create the nested scopes of the negation operators of the later $\TRC^{\neg \exists \wedge}$ expression, connected via conjunction whenever a partition has multiple children in the negation hierarchy.
For \cref{Fig_disjunction_bigger_builtin}, this interpretation leads to 
$\varphi = \neg(\neg(\neg (\neg(\_) \wedge \neg(\_) ) ) )$.
Here and in the following ``$\_$'' stands for a placeholder.

\emph{(2) Optional output table}:
If the diagram contains an output relation $Q$ with non-empty header $\vec H$ in the base partition
(which represents the only free variable), 
then use the set builder notation
$\{q(A) \mid \varphi )\}$.

\emph{(3) Binding predicates for table variables}:
Each input relation in a partition corresponds to an existentially-quantified table variable. 
WLOG, we use a small letter indexed by the number of occurrences for repeated tables.
We add those quantified table variables in the respective scope of the negation hierarchy.
Similar to $\diagrams$, we require that the leaves of the partition are not empty and contain at least one relation (whether input or built-in).
For our running example, we now have:
$\{q(A) \mid \exists r_1 \in R
    [\neg(\exists s_1 \in S
    [\neg(\exists r_2 \in R[ \neg (\neg(\_) \wedge \neg(\_) ) ] ) ] ) ] \}$

\emph{(4) Join and selection predicates for built-in relations}:
For each unary built in relation $\theta c$, we place a selection predicate template $\_ \theta c$
into the appropriate negation scope (here ``$\_$'' is a placeholder for table attribute to be added later).
For each binary built in relation $\theta$, we place a join predicate template 
``$\_ \theta \_$''
into the appropriate negation scope.
Whenever two table attributes $R.A$ and $S.B$ are directly connected via an edge without any built-in predicate in between, 
then this because an equijoin predicate 
``$r_i.A=s_j.B$'' occurs in the same negation scope as one of their relations $R$ or $S$. 
In this case, we place an equijoin predicate template ``$\_ = \_$''
into the negation scope of either $R$ or $S$, whichever one appears most deeply nested (they may also be equally nested).
In our running example, this is the case for the equijoin between the two $R$ tables
(``$r_1.A \equal r_2.A$'') being in the same partition as $r_2$.
Multiple such join and selection predicates are connected via conjunctions within the same negation scope.
At this point, our formula template is:
\begin{align*}
    &\{q(A) \mid  \exists r_1 \in R [
        \_ \equal \_ \wedge \neg(\exists s_1 \in S[\_\!>\!0 \;\wedge     \\
    &\hspace{4mm} \neg(\exists r_2 \in R [ 
        \neg(\neg(\_ \equal \_ ) \wedge \neg( \_ \equal \_ )) 
        \wedge \_ \equal \_ ] ) ] )] \}
\end{align*}

\emph{(5) Replace edges with table attributes}:
For each edge between a table attribute and a built-in predicate, 
we add the table attribute reference to the appropriate predicate template.
For example, 
``$\_\!>\!0$''
gets replaced by 
``$s_2\!>\!0$''
and the second occurrence of
``$\_ \equal \_$''
gets replaced by
``$r_2.B \equal s_1.B$''.
Our final $\TRC$ query is then:
\begin{align*}
    &\{q(A) \mid \exists r_1 \in R [
        q.A \equal r_1.A \wedge \neg(\exists s_1 \in S[s_1.A\!>\!0 \;\wedge     \\
    &\hspace{4mm} \neg(\exists r_2 \in R [ 
        \neg(\neg(r_2.B \equal s_1.B ) \wedge \neg( r_2.C \equal s_1.C )) 
        \wedge r_2.A \equal r_1.A ] ) ] )] \}
\end{align*}

\introparagraph{Soundness} 
Notice that this 5-step translation guarantees that the resulting $\TRC$ query is well-formed and uniquely determined up to 
(1) renaming of the tuple variables; 
(2) reordering the predicates in conjunctions (commutativity of conjunctions), and 
(3) flipping the left/right positions of attributes in each predicate. 
Neither of these possible transformations changes the semantics of the query.
It follows that \diagrams~with built-in predicates are sound, and their logical interpretation is unambiguous.

\subsubsection{Valid \diagrams~with built-in relations}
\label{sec:RD_validity}
The validity conditions for $\diagrams$ with built-in relations
are similar to $\diagrams$ and follow from the conditions for each of the 5
steps of the translation to \TRC. The differences are:
(1)~Edges contain no labels and represent equijoins.
(2) Built-in predicates can be of two types (unary and binary) and have one of 6 comparison operators.
(3) Edges can only happen between attributes of tables that are descendants of each other, and only between input relations, or between one input relation and one built-in predicate.
(4) The optional single output table needs to be outside any negation scope and disjunction box, 
but each attribute can be connected to an arbitrary number of tables in any partition.
The constructive translations from 
\cref{sec:translation from TRC,sec:translation back to TRC},
which preserves \emph{all atoms} in both directions, constitute the proof.

\section{Details for \cref{sec:real solution}}

\subsection{Example for \cref{sec:simpler:anchors}}
\label{app:simpler:anchors}

\begin{example}[Canonical selections]
    \label{ex:selection_conditions}
    Consider the Boolean query
    $$
        \exists r\in R, \exists s\in S[r.A \equal s.B \wedge r.A \!>\!0 \wedge r.A \!<\!5]
    $$
    shown in 4 variants in \cref{Fig_selection_conditions}:
    \Cref{Fig_selection_conditions_1} is our canonical representation and is topologically uniquely determined.
    It defines our formal semantics and is necessary to represent pattern-equivalent nested disjunctions.
    \Cref{Fig_selection_conditions_2} is the canonical representation for \diagrams. It repeats the attribute of a table for every selection predicate and fuses the selections into the attributes. We apply this representation as a visual shortcut whenever possible.
    \Cref{Fig_selection_conditions_3} is a possible further shortcut.
    The choice of the selection attribute that is fused into the join is non-deterministic, which is why we do not recommend it.
    \Cref{Fig_selection_conditions_4} uses a symbolic representation for conjunction $\wedge$ and is thus not diagrammatic (see also \cref{Fig_Prior_Disjunctions_1}).
\end{example}

\begin{figure}[t]
    \centering
    \begin{subfigure}[b]{0.15\linewidth}
        \centering
        \includegraphics[scale=0.4]{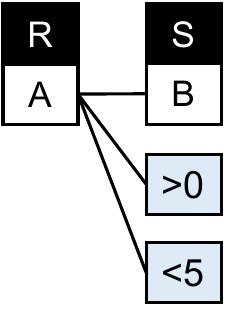}
        \caption{}
        \label{Fig_selection_conditions_1}
    \end{subfigure}	
    \hspace{6mm}
    \begin{subfigure}[b]{0.14\linewidth}
        \centering
        \includegraphics[scale=0.4]{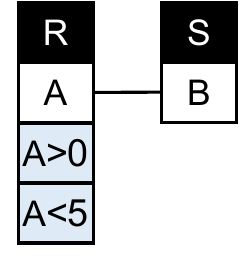}
        \vspace{4mm}
        \caption{}
        \label{Fig_selection_conditions_2}
    \end{subfigure}	
    \hspace{6mm}    
    \begin{subfigure}[b]{0.14\linewidth}
        \centering
        \includegraphics[scale=0.4]{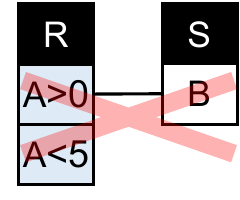}
        \vspace{8mm}
        \caption{}
        \label{Fig_selection_conditions_3}
    \end{subfigure}	
    \hspace{6mm}    
    \begin{subfigure}[b]{0.18\linewidth}
        \centering
        \includegraphics[scale=0.4]{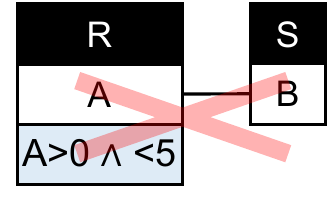}
        \vspace{8mm}
        \caption{}
        \label{Fig_selection_conditions_4}
    \end{subfigure}	    
    \caption{\cref{ex:selection_conditions} in \cref{app:simpler:anchors}: (a) and (b) are two alternative representations for $\systemx$. 
    (a) is ``canonical'', defines our formal semantics, and is necessary to represent pattern-equivalent nested disjunctions. 
    (b) is the representation of $\diagrams$ that we use as short-form whenever possible.
    We do not permit (c) which is non-deterministic (and thus somewhat arbitrary),
    nor (d) which requires a symbolic interpretation of $\wedge$ and is thus not diagrammatic.
    }
    \label{Fig_selection_conditions}
\end{figure}

\subsection{Modified translation for \systemx\ (\cref{sec:representationB}, \cref{th:safety preservation})}

We prove \cref{th:safety preservation} by giving a translation from 
$\TRC^{\neg \exists \wedge \vee}$ to $\systemx$.
It is a simple adaptation of the translation from \cref{sec:translation from TRC} with only two modifications:

(1) Instead of built-in relations, we use the visual shortcuts from \cref{sec:anchor simplifications}.
We thus only need to transform $\TRC$ into $\TRC^{\neg \exists \wedge \vee}$, which preserves the safety conditions.
We can also represent selection and join predicates directly via their shortcuts and do not need to go the full route via built-in relations.

(2) The second step of ``Creating canvas partitions'' now also includes disjunction boxes in addition to negation scopes: 
each subtree rooted below a negation $\neg$ gets a negation scope, 
and each subtree rooted below a disjunction node $\vee$ gets a disjunction box.
Sibling boxes are then either connected with a dotted line, or adjacent to each other.

In the reverse direction, we need to modify the first step of ``Determine the negation hierarchy'':
From the nested canvas negation and disjunction partitions, 
we now create the nested scopes of the negation and disjunction operators of the later $\TRC^{\neg \exists \wedge \vee}$ expression
For \cref{Fig_disjunction_bigger_RDplus_3}, this interpretation leads to 
$\varphi = \neg(\neg( \_ \vee \_ ) ) $.

\subsection{$\systemx$ is exponentially more succinct than $\diagrams$ (\cref{sec:size})}
\label{sec:exponential succinctness}

We give here the proof for \cref{prop:exponentialsize} in \cref{sec:size} that 
$\systemx$ is an exponentially more succinct representation than \diagrams.
We proceed in two steps:
(1) We first show that every \diagram\xspace can be immediately interpreted as $\systemx$ 
(with the exception of ``union cells'' (or union boxes), where $\systemx$ replaces multiple outputs relation with only one).
(2) We then give a parameterized query that has an exponentially bigger size than \diagram.

\subsubsection{$\systemx$ has equal or smaller size}
\label{sec:systemx equal or smaller}
From \cref{sec:real solution} and our translation from \cref{sec:translation back to TRC}, 
it follows that every $\diagram$ without union boxes is immediately also a $\systemx$.
For every $\diagram$ with union boxes, all the output relations in the union boxes can be replaced with one single output relation outside the disjunction boxes to give a valid $\systemx$. It follows that for all $\diagrams$ there is a $\systemx$ of equal or smaller size.
We give an example.

\begin{example}[\cref{ex:disjunction_union} continued]
    \label{ex:disjunction_union_cont}
    $\diagrams$ need to transform query \cref{eq:union query} into a union of disjunction-free queries
    \begin{align*}
        \h{\{}q(A) \mid \exists r\in R[q.A = r.A]\h{\}} 
        \;\h{\cup}\; 
        \h{\{}q(A) \mid \exists s\in S[q.A = s.A]\h{\}} 
    \end{align*}
    As a consequence, the output table will be repeated twice, 
    each in a separate union cell and with an identical attribute signature,
    just like Datalog (see \cref{Fig_AST_Union_2}).
    
    In contrast, $\systemx$ has maximally one output relation, just like $\TRC$. 
    Therefore, \cref{Fig_AST_Union_2} is not a valid $\systemx$ representation.
    Instead, $\systemx$ represents
    query \cref{eq:union query} with one single output table 
    (see \cref{Fig_AST_Union_1}).

    \Cref{Fig_AST_Union_3,Fig_AST_Union_4} show their ASTs. 
    Notice how the \diagram\xspace has the same atoms, except for repeated output tables.
\end{example}

\begin{figure}[t]
    \centering	
    \begin{subfigure}[b]{.41\linewidth}
        \centering
        \includegraphics[scale=0.42]{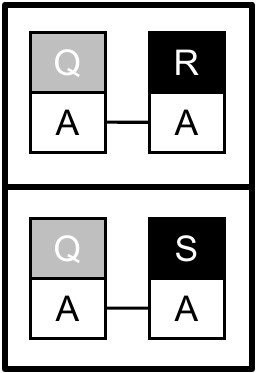}
    \caption{}
    \label{Fig_AST_Union_2}
    \end{subfigure}	
    \hspace{1mm}
    \begin{subfigure}[b]{.36\linewidth}
        \centering
        \includegraphics[scale=0.42]{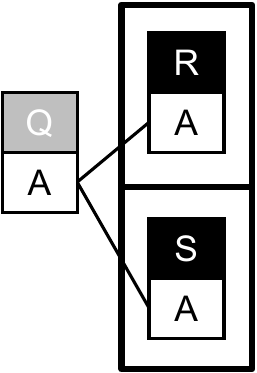}
    \caption{}
    \label{Fig_AST_Union_1}
    \end{subfigure}	
    \hspace{1mm}
    \begin{subfigure}[b]{.41\linewidth}
        \centering
        \includegraphics[scale=0.35]{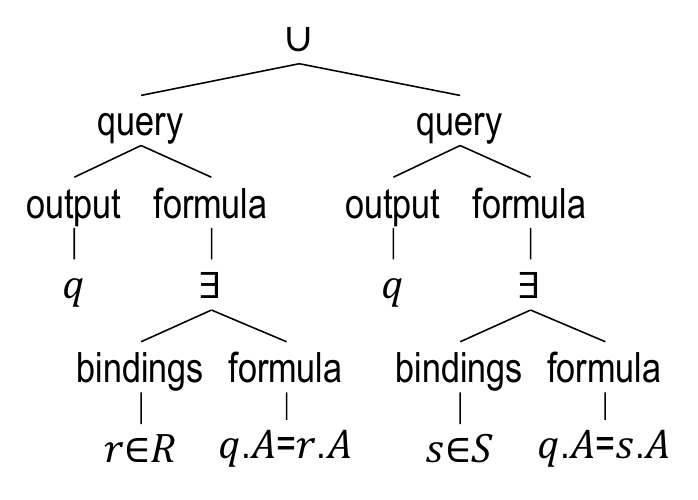}
    \caption{}
    \label{Fig_AST_Union_4}
    \end{subfigure}	
    \hspace{1mm}
    \begin{subfigure}[b]{.36\linewidth}
        \centering
        \includegraphics[scale=0.35]{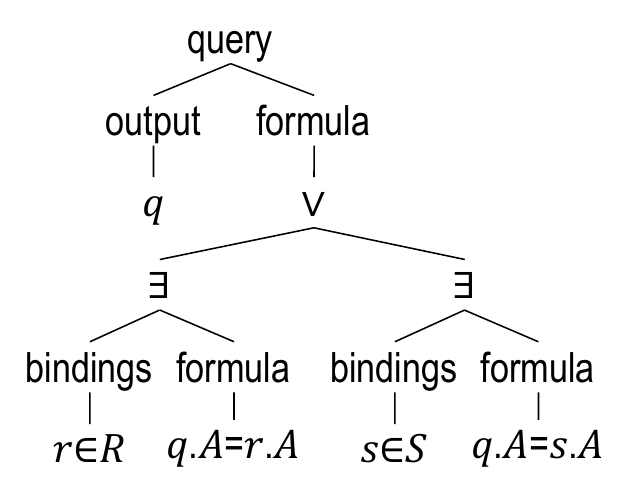}
    \caption{}
    \label{Fig_AST_Union_3}
    \end{subfigure}	
    \hspace{1mm}
    \caption{\Cref{ex:disjunction_union_cont}: Every $\diagram$ can be represented as $\systemx$ in equal or smaller size.}
    \label{Fig_AST_Union}
\end{figure}

\subsubsection{$\systemx$ can be exponentially smaller than $\diagrams$}
\label{sec:proof exponentially smaller}

We now give the proof for \cref{prop:exponentialsize} from \cref{sec:disjunctions}.

\begin{proof}[Proof \cref{prop:exponentialsize}]
Consider the following parameterized Boolean $\TRC$ query:
\begin{align*}
    \exists r \in R [(r.A_1 \equal c_{11} \vee r.A_1 \equal c_{12}) \wedge \ldots \wedge (r.A_k \equal c_{k1} \vee r.A_k \equal c_{k2})]
\end{align*}
It can be represented as $\systemx$ with size $\O(k)$ in the number of symbols.
Concretely, $\systemx$ needs $5k+1$ boxes and $2k$ edges (see \cref{Fig_exponential_example_1}).

$\diagrams$ can represent the disjunction as DNF and thus first need to transform the query into DNF:
\begin{align*}
    &\exists r \in R [r.A_1 \equal c_{11} \wedge \ldots \wedge r.A_k \equal c_{k1}] \,\vee \\
    &\exists r \in R [r.A_1 \equal c_{11} \wedge \ldots \wedge r.A_k \equal c_{k2}] \,\vee \\    
    &\ldots \\
    &\exists r \in R [r.A_1 \equal c_{12} \wedge \ldots \wedge r.A_k \equal c_{k2}] 
\end{align*}
The DNF has $2^k$ clauses and its size and also the size of the $diagrams$ representation is thus $\O(2^k)$.
Concretely, $\diagrams$ requires $(k+1)\cdot 2^{k}$ boxes (see \cref{Fig_exponential_example_2} for the use of union cells).

Since every $\diagram$ is also a $\systemx$ representation (except the union cells from \cref{sec:systemx equal or smaller}), 
it follows that $\systemx$ is an exponentially smaller representation than $\diagrams$.

To see that $\TRC$ and $\systemx$ have the same asymptotic size, notice that the translation from $\TRC$ 
succeeds by first translating $\TRC$ into $\TRC^{\neg \exists \wedge \vee}$ 
(which does not change the atoms, and hence the leaves)
and then directly mapping the atoms to visual symbols:
(1) binding atoms get mapped to relational tables;
(2) join predicates get mapped to two attributes and a connecting line with label appropriate placed;
(3) selection predicates get mapped to either one or two attribute boxes.
The negation scopes and disjunction boxes are in 1-to-1 correspondences with the levels of negation and the appearance of disjunctions in the AST.
\end{proof}

\begin{figure}[t]
    \centering	
    \begin{subfigure}[b]{.3\linewidth}
        \centering
        \includegraphics[scale=0.42]{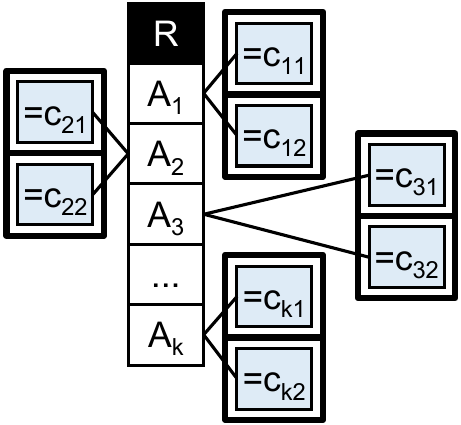}
    \caption{$\O(k)$ size.}
    \label{Fig_exponential_example_1}
    \end{subfigure}	
    \hspace{2mm}
    \begin{subfigure}[b]{.4\linewidth}
        \centering
        \includegraphics[scale=0.42]{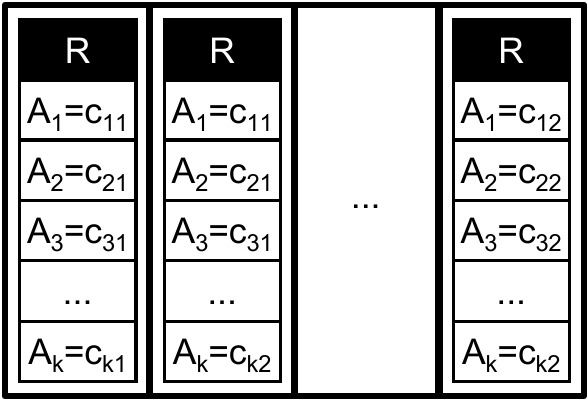}
        \vspace{1mm}
    \caption{$\O(2^k)$ size.}
    \label{Fig_exponential_example_2}
    \end{subfigure}	
    \caption{\Cref{sec:proof exponentially smaller}: Illustration for the proof of \cref{prop:exponentialsize}.}
\end{figure}

\section{Details for (\cref{sec:generality})}

\subsection{An example for recovering edge-based disjunctions (\cref{sec:generalizing prior solutions})}
\label{app:recover edge-based disjunctions}

\begin{example}[\cref{ex:disjunction} continued]
\label{ex:disjunctions as edges}
We next show how the disjunction
$\exists r \in R [r.A \equal 1 \;\h{\vee}\; r.A \equal 2]$
can be visualized in the spirit of edge-based disjunctions (\cref{Fig_Prior_Disjunctions_5}).
Notice how constants are semantically not bound to particular negation scopes and the query can 
be expressed with built-in predicates as
\begin{align*}
    & \exists r \in R, 
    \exists c_1 \in \textup{``$=\!1$''}, 
    \exists c_2 \in \textup{``$=\!2$''} 
    [\\
    &\hspace{2.0mm} \neg (\neg( \exists e_1 \in \textup{``$=$''} [r.A \equal e_1.\$1 \wedge e_1.\$2 \equal c_1.\$1]) \,\wedge  \\
    &\hspace{5.5mm} \neg( \exists e_2 \in \textrm{``$=$''}  [r.A \equal e_2.\$1 \wedge e_2.\$2 \equal c_2.\$1])
     )] 
\end{align*}    
\Cref{Fig_Disjunctions_5} shows the atom-preserving representation as \diagrams\xspace with built-in relations,
and \Cref{Fig_Disjunctions_6} the corresponding $\systemx$.
\Cref{Fig_Disjunctions_7,Fig_Disjunctions_8} shows the process of making the anchors of the equijoins increasingly small,
leading in the limit to the edge-based disjunction from \cref{Fig_Prior_Disjunctions_5}.
\end{example}

\begin{figure}[t]
    \centering	
    \begin{subfigure}[b]{.25\linewidth}
        \includegraphics[scale=0.4]{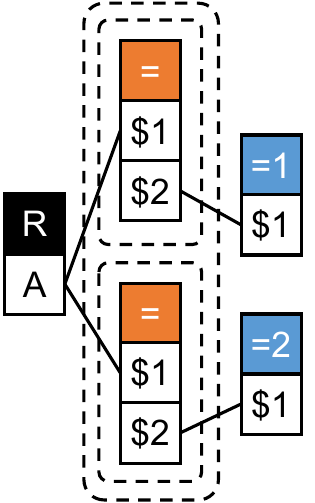}
    \caption{}
    \label{Fig_Disjunctions_5}
    \end{subfigure}	
    \hspace{2.7mm}
    \begin{subfigure}[b]{.2\linewidth}
        \includegraphics[scale=0.4]{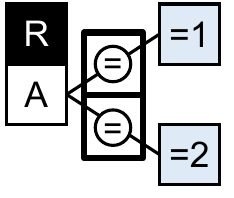}
    \vspace{8mm}
    \caption{}
    \label{Fig_Disjunctions_6}
    \end{subfigure}	
    \hspace{2.7mm}
    \begin{subfigure}[b]{.2\linewidth}
        \includegraphics[scale=0.4]{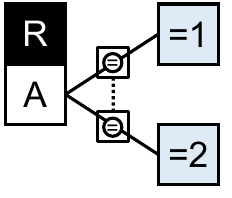}
    \vspace{8mm}
    \caption{}
    \label{Fig_Disjunctions_7}
    \end{subfigure}	
    \hspace{2.7mm}
    \begin{subfigure}[b]{.2\linewidth}
        \includegraphics[scale=0.4]{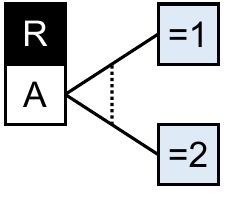}
    \vspace{8mm}
    \caption{}
    \label{Fig_Disjunctions_8}
    \end{subfigure}	
    \caption{\Cref{ex:disjunctions as edges} in \cref{app:recover edge-based disjunctions}: 
    $\systemx$ can recover edge-based disjunctions (\cref{Fig_Prior_Disjunctions_5})
    via increasingly small anchors on the equijoins.}
    \label{Fig_Disjunctions}
\end{figure}

\subsection{Perceptual justification of our solution and Peirce shading (\cref{sec:perceptual choices})}
\label{sec:perceptual justifications}

Our main contributions in this paper is a representation system that solves the disjunction and the safety problem
with its constituitive features. 
We discuss here some of the enabling features of our representation which, 
recall from \cref{sec:enabling feature}, are those are features that facilitate interpretation but are not essential.
Those may thus appear as orthogonal to some readers and we list them only as supplemental information.
But they informed us in our design.

Chamberlin in his recent CACM article on SQL~\cite{DBLP:journals/cacm/Chamberlin24} states 
``\emph{Our specific goals were to design a query language with the following properties:
...
user with no specialized training could, in simple cases, understand the meaning of a query simply by reading it. We called this the `walk-up-and-read' property}.''
Diehl~\cite{Diehl:2007ly} writes:
``\emph{If done right, diagrams group relevant information together to make searching more eﬃcient, and use visual cues to make information more explicit.}''

Similarly, our goal was to develop an ($i$) intuitive and ($ii$) principled diagrammatic representation for ($iii$) arbitrary nestings of disjunctions, ($iv$) with minimal additional notations.
Our solution $\systemx$ 
is heavily inspired by the design choices of $\diagrams$
and meets the challenge without introducing any fundamentally novel visual symbol to the visual grammar of $\diagrams$.
We merely give a stricter syntactic interpretation of edge labels, allow constants outside table attributes, and allow disjunction boxes used inside the diagrams.

\Cref{Fig_graphical_codes} shows how we see our conservative extension with disjunctions fit into the overall
visual grammar and semantic patterns (sometimes called ``codes'')
of node-link diagrams and relationship representations~\cite{ware:2021:visual_thinking}.
As Ware writes, ``\emph{a good diagram takes advantage of basic perceptual mechanisms that have evolved to perceive structure in the environment}''~\cite{ware:2021:visual_thinking} and we tried to follow that practice.
Notice how our choices of disjunctions inherit the ``nested regions'' code from negation scopes and DeMorgan,
while also using the ``shapes in proximity'' code (widely used and known from UML) and alternatively the ``shapes connected by contour'' code.

\begin{figure}[t]
    \centering
    \includegraphics[scale=0.39]{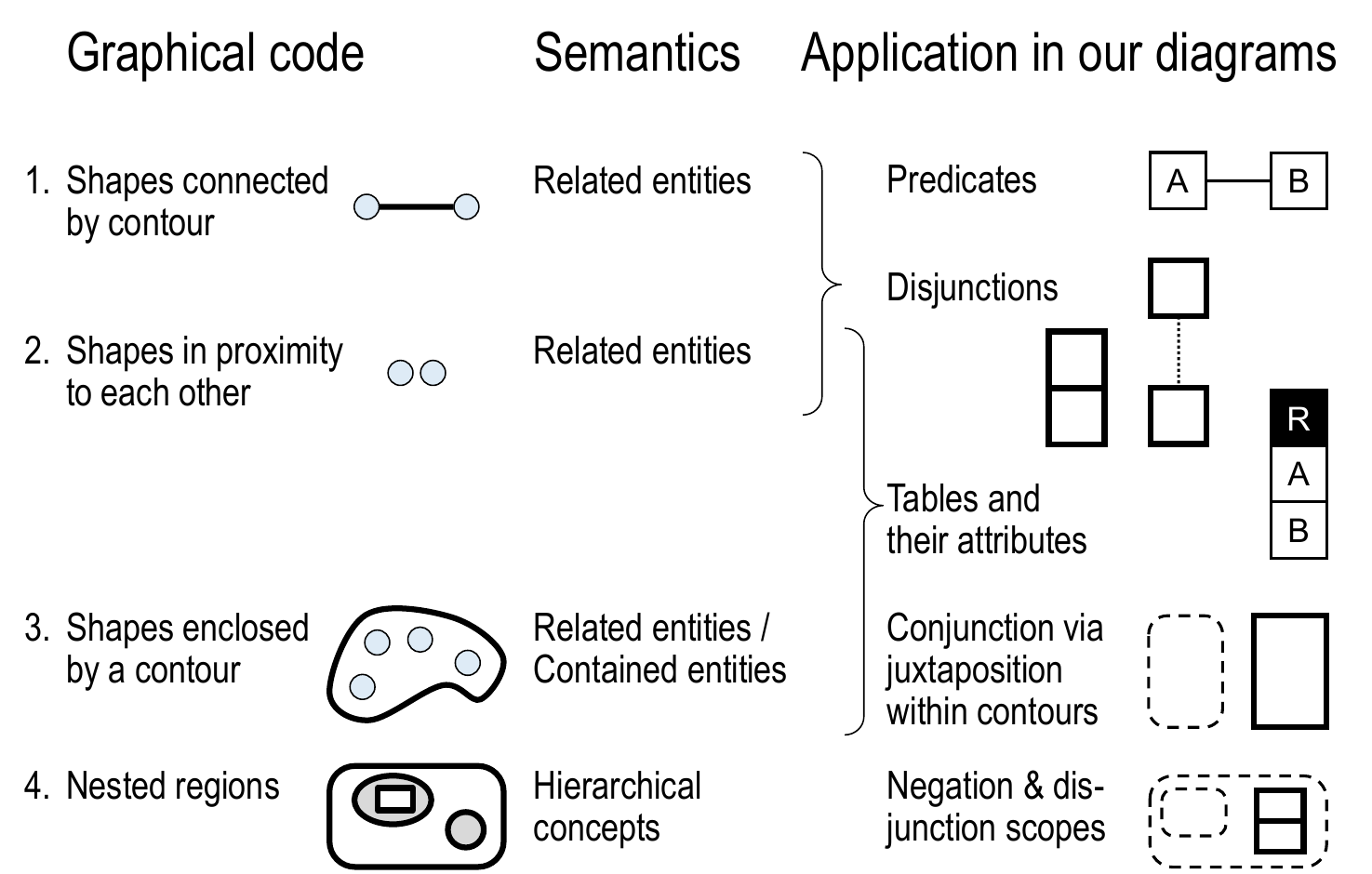}
    \caption{
    \Cref{sec:perceptual justifications}:
    Our design choices in the larger context of the visual grammar of node-link diagrams and relationship representations.
    Figure is heavily inspired from~\cite{ware:2021:visual_thinking} but modified to our context.
    }
    \label{Fig_graphical_codes}
\end{figure}

\begin{figure}[t]
    \centering
    \begin{subfigure}[b]{.3\linewidth}
        \includegraphics[scale=0.4]{figs/Fig_my_safety_example_Diagram2}
        \caption{}
        \label{Fig_my_safety_example_Diagram2_appendix}
    \end{subfigure}	
    \hspace{10mm}
    \begin{subfigure}[b]{.3\linewidth}
        \includegraphics[scale=0.4]{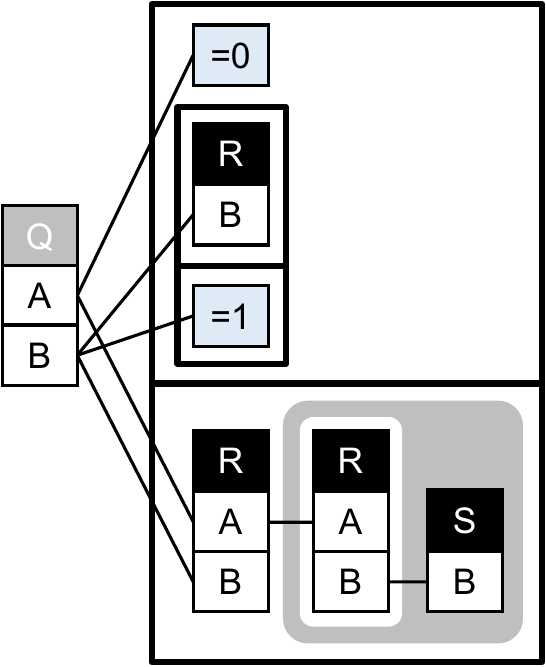}
        \caption{}
        \label{Fig_my_safety_example_Diagram}
    \end{subfigure}	
    \caption{(a): \Cref{ex:safety_continued} (\cref{ex:safety} continued).
    (b) shows Peirce shading applied.}
\end{figure}

\subsubsection{Peirce shading to facilitate reading}
\label{sec:multiple readings}
We give here an additional enabling feature for $\systemx$
that was originally proposed by Peirce~\cite[Paragraph~4.621]{peirce:1933}
and became more widely known due to Sowa~\cite{SOWA2008213,Sowa+2011+347+394}.
Call zones in the diagram as either \emph{positive} 
(if nested in an even number of negation scopes, possibly zero) 
or \emph{negative} (if nested in an odd number of negations).
Then to improve contrast and facilitate reading, 
shade negative areas in gray
and keep positive areas in white.
For example, \cref{Fig_my_safety_example_Diagram} shows Perice shading applied to \cref{ex:safety_continued}.

Since the semantics of DeMorgan-fuse boxes is based on double negation, 
and they thus \emph{do not change the parity of a zone},
Peirce shadings very naturally complement with disjunction (see \cref{Fig_anquors}).
We illustrate by continuing \cref{ex:Fig_disjunction_bigger_AST}.

\begin{figure}[t]
    \centering
    \includegraphics[scale=0.42]{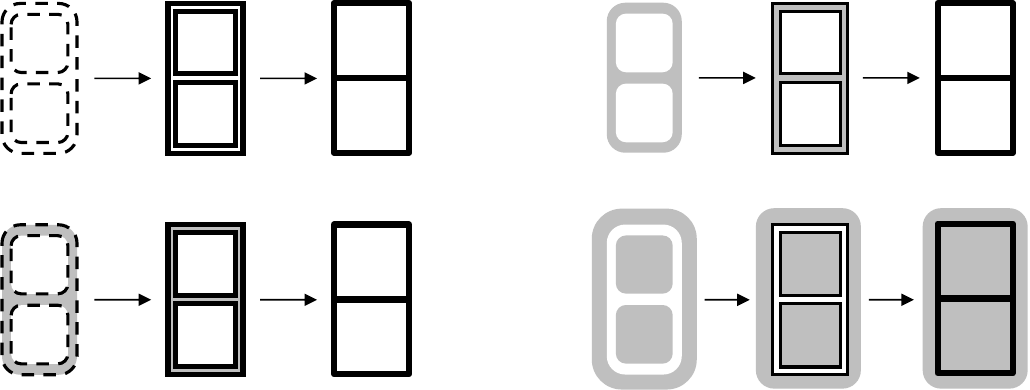}
    \caption{\cref{sec:multiple readings}: DeMorgane-fuse boxes
    do not change the parity of a zone and thus effortlessly interact with Peirce shadings.}
    \label{Fig_anquors}
\end{figure}

\begin{example}[\cref{ex:Fig_disjunction_bigger_AST} continued]
    \label{ex:Fig_disjunction_bigger_AST_continued}
    \Cref{Fig_disjunction_bigger_builtin} showed 
    \cref{ex:Fig_disjunction_bigger_AST}
    as $\diagram$ with built-in relations. 
    \Cref{Fig_disjunction_bigger_RDplus_2} now shows the $\systemx$ representation by 
    translating from query \cref{eq:Fig_disjunction_bigger_AST:3} in $\TRC^{\neg \exists \wedge}$.
    In contrast, \cref{Fig_disjunction_bigger_RDplus_3} shows $\systemx$ by 
    translating directly from the $\TRC^{\neg \exists \wedge \vee}$ query \cref{eq:Fig_disjunction_bigger_AST:2}
    and using disjunctions.
    Notice how the disjunction boxes do not change the parity of the zones in which they are 
    and can really be interpreted as visual shortcuts (thus providing a post hoc justification for our name ``DeMorgan-fuse'').
\end{example}

\begin{figure}[t]
    \centering
    \begin{subfigure}[b]{0.47\linewidth}
        \includegraphics[scale=0.42]{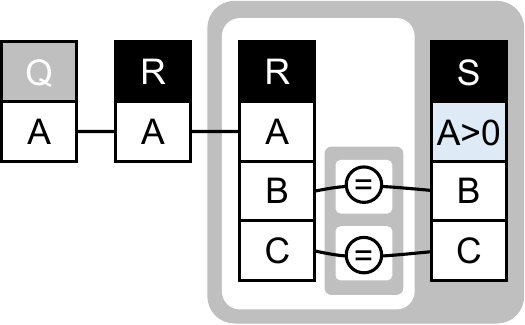}
        \caption{}
        \label{Fig_disjunction_bigger_RDplus_2}
    \end{subfigure}	
    \hspace{1mm}
    \begin{subfigure}[b]{0.47\linewidth}
        \includegraphics[scale=0.42]{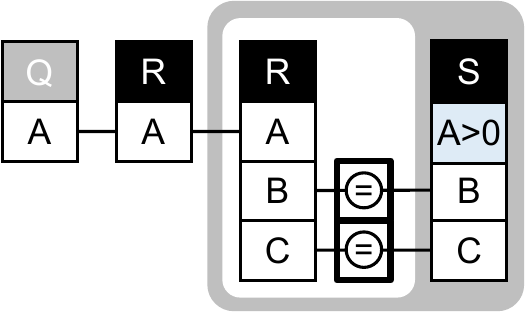}
        \caption{}
        \label{Fig_disjunction_bigger_RDplus_3}
    \end{subfigure}	
    \hspace{1mm}
    \caption{\cref{ex:Fig_disjunction_bigger_AST_continued}:
    $\systemx$ by translating from $\TRC^{\neg \exists \wedge}$ (a)
    or directly from $\TRC^{\neg \exists \wedge \vee}$ with DeMorgan-fuse boxes (b), both with Peirce shading applied.
    }
    \label{Fig_disjunction_bigger_RDplus}
\end{figure}

\subsection{$\systemx$ representations for 4 prior challenges (\cref{sec:prior challenges}, \cref{sec:introduction})}
\label{sec:RDconclusions}

We give here solutions in $\systemx$ to 4 challenges listed in the introduction. 
We use the word ``query'' also for a statement (i.e.\ a Boolean query).

\Cref{Fig_Thalheim_solution} shows $\systemx$ for \cref{Fig_Prior_Disjunctions_6},
a simplified version of a problem listed in  
two presentations~\cite{Thalheim:2007:Slides,Thalheim:2013:Slides}
by Thalheim (the original images are shown in \cref{Fig_original_disjunctions_Thalheim2,Fig_original_disjunctions_Thalheim1}). The representation reads as:
\begin{align}
    R.A \equal 1 \vee R.B \equal 2 \vee (R.C \equal 3 \wedge R.D \equal 4) \label{eq:Thalheim}
\end{align}

\Cref{Fig_Disjunction_FutureWork_1,Fig_Disjunction_FutureWork_2} shows $\systemx$
for the two challenges listed in~\cite{DBLP:journals/pacmmod/GatterbauerD24}:
\begin{align}
    & \exists r\in R, \exists s \in S[r.A \!<\! s.E \wedge (r.B \!<\! s.F \vee r.C \!<\! s.G)]                       \label{eq:RD_dis_1}\\
    & \exists r\in R[(r.A \!>\! 0 \wedge r.A \!<\! 10) \vee (r.A \!>\! 20 \wedge r.A \!<\! 30)]     \label{eq:RD_dis_2}
\end{align}

\begin{figure}[t]
\centering
\begin{subfigure}[b]{0.19\linewidth}
    \centering
    \includegraphics[scale=0.42]{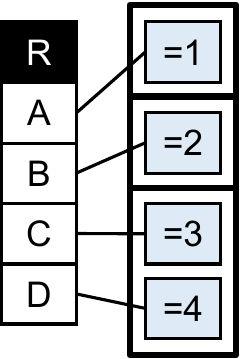}
    \caption{Query \cref{eq:Thalheim}, \cref{Fig_Prior_Disjunctions_6}}
    \label{Fig_Thalheim_solution}
\end{subfigure}	
\hspace{8mm}
\begin{subfigure}[b]{0.17\linewidth}
    \centering
    \includegraphics[scale=0.42]{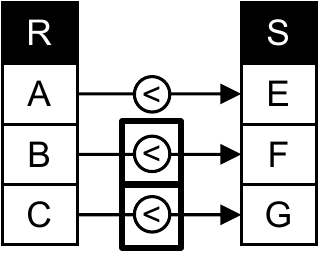}
    \vspace{1.3mm}
    \caption{Query \cref{eq:RD_dis_1}}
    \label{Fig_Disjunction_FutureWork_1}
\end{subfigure}	
\hspace{8mm}
\begin{subfigure}[b]{0.14\linewidth}
    \centering
    \includegraphics[scale=0.42]{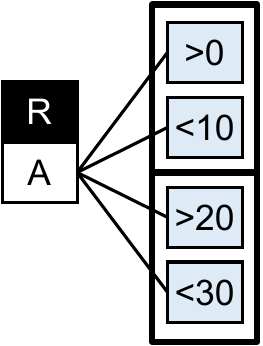}
    \vspace{0.6mm}    
    \caption{Query \cref{eq:RD_dis_2}}
    \label{Fig_Disjunction_FutureWork_2}
\end{subfigure}	
\hspace{8mm}
\begin{subfigure}[b]{0.2\linewidth}
    \centering
    \includegraphics[scale=0.35]{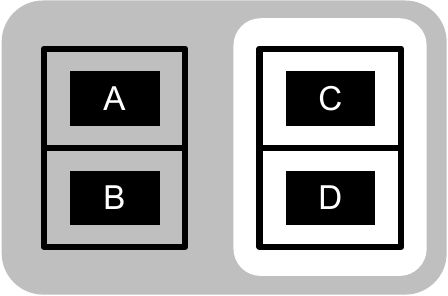}
    \caption{Query \cref{eq:Gardner}, \cref{Fig_Gardner_VENN}}
    \label{Fig_Gardner_VENN2}
\end{subfigure}	
\begin{subfigure}[b]{0.3\linewidth}
    \centering
    \includegraphics[scale=0.42]{figs/Fig_Disjunction_solutions_1}
    \vspace{0mm}    
    \caption{\cref{Fig_Prior_Disjunctions_10}, interpretation \cref{eq:Fig_Disjunction_solutions_1}}
    \label{Fig_Disjunction_solutions_1_appendix}
\end{subfigure}	
\hspace{1mm}
\begin{subfigure}[b]{0.3\linewidth}
    \centering
    \includegraphics[scale=0.42]{figs/Fig_Disjunction_solutions_2}
    \vspace{0mm}    
    \caption{\cref{Fig_Prior_Disjunctions_10}, interpretation \cref{eq:Fig_Disjunction_solutions_2}}
    \label{Fig_Disjunction_solutions_2_appendix}
\end{subfigure}	
\caption{\cref{sec:RDconclusions}: $\systemx$ representations
for challenging examples raised in prior work.}
\label{Fig_Disjunction_FutureWork}
\end{figure}

Gardner in his 1958 book `Logic Machines and Diagrams'~\cite{Gardner:1958:logicMachines}
discusses a challenging disjunction 
and concludes that
``there seems to be no simple way in which the statement, as it stands, can be diagramed''~\cite[Section 4.3]{Gardner:1958:logicMachines}:
\begin{align}
    (A \vee B) \rightarrow (C \vee D) \label{eq:Gardner}
\end{align}    
He proposes to use what he calls a \emph{diagrammatic compound statement} shown in 
\cref{Fig_Gardner_VENN3}, which needs to repeat the individual predicates, i.e.\ it is thus not pattern-isomorphic to query \cref{eq:Gardner}.
\Cref{Fig_Gardner_VENN2} shows $\systemx$ for query \cref{eq:Gardner}, 
which follows from rewriting the statement as:
\begin{align*}
    \neg(    
    (A \vee B) 
    \wedge
    \neg (C \vee  D) 
    )
\end{align*}
Notice that while $\TRC$ is commonly not defined with nullary ($0$-ary) predicates,
atomic propositions like $A$ can be naturally interpreted as nullary relations $A()$ that can be set to true or false.
In other words, a symbol ``$A$'' is true iff ``$\exists a \in A$'', i.e.\  thus table $A$ is not empty.

\begin{figure}[t]
    \centering
    \begin{subfigure}[b]{0.3\linewidth}
        \centering
        \includegraphics[scale=0.35]{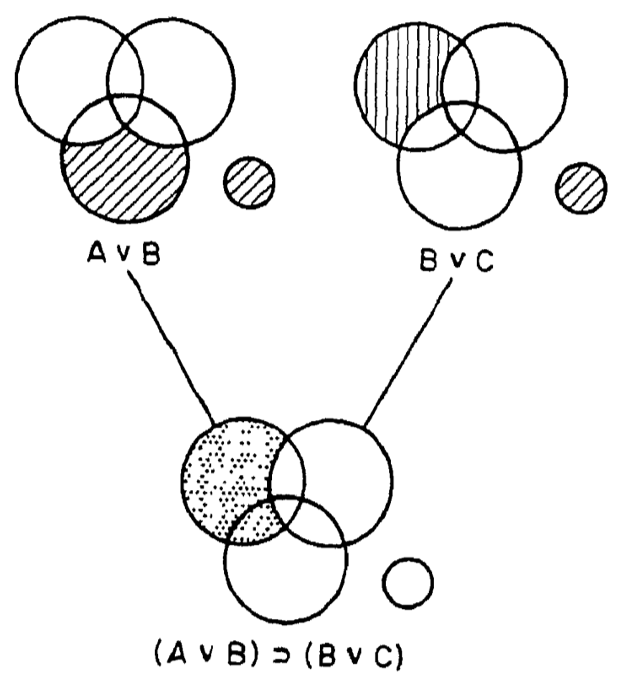}
        \caption{\cite[Fig.~44]{Gardner:1958:logicMachines}}
        \label{Fig_Gardner_VENN3}
    \end{subfigure}	
    \hspace{10mm}
    \begin{subfigure}[b]{0.4\linewidth}
        \centering
        \includegraphics[scale=0.38]{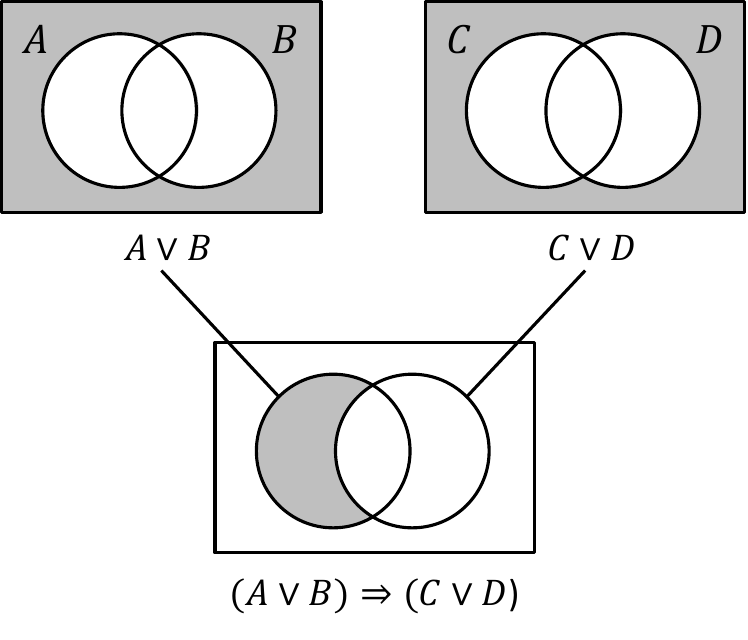}
        \caption{}
        \label{Fig_Gardner_VENN1}
    \end{subfigure}	
    \caption{Garner's challenging disjunction 
    $(A \vee B) \rightarrow (C \vee D)$ shown in its original form (a) and an equivalent more modern use of Venn diagrams.}
    \label{Fig_Gardner_VENN}
\end{figure}

Most interesting is \cref{Fig_Prior_Disjunctions_10}, raised as a challenge in the earlier mentioned tutorial~\cite{ICDE:2024:diagrammatic:tutorial}.
Here we point out an issue that -- to the best of our knowledge -- has never been raised in the literature on visual representation and diagrammatic reasoning: 
There are two different ways to interpret the figure shown in \cref{Fig_Prior_Disjunctions_10}:
\begin{align}
    & \exists r\in R, \exists s \in S[(r.A \equal s.A \wedge r.B=0) \vee (r.B \equal 1 \wedge r.C \equal 2)] 
        \label{eq:Fig_Disjunction_solutions_1}\\
    & \exists r\in R [(\exists s \in S[r.A \equal s.A] \wedge r.B=0) \vee (r.B \equal 1 \wedge r.C \equal 2)] 
        \label{eq:Fig_Disjunction_solutions_2}
\end{align}    

\noindent
The difference is that 
query \cref{eq:Fig_Disjunction_solutions_1} 
is true on the database $\{R(0,1,2)\}$, 
whereas query \cref{eq:Fig_Disjunction_solutions_2} is false.
The reason is the different scoping of $S$.
Our principled translation into $\systemx$ with DeMorgan-fuse boxes creates the two \emph{distinct} diagrams
\cref{Fig_Disjunction_solutions_1_appendix,Fig_Disjunction_solutions_2_appendix} and can thus handle the distinction, as expected.
We do not know any edge-based approach that could handle and distinguish the cases.

\section{Details on textbook benchmark (\cref{sec:pattern coverage})}
\label{sec:pattern coverage appendix}

We list here the 3 queries that $\diagrams$ cannot pattern-represent:
\circled{1}~``Find the names of sailors who have reserved a red or a green 
boat''~\cite[Query 5, Sect.~5.2]{cowbook:2002}.
\circled{2}~``Make a list of project numbers for projects that involve an employee whose last name is `Smith', either as a worker or as manager of the controlling department for the 
project''~\cite[Query 4, Ch.~8.6]{Elmasri:dq}.
\circled{3}~``Get part numbers for parts that either weigh more than 16 pounds or are supplied by supplier S2, or 
both''~\cite[Query 9, Ch.8.3]{date2004introduction}.
\Cref{Fig_Textbook} shows those queries in their original definitions
together with their $\systemx$ representations.

\begin{figure}[h]
\centering
\begin{subfigure}[b]{0.58\linewidth}
    \includegraphics[scale=0.36]{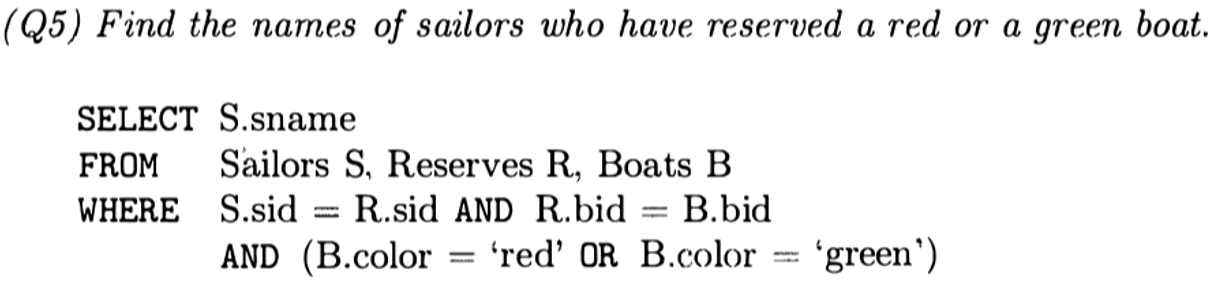}
    \caption{\cite[Query 5, Sect.~5.2]{cowbook:2002}}
    \label{Fig_CowBook_Query}
\end{subfigure}	
\hspace{1mm}
\begin{subfigure}[b]{0.4\linewidth}
    \includegraphics[scale=0.35]{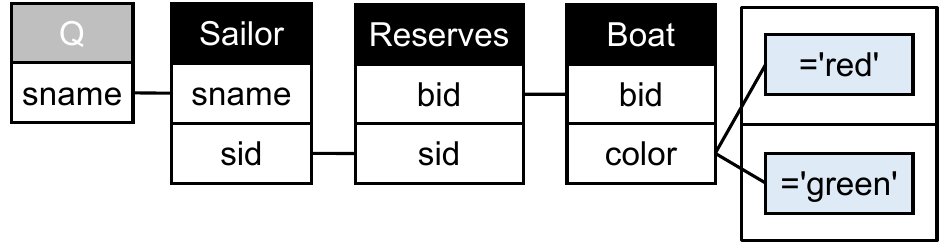}
    \caption{}
    \label{Fig_CowBook_RD}
\end{subfigure}	
\hspace{1mm}
\begin{subfigure}[b]{0.58\linewidth}
	\vspace{1mm}
    \includegraphics[scale=0.37]{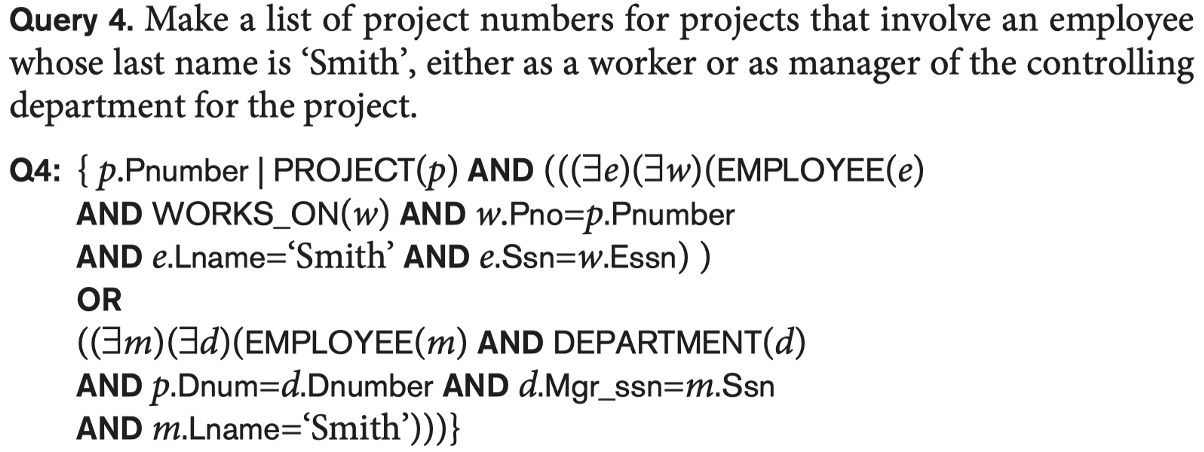}
    \caption{\cite[Query 4, Ch.~8.6]{Elmasri:dq}}
    \label{Fig_Elmasri_Query}
\end{subfigure}	
\hspace{1mm}
\begin{subfigure}[b]{0.4\linewidth}
    \includegraphics[scale=0.35]{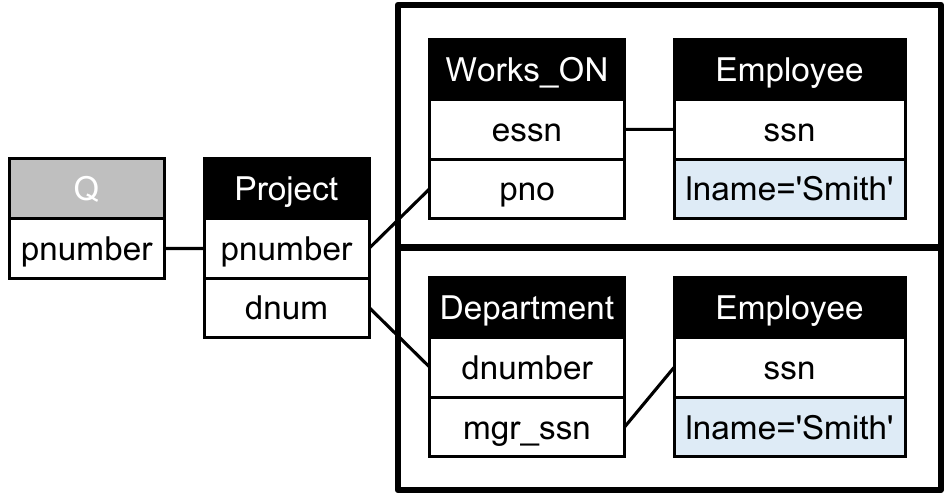}
    \caption{}
    \label{Fig_Elmasri_RD}
\end{subfigure}	
\hspace{1mm}
\begin{subfigure}[b]{0.68\linewidth}
    \includegraphics[scale=0.39]{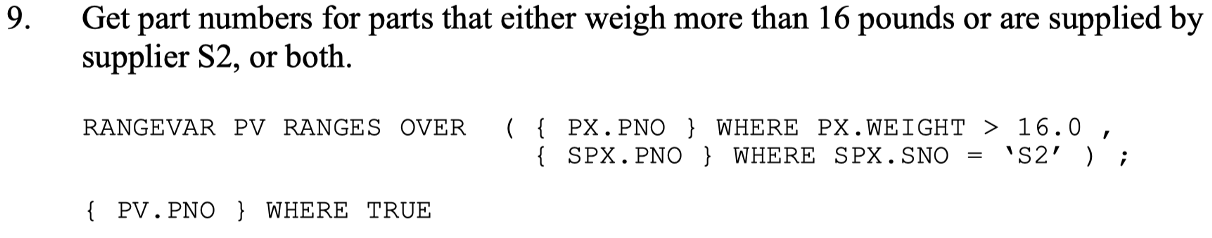}
    \caption{\cite[Query 9, Ch.8.3]{date2004introduction}}
    \label{Fig_Date_Query}
\end{subfigure}	
\hspace{1mm}
\begin{subfigure}[b]{0.3\linewidth}
    \includegraphics[scale=0.35]{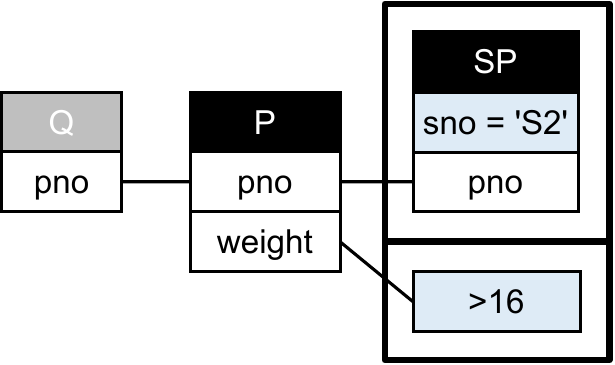}
    \caption{}
    \label{Fig_Date_RD}
\end{subfigure}	
\hspace{1mm}
\caption{\Cref{sec:pattern coverage appendix}: $\systemx$ representations for 3 relational query patterns 
that cannot be represented by prior approaches.}
\label{Fig_Textbook}
\end{figure}

\clearpage

\end{document}